\documentclass[a4paper,reqno,11pt]{article}
\usepackage[hmargin=2.2cm,vmargin=2.2cm]{geometry}
\usepackage{amsmath,amssymb,amsthm,mathtools,cancel,mathrsfs,mathdots,framed,color,graphicx,enumitem,enumerate,caption,subcaption}
\usepackage[colorlinks=true, pdfstartview=FitV, urlcolor=blue, citecolor=red, linkcolor=blue]{hyperref}

\usepackage[utf8]{inputenc}
\usepackage[T1]{fontenc}

\def\Re{\mathrm {Re}\,}
\def\Im{\mathrm {Im}\,}
\def\C{\mathbb{C}}
\def\R{\mathbb{R}}
\def\wt{\widetilde}
\def\wh{\widehat}
\def\G{\Gamma}
\def\1{\mathbf{1}}
\def\d{\mathrm d}
\def\e{\mathrm{e}}
\def\i{\mathrm{i}}
\def\pa{\partial}
\newcommand{\tr}[1][]{\mathrm{tr}_{#1}\,}
\def\Ai{\mathrm{Ai}}
\def\s{\sigma}
\def\res{\mathop{\mathrm {res}}\limits_}
\def\ll{\lambda}

\DeclareMathOperator*{\esssup}{ess\,sup}

\def\be{\begin{equation}}
\def\ee{\end{equation}}

\newtheorem{theorem}{Theorem}[section]
\newtheorem{lemma}[theorem]{Lemma}
\newtheorem{proposition}[theorem]{Proposition} 
\newtheorem{corollary}[theorem]{Corollary}
\newtheorem{remark}[theorem]{Remark}
\newtheorem{example}[theorem]{Example}

\newtheorem{assumption}{Assumption}

\newtheorem{theoremintro}{Theorem}

\usepackage{tikz}
\usetikzlibrary{arrows}
\usetikzlibrary{decorations.pathmorphing}
\usetikzlibrary{decorations.markings}
\usetikzlibrary{patterns}
\usetikzlibrary{automata}
\usetikzlibrary{positioning}
\usepackage{tikz-cd}
\tikzset{->-/.style={decoration={
				markings,
				mark=at position #1 with {\arrow{latex}}},postaction={decorate}}}
	
	\tikzset{-<-/.style={decoration={
				markings,
				mark=at position #1 with {\arrowreversed{latex}}},postaction={decorate}}}

\usetikzlibrary{shapes.misc}\tikzset{cross/.style={cross out, draw, 
         minimum size=2*(#1-\pgflinewidth), 
         inner sep=0pt, outer sep=0pt}}

\begin{document}

\numberwithin{equation}{section}

\title{J\'anossy densities and Darboux transformations \\
for the Stark and cylindrical KdV equations}
\author{Tom Claeys$^1$, Gabriel Glesner$^1$, Giulio Ruzza$^2$, Sofia Tarricone$^3$}

\date{}
\maketitle

\begin{center}
{$^1$\footnotesize
{\textit{Institut de Recherche en Math\'ematique et Physique, Universit\'e catholique de Louvain,
\\
Chemin du Cyclotron 2, 1348 Louvain-la-Neuve, Belgium}
\\
\texttt{tom.claeys@uclouvain.be,\,gabriel.glesner@uclouvain.be}}
\\
\medskip
$^2$\footnotesize
{\textit{Departamento de Matem\'atica, Faculdade de Ci\^encias da Universidade de Lisboa,
\\
Campo Grande Edif\'{i}cio C6, 1749-016, Lisboa, Portugal}
\\
\texttt{gruzza@fc.ul.pt}}
\\
\medskip
$^3$\footnotesize
{\textit{Institut de Physique Th\'eorique, Universit\'e Paris-Saclay, CEA, CNRS, F-91191 Gif-sur-Yvette, France}
\\
\texttt{sofia.tarricone@ipht.fr}}
}
\end{center}

\medskip

\begin{abstract}
We study J\'anossy densities of a randomly thinned Airy kernel determinantal point process.
We prove that they can be expressed in terms of solutions to the Stark and cylindrical Korteweg--de\,Vries equations; these solutions are Darboux tranformations of the simpler ones related to the gap probability of the same thinned Airy { point} process.
Moreover, we prove that the associated wave functions satisfy a variation of Amir--Corwin--Quastel's integro-differential Painlev\'e~II equation. 
Finally, we derive tail asymptotics for the relevant solutions to the cylindrical Korteweg--de\,Vries equation and show that they decompose asymptotically into a superposition of simpler solutions.
\end{abstract}

\medskip
\medskip



\medskip

\section{Introduction}\label{sec:intro}

The cylindrical Korteweg--de\,Vries equation admits a family of solutions which are expressed in terms of Fredholm determinants involving the Airy kernel operator \cite{Poppe, PoppeSattinger, CafassoClaeysRuzza}.
These solutions have an interesting probabilistic interpretation, as they are gap probabilities for random thinnings of the Airy point process \cite{CafassoClaeysRuzza}, and they are therefore connected to important problems in integrable probability, such as the extreme value statistics of finite temperature free fermions \cite{Johansson, LiechtyWang, MNS, LDMRS}, the distribution of the narrow wedge solution of the Kardar--Parisi--Zhang equation \cite{KPZKPDoussal, QuastelRemenik, BPS}, the edge eigenvalue statistics in the complex elliptic Ginibre Ensemble at weak non-Hermiticity \cite{BothnerLittle}, and multiplicative statistics of Hermitian random matrices~\cite{GS}.

\medskip

In this work, we show that Darboux transformations of such solutions also enjoy a probabilistic interpretation: they are J\'anossy densities of random thinnings of the Airy point process. We investigate the integrable structure of these solutions, and show how they are connected to the Stark equation and to an integro-differential Painlev\'e II equation.
In this way, we reveal a remarkable connection between the Airy point process and scattering theory for a large class of solutions of the cylindrical Korteweg--de\,Vries equation, which describes an isospectral deformation of the Stark equation.
Moreover, we show that their tail asymptotics can be described as a superposition of simpler solutions. 
This soliton-like behavior finds its origin in the fact that certain conditional ensembles related to the Airy point process decorrelate asymptotically.

Whereas it is understood since several decades that specific solutions of integrable PDEs admit probabilistic interpretations in terms of gap probabilities of determinantal point processes, recent results indicate that large classes of solutions to integrable PDEs admit probabilistic interpretations as well. From a broad perspective, our results confirm this in the case of the cylindrical Korteweg--de\, Vries equation, and further extend the class of solutions with probabilistic interpretations, by considering J\'anossy densities instead of gap probabilities.
In this sense, we improve and extend known results relating the theory of integrable probability  and classical integrable systems in the case of the thinned Airy point process. To give a clear picture of these connections, we will start with a brief introduction to the Airy point process to define our quantities of interest, the gap probabilities and the J\'anossy densities, then we will give an overview of the main results already known for the gap probabilities, and finally we will state our new results concerning the J\'anossy densities. 

\subsubsection*{Gap probabilities and J\'anossy densities of the thinned Airy point process}
The Airy point process is the determinantal point process on $\R$, defined by its correlation functions
\be
\label{eq:corr}
\rho^\Ai_n(\lambda_1,\dots,\lambda_n):=\det\begin{pmatrix}
	K^\Ai(\lambda_i,\lambda_j)\end{pmatrix}_{1\leq i,j\leq n},
\ee
where $K^\Ai$ is  the {\it Airy kernel}
\be
\label{eq:standardairykernel}
K^\Ai(\ll,\mu):=\int_0^{+\infty}\Ai(\lambda+\eta)\Ai(\mu+\eta)\,\d\eta=\frac{\Ai(\lambda)\Ai'(\mu)-\Ai'(\lambda)\Ai(\mu)}{\lambda-\mu},
\ee
and $\Ai$ and $\Ai'$ are the Airy function and its derivative, respectively. 
Equivalently \cite{Soshnikov}, it is the probability distribution on the space of locally finite configurations of points on the real line characterized by its $n$-point correlation functions $\rho_{n}^\Ai$ defined in~\eqref{eq:corr}.
More explicitly, this means that for all disjoint Borel sets $B_1,\dots,B_\ell\subseteq\R$ and for all integers $k_1,\dots,k_\ell\geq 1$ summing up to $\sum_{j=1}^\ell k_j=n$, the expected number of $n$-tuples of points in a random configuration of which $k_1$ lie in $B_1$, $k_2$ in $B_2$, $\dots$, is
\be
\frac 1{k_1!\cdots k_\ell!}\int_{B_1^{k_1}\times\cdots\times B_l^{k_\ell}}\rho^\Ai_{n}(\lambda_1,\dots,\lambda_m)\,\d\lambda_1\cdots\d\lambda_m.
\ee
Random configurations of points in this process are known to contain almost surely an infinite number of points, and to have almost surely a largest one,  see Remark~\ref{remark:numberofparticles}.

The same characterization holds when introducing two parameters $X\in\R$, $T>0$ and considering the shifted and dilated Airy kernel 
\be
\label{eq:shifteddilatedAirykernel}
K^\Ai_{X,T}(\ll,\mu):=T^{-\frac 13}K^\Ai(T^{-\frac 13}(\ll+X),T^{-\frac 13}(\mu+X)\bigr).
\ee
and the corresponding shifted and dilated Airy point process with correlation functions
\be
\label{eq:Tcorr}
\rho_{X,T;n}^\Ai(\lambda_1,\dots,\lambda_n):=\det\begin{pmatrix}
	K^{\rm Ai}_{X,T}(\lambda_i,\lambda_j)\end{pmatrix}_{1\leq i,j\leq n}.
\ee

With this ground process depending on two parameters $X,T$ in mind, we can now define its $\s$-thinning, for a suitable class of functions $\s$.
\begin{assumption}
	\label{assumption:weak}
	The function $\s:\R\to[0,1]$ is smooth and there exists $\kappa>0$ such that $\s(\ll)=O\bigl(|\ll|^{-\frac 32-\kappa}\bigr)$ as~$\ll\to-\infty$.
\end{assumption}
The $\s$-thinned shifted and dilated Airy point process is obtained from the shifted and dilated Airy point process by removing each particle $\lambda$ in a configuration independently with (position-dependent) probability $1-\s(\lambda)$, thus retaining it with probability $\s(\lambda)$.
This point process is also determinantal, with a correlation kernel given by $\sqrt{\sigma(\lambda)}K^{\rm Ai}_{X,T}(\lambda,\mu)\sqrt{\sigma(\mu)}$, cf.~\cite[eq~(2.5)]{LavancierMollerRubak}.
Then a fundamental quantity in the study of the shifted and dilated Airy point process and its thinning is the Fredholm determinant
\be
\label{eq:TJ0}
J_\s(X,T):=1+\sum_{n=1}^\infty\frac{(-1)^n}{n!}\int_{\mathbb R^n}\rho_{X,T;n}^\Ai(\lambda_1,\dots,\lambda_n)\prod_{i=1}^n\sigma(\lambda_i)\d \lambda_i.
\ee
Indeed, $J_{\s}(X,T)$  enjoys the following two different probabilistic interpretations:
first, we have
\be
\label{eq:firstprobid}
J_\s(X,T)=\mathbb E_{X,T}\biggl[\prod_{j\geq 1}\bigl(1-\sigma(\lambda_j)\bigr)\biggr],
\ee
where the expectation on the right involves the particles $\lambda_1>\lambda_2>\cdots$ of the shifted and dilated Airy point process.
In other words, $J_\s(X,T)$ is an average multiplicative statistic of the ground process. Secondly,
\be
\label{eq:secondprobid}
J_\s(X,T)=\mathbb{P}_{X,T}^\s\left(\mathfrak X=\emptyset\right),
\ee
where the probability on the right involves the random configuration~$\mathfrak X$ of the $\sigma$-thinned ground process.
In other words, $J_\s(X,T)$ is the full gap probability of the $\sigma$-thinned ground process.
The first identity~\eqref{eq:firstprobid} follows from the general theory of determinantal point processes, cf.~\cite[eq.~(11.2.4)]{Borodin},
while the second~\eqref{eq:secondprobid} is an immediate consequence of the definition of the $\sigma$-thinning.
There is actually even a third probabilistic interpretation of $J_\s(X,T)$ in terms of the largest particle distribution of a {\it finite-temperature deformation} of the Airy process, cf.~\cite{AmirCorwinQuastel,CafassoClaeysRuzza}.
However this interpretation is less relevant for our purposes.

Before proceeding with the definition of the J\'anossy densities of the $\s$-thinned shifted and dilated Airy point process, we introduce a special notation for a particular case of these deformations. The parametric dependence of $J_\s(X,T)$ can be studied by keeping $T$ constant first. Then, because of the relation $J_\s(X,T)=J_{\tilde{\s}}(XT^{-\frac{1}{3}},1),$ where $ \tilde{\s}(\lambda):=\s(T^\frac{1}{3}\lambda)$,
$T$ can be set to $1$ without loss of generality. In this case, for the Fredholm determinant in \eqref{eq:TJ0} we will use the shorthand notation
\be 
j_\s(s):= J_{\s}(s,1)=1+\sum_{n=1}^\infty\frac{(-1)^n}{n!}\int_{\mathbb R^n}\rho_n^\Ai(\lambda_1+s,\dots,\lambda_n+s)\prod_{i=1}^n\sigma(\lambda_i)\d \lambda_i,\qquad s\in\R,
\ee 
corresponding to the full gap probability of the $\s$-thinned shifted (no dilation) Airy point process.
If we take for $\sigma$ the indicator function of $(0,+\infty)$, $j_\s(s)$ reduces to the classical Fredholm determinant representation of the Tracy--Widom distribution~\cite{TracyWidom}, see Example~\ref{examples:zeroHeaviside} below; note however that this choice of $\sigma$ is not admitted in view of Assumption \ref{assumption:weak}.

\begin{remark}
	For $j_\s(s)$ to be well defined, we need the trace condition $\int_\R\s(\lambda) K^\Ai(\lambda,\lambda)\d\lambda<\infty$. The decay in Assumption~\ref{assumption:weak} is slightly stronger than this requirement.
	This will allow us to control the $s\to +\infty$ behavior of certain objects, see in particular Section~\ref{sec:asymptoticssplusinfty}.
\end{remark}

We are now ready to define the main object in the present work: the J\'anossy densities of the $\s$-thinned shifted and dilated Airy point process.
We first notice that when $\s:\R\to[0,1]$ decays sufficiently fast at~$-\infty$ (for instance, if it satisfies Assumption~\ref{assumption:weak}), the $\s$-thinned shifted and dilated Airy point process has almost surely a finite number of particles (see Remark~\ref{remark:numberofparticles}).
It is indeed in such a case that, for a given set of \textit{distinct} real values $\nu_i,i=1\dots,m$, we can define the global {\it J\'anossy density} of order $m\geq 0$, denoted $J_\s(X,T|\nu_1,\dots,\nu_m)$. 
It represents the likelihood that a random configuration in the $\s$-thinned ground process contains no other points than the prescribed points at positions $\nu_1,\ldots, \nu_m$. In this way, it encodes information about the ground process, given the position of $m$ of its particles. The J\'anossy density of order $m$ is defined by the property that for all disjoint Borel subsets $B_1,\dots,B_\ell\subseteq\R$ such that $\sqcup_{j=1}^\ell B_j=\R$ and all integers $k_1,\dots,k_\ell\geq 1$ summing up to $\sum_{j=1}^\ell k_j=m$, the probability that a configuration in the $\s$-thinned shifted and dilated Airy point process contains \textit{exactly} $m$ particles, of which $k_1$ lie in $B_1$, $k_2$ in $B_2$, $\dots$, is
\be
\frac 1{k_1!\cdots k_\ell!}\int_{B_1^{k_1}\times\cdots\times B_l^{k_\ell}}J_\s(X,T|\nu_1,\dots,\nu_m)\prod_{j=1}^m\s(\nu_j)\d \nu_j.
\ee
It is known~\cite[eq.~(1.38)]{Soshnikov} that J\'anossy densities can be expressed as
\be
\label{def:Jsigma}
J_\s(X,T|\nu_1,\dots,\nu_m)=\sum_{n=0}^{\infty}\frac{(-1)^n}{n!}\int_{\mathbb R^n}\rho^\Ai_{X,T;n+m}(\ll_1,\dots,\ll_n,\nu_1,\dots,\nu_m)\prod_{i=1}^n\sigma(\ll_i)\d\ll_i.
\ee
For brevity, we will collect the {\it distinct} real numbers $\nu_i$ into a vector $\underline \nu:=(\nu_1,\dots,\nu_m)$ and denote $J_\sigma(X,T|\underline \nu):=J_\s(X,T|\nu_1,\dots,\nu_m)$.
In agreement with the previously introduced notation, we have $J_\s(X,T)=J_\s(X,T|\emptyset)$.

\medskip

As in the case $m=0$, it is interesting to study the situation when $T$ is kept constant first.
In such case, $T$ can be set to $1$ without loss of generality because again
\be
\label{eq:Janossyreduction}
J_\s(X,T|\underline \nu)=T^{-\frac m3}\,J_{\wt \s}(XT^{-\frac 13},1| T^{-\frac 13}\underline \nu),\qquad\wt\s(\ll):=\s(T^{\frac 13}\ll).
\ee
Accordingly, we will formulate our results in which $T$ is constant more concisely in terms of
\be
\label{eq:reducedJanossy}
j_\s(s|\underline \nu):=J_\s(s,1|\underline \nu),\qquad s\in\R,
\ee
the $m$-th J\'anossy density of the $\s$-thinned shifted (no dilation) Airy point process.
\begin{remark}
	J\'anossy densities in general carry important information about a point process, but global J\'anossy densities of a determinantal point process are defined only if the associated kernel defines a trace class operator. This is not the case for the Airy kernel operator. To remedy this, one commonly considers J\'anossy densities of determinantal point processes restricted to bounded sets $B$ \cite{BorodinSoshnikov}.
	It is less customary to consider J\'anossy densities of thinned determinantal point processes, as we do, but this has the advantage that, for a suitable class of thinning functions $\sigma$, global J\'anossy densities exist. 
{More important for us, global J\'anossy densities turn out to provide a natural framework in view of underlying differential equations.}	
	In the degenerate case $\sigma=1_B$, we recover the J\'anossy density of the determinantal point process restricted to the set $B$.
\end{remark}

\subsubsection*{Known results concerning the gap probabilities}

We now give an overview of recent results revealing connections between the gap probabilities $J_\s(X,T)$ (and $j_\s(s)$) and the theory of integrable systems. 

First, it was proved in a series of recent papers \cite{AmirCorwinQuastel, Bothner, BothnerCafassoTarricone,  CafassoClaeysRuzza} in slightly different forms, that with
\be
\label{eq:potential}
v_\s(s):=\pa_s^{2}\log j_\s(s),
\ee
one has
\be
\label{eq:trace}
v_\s(s)=-\int_\R\varphi_\s(\ll;s)^2\s'(\ll)\d\ll,\qquad\s'(\ll):=\frac{\d\s(\ll)}{\d\ll},
\ee
where $\varphi_\s(\ll;s)$ solves the {\it Stark equation}\footnote{The Stark equation is nothing else than the Schr\"odinger equation $\left(\partial_s^2+2u(s)\right)\varphi(\ll;s)=\ll \varphi(\ll;s)$, for a potential $u(s)=v_\sigma(s)-s/2$ with linear background $-s/2$.} $\left(\partial_s^2+2v_\sigma(s)-s\right)\varphi(\ll;s)=\ll\, \varphi(\ll;s)$.
More precisely, $\varphi_\s(\ll;s)$ is the unique solution to the {\it Stark boundary value problem}
\be
\label{eq:Stark}
\bigl(\pa_s^2+2v_\s(s)-s\bigr)\varphi_\s(\ll;s)=\ll\,\varphi_\s(\ll;s),\qquad\varphi_\s(\ll;s)\sim\Ai(\ll+s),\ \ s\to+\infty.
\ee
See Proposition~\ref{prop:asympwavefunction} for the proof of the boundary value under our current assumptions on $\s$ and Corollary~\ref{cor:uniqueness} for the uniqueness.
In particular, it follows by combining~\eqref{eq:trace} and \eqref{eq:Stark} that $\varphi_\s$ solves the {\it integro-differential Painlev\'e~II equation} of Amir, Corwin, and Quastel~\cite{AmirCorwinQuastel}
\be
\label{IDPII}
\pa_s^2\varphi_\s(\ll;s)=\biggl(\ll+s+2\int_\R\varphi_\s(\mu;s)^2\s'(\mu)\d\mu\biggr)\varphi_\s(\ll;s).
\ee
It is worth observing that the relation~\eqref{eq:trace} is the analogue, for potentials with linear background, of the classical {\it Trace Formula} obtained by Deift and Trubowitz~\cite[Equation $(1)_R$, page 183]{DeiftTrubowitz} in scattering theory for the Schr\"odinger equation for potentials with zero background. 

Concerning the $X,T $ dependence, we can exploit the relation $J_\s(X,T)=j_{\wt\s}(XT^{-\frac 13})$, where $\wt\s(\ll):=\s(T^{\frac 13}\ll)$, which implies that
\be
\label{eq:Tpotential}
V_\s(X,T):=\pa_X^2\log J_\s(X,T) = T^{-\frac 23}v_{\wt\s}(XT^{-\frac 13})
\ee
satisfies
\be
\label{eq:Ttrace}
V_\s(X,T)=-T^{-\frac 12}\int_\R\wh\varphi_\s(\ll;X,T)^2\s'(\ll)\d\ll
\ee
in terms of the function $\wh\varphi_\s(\ll;X,T):=T^{-\frac 1{12}}\varphi_{\wt\s}(\ll T^{-\frac 13};XT^{-\frac 13})$.
The latter is also equivalently characterized as the unique solution to the boundary value problem
\be
\label{eq:TStark}
\begin{cases}
	\mathscr L \,\wh\varphi_\s(\ll;X,T) = \ll \, \wh\varphi_\s(\ll;X,T)\,, & 
	\\
	\wh\varphi_\s(\ll;X,T)\sim T^{-\frac 1{12}}\,\Ai\bigl(T^{-\frac 13}(\ll+X)\bigr)\,,& X\to+\infty,
\end{cases}
\ee
where $\mathscr L:=T\,\pa_X^2+2\,T\,V_\s(X,T)-X$.

Moreover, it was proved in~\cite[Theorem~1.3]{CafassoClaeysRuzza} that $V=V_\s(X,T)$ solves the {\it cylindrical Korteweg--de\,Vries} (cKdV) equation
\be
\label{cKdV}
\pa_TV+\frac 1{12}\,\pa_X^3V+V\,\pa_XV +\frac 1{2T}\,V\,=\,0.
\ee
More precisely, the variables~$x\in\R,t>0$ of \cite{CafassoClaeysRuzza} are related to~$X\in\R,T>0$ as~$x=-XT^{-\frac 12}$, $t=T^{-\frac 12}$, such that the KdV equation for~$u=u(x,t)$
\be
\label{KdV}
\pa_tu+\frac 16\, \pa_x^3u+2\,u\,\pa_xu\,=\,0
\ee
\cite[equation~(1.7)]{CafassoClaeysRuzza} is equivalent to the cKdV equation~\eqref{cKdV} for the function
\be
\label{solutionKdVcKdV}
V(X,T):=T^{-1}\,u\bigl(x=-XT^{-\frac 12},t=T^{-\frac 12}\bigr)+\frac12XT^{-1}.
\ee
In the language of integrable PDEs, this implies that the Fredholm determinant $J_\s(X,T)$ is a {\it tau function} of the cKdV equation.
In particular, $J=J_\s(X,T)$ solves the {\it bilinear form} of the cKdV equation
\be
\label{eq:bilinearcKdV}
\pa_XJ\,\pa_TJ-J\,\pa_X\pa_TJ-
\frac 14\,\bigl(\pa_X^2J\bigr)^2+\frac 13\, \pa_XJ\,\pa_X^3J-\frac 1{12}\,J\,\pa_X^4J-\frac 1{2T}\,J\,\pa_XJ \,=\,0.
\ee
The direct and inverse scattering transform for the cKdV equation has been established in~\cite{Santini1, Santini2, ItsSukhanov0, ItsSukhanov} for smooth and sufficiently fast decaying initial data.

\begin{example}\label{examples:zero}
	The simplest situation occurs when~$\sigma=0$.
	Then,
	\be
	\quad J_0(X,T)=1,\qquad V_0(X,T)=0,\qquad \varphi_0(\lambda;X,T)={\rm Ai}\bigl(T^{-\frac 13}(\lambda+X)\bigr).
	\ee
\end{example}
\begin{example}\label{examples:zeroHeaviside}
	The function $\s=1_{(0,+\infty)}$ does not satisfy Assumption~\ref{assumption:weak}, but it is nevertheless an instructive degenerate situation.
	(The present setting could be extended to include such case, cf. Remark~\ref{rem:jumpsing} and Section \ref{section:discont}.)
	The integro-differential Painlev\'e~II equation reduces to the Painlev\'e II equation, and~\eqref{eq:trace} is the celebrated Tracy--Widom formula \cite{TracyWidom}.
	The cKdV tau function and solution are given by
	\be
	J_{1_{(0,+\infty)}}(X,T)=F_{\rm TW}(XT^{-\frac 13}),\qquad
	V_{1_{(0,+\infty)}}(X,T)=-T^{-\frac 23}y_{\rm HM}(XT^{-\frac 13})^2,
	\ee
	where~$F_{\rm TW}$ is the Tracy--Widom distribution \cite{TracyWidom} and $y_{\rm HM}$ is the Hastings--McLeod solution to Painlev\'e~II \cite{HastingsMcLeod}.
	The asymptotics of the Hastings--McLeod solution imply that (see Figure~\ref{fig:HMsol})
	\be
	V_{1_{(0,+\infty)}}(X,T)\sim\begin{cases}\frac 12XT^{-1},&XT^{-\frac 13}\to -\infty,\\
		-T^{-\frac 23}\Ai\bigl(XT^{-\frac 13}\bigr),&XT^{-\frac 13}\to +\infty.
	\end{cases}
	\ee
\end{example}

\begin{remark}
	The KdV and cKdV equations, \eqref{KdV} and \eqref{cKdV}, respectively, are completely equivalent from an algebraic point of view, since the transformation~\eqref{solutionKdVcKdV} defines a one-to-one correspondence of solutions.
	On the other hand, this correspondence drastically changes the analytic properties of solutions; e.g., if $V$ is bounded then $u$ is not, and vice versa. In view of the analytic properties of the solutions under consideration, we find it more natural to work with the cKdV equation; moreover, the relevant Riemann--Hilbert problem in our analysis formally matches with the one for the inverse scattering theory of the cKdV of Its and Sukhanov~\cite{ItsSukhanov0,ItsSukhanov}, cf. Section~\ref{sec:IS}.
\end{remark}

\begin{figure}[t]
	\centering
	\includegraphics[scale=.8]{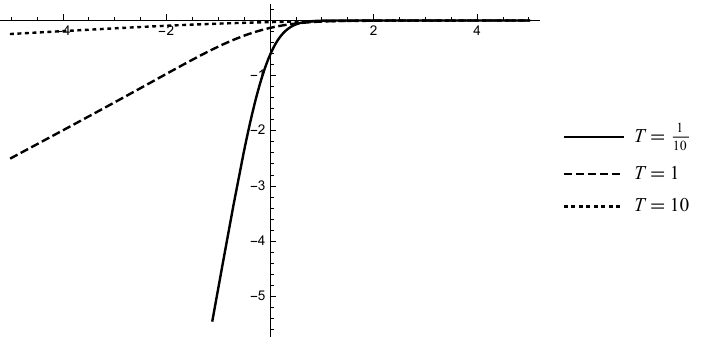}
	\caption{The solution $V_{1_{(0,+\infty)}}(X,T)$ as a function of $X$ for some values of $T$.}
	\label{fig:HMsol}
\end{figure}

The asymptotic behavior of the quantities $J_\sigma(X,T)$ and $V_{\sigma}(X,T)$ in different regions of the $(X,T)$-plane has been established recently in  \cite{CafassoClaeysRuzza, CharlierClaeysRuzza}. In particular,  left-tail ($X\to -\infty$) asymptotics of these quantities were obtained via Riemann--Hilbert techniques.
Some of these asymptotic results will play a role in the description of the asymptotic behavior of the $\underline \nu$-depending quantities $J_\sigma(X,T|\underline \nu)$ and $V_{\sigma}(X,T|\underline \nu)$ studied in the present paper, and we summarize the relevant results from \cite{CharlierClaeysRuzza}, which hold under more severe assumptions on $\sigma$, next. 

\begin{assumption}
	\label{assumption:strong}
	The function $F:=\tfrac1{1-\sigma}$ extends to an entire function.
	Moreover:
	\begin{itemize}
		\item $F'\geq 0$ and $(\log F)''\geq 0$ on the real line;
		\item $F(\lambda)=1+c'_-\e^{-c_-|\lambda|}(1+o(1))$ as $\lambda\to-\infty$, $F(\lambda)=c'_+\e^{c_+\lambda}(1+O(\e^{-\epsilon\lambda}))$ as $\lambda\to+\infty$, for some $c_\pm,c_\pm',\epsilon>0$;
		\item $F(\lambda)=O(\e^{c_+\Re\lambda})$ as $\Re\lambda\to+\infty$.
	\end{itemize}
	In particular, the second assumption implies the strong decay $\s(\ll)=c_-'\e^{-c_-|\ll|}\bigl(1+o(1)\bigr)$ as~$\ll\to-\infty$ and $\s(\ll)=1-\tfrac 1{c_+'}\e^{-c_+\ll}\bigl(1+O(\e^{-\epsilon\ll})\bigr)$ as~$\ll\to+\infty$ (for the same constants $c_\pm,c_\pm',\epsilon>0$).
\end{assumption}
The reader may want to keep in mind the prototype example of an admissible function $\sigma$ given by $\sigma(\lambda)=\frac{1}{1+\e^{-\lambda}}$, such that $F(\lambda)=1+\e^\lambda$.
The following result is part of \cite[Theorem 1.3]{CharlierClaeysRuzza}.
\begin{theorem}\label{thm:CCR}
	Let $\s$ satisfy Assumption~\ref{assumption:strong}.
	For any $T_0>0$ there exists $K>0$ such that
	\begin{align}
		\log J_\s(X,T|\emptyset)&=\rho^3T^2\left(-\frac{4}{15}\left(1-\xi\right)^{\frac 52}+\frac 4{15}-\frac 23\xi+\frac 12\xi^2\right)+O\left(|X|^{\frac 32}T^{-\frac 12}\right),\\
		\label{eq:heavytail}
		V_\s(X,T|\emptyset)&=\rho\left(1-\sqrt{1-\xi}\,\right)+O\left(|X|^{-\frac 12}T^{-\frac 12}\right),
	\end{align}
	where $\rho:=c_+^2/\pi^2$ and $\xi:=X/(\rho T)$, uniformly for $X\leq -KT^{\frac 12}$ and $T\geq T_0$.
\end{theorem}

\subsubsection*{New results for J\'anossy densities} 
 Our main novel contributions consist in extending the connection between the Airy point process, the cKdV equation, the integro-differential Painlevé II equation, and the Stark equation, to J\'anossy densities $ j_\s(s|\underline{\nu}), J_\s(X,T|\underline{\nu})$ for any given $\underline{\nu} $, and in establishing the asymptotic behavior of these quantities.

Our first result on the integrable structure of the J\'anossy densities is their expression in terms of the eigenfunctions of the Stark operator.

\begin{theoremintro}
	\label{thm:Janossyphi}
	Let~$\s$ satisfy Assumption~\ref{assumption:weak} and let $\varphi_\s(\lambda;s)$ be the unique solution to the Stark boundary value problem~\eqref{eq:Stark}.
	For all~$s\in\R$ and all $\underline \nu=(\nu_1,\dots,\nu_m)$ with $\nu_i\not=\nu_j$ for all $i\not=j$, we have
	\be
	\label{thm1eq1}
	j_\s(s|\underline \nu)=\det\left(L_s^\s(\nu_i,\nu_j)\right)_{i,j=1}^m j_\s(s|\emptyset),
	\ee
	where
	\be
	\label{thm1eq2}
	L_s^\s(\lambda,\mu)=\int_s^{+\infty}\varphi_\s(\ll;r)\varphi_\s(\mu;r)\,\d r=\frac{\varphi_\s(\lambda;s)\partial_s\varphi_\s(\mu;s)-\partial_s\varphi_\s(\lambda;s)\varphi_\s(\mu;s)}{\lambda-\mu},
	\ee
	and
	\be
	\label{thm1eq3}
	j_\sigma(s|\emptyset)=\exp\biggl(-\int_s^{+\infty}(r-s)\biggl(\int_\R\varphi_\s(\lambda;r)^2\s'(\ll)\d\ll\biggr)\d r\biggr).
	\ee
\end{theoremintro}

The proof is given in Section~\ref{sec:proofs12}.

\begin{remark}
	The kernel~$L_s^\s(\cdot,\cdot)$ induces a determinantal point process (depending parametrically on $s\in\R$) which is obtained via a conditioning of the shifted Airy point process, in the following sense: assign independently to each point~$\lambda$ in a random configuration mark $1$ with probability~$\sigma(\ll)$ and mark $0$ otherwise, then condition the resulting marked shifted Airy point process on the event that no points have mark $1$.
	The conditional point process obtained in this manner is determinantal, and has correlation kernel $L_s^\s(\cdot,\cdot)$, see Section~\ref{sec:conditional} and~\cite{ClaeysGlesner} and \cite{Bufetov1, Bufetov2, BufetovQiuShamov, GS} for details.
	The factorization~\eqref{thm1eq1} receives another interesting probabilistic interpretation: it is the product of an $m$-point correlation function in this conditional determinantal point process and of the gap probability of the $\s$-thinned shifted Airy point process.
\end{remark}

\begin{example}\label{example:zero2}
	It is instructive to consider again the case~$\s=0$, in which case $\varphi_0(\lambda;s)=\Ai(\lambda+s)$, and so~\eqref{thm1eq2} reduces to the shifted Airy kernel $L_s^0(\lambda,\mu;s)=K^{\Ai}(\ll+s,\mu+s)$, cf.~\eqref{eq:standardairykernel}.
	In this sense we can regard $\varphi_\s$ as a generalization of the Airy function, and $L_s^\sigma$ as a generalization or $\sigma$-deformation of the shifted Airy kernel; it is interesting to check that several properties of the Airy function and the shifted Airy kernel are preserved  under this deformation. See for instance Proposition~\ref{prop:asympwavefunction} and Corollary \ref{corollary:LkernelIntegral}.
\end{example}

\begin{example}\label{example:Heaviside2}
	The choice $\sigma=1_{(0,+\infty)}$ is again not admissible in view of Assumptions \ref{assumption:weak}, but with $v_\sigma(s)=\varphi(0;s)^2=-y_{\rm HM}(s)^2$, the degenerate case of~\eqref{eq:trace}, the Stark boundary value problem~\eqref{eq:Stark} still makes sense, and the kernel $L_s^{1_{(0,+\infty)}}(\lambda,\mu)$ is defined.
	Moreover, this kernel has appeared in the soft-to-hard edge transition in random matrix theory \cite[Theorem 1.3]{ClaeysKuijlaars}. In terms of the notations used in \cite[Theorem 1.3]{ClaeysKuijlaars}, we have
	\begin{align}
		\nonumber
		\varphi_{1_{(0,+\infty)}}(\lambda;s)&=-\frac{1}{\sqrt{2\pi}}f_0(-\lambda;s),\\
		\nonumber
		\partial_s\varphi_{1_{(0,+\infty)}}(\lambda;s)&=-\frac{1}{\sqrt{2\pi}}\bigl(g_0(-\lambda;s)+p_0(s)f_0(-\lambda;s)\bigr),
		\\ \intertext{and}
		L_s^{1_{(0,+\infty)}}(\lambda;\mu)&=\mathbb K_0^{\rm soft/hard}(-\lambda,-\mu;s).
	\end{align}
	Then we can identify the Stark boundary problem~\eqref{eq:Stark} with~\cite[Equations (1.12)--(1.15)]{ClaeysKuijlaars}.
	Furthermore, $L_s^{1_{(0,+\infty)}}$ is the kernel of the determinantal point process obtained by conditioning the shifted Airy point process on absence of particles on $(0,\infty)$.
	Recall also that $j_{1_{(0,+\infty)}}(s|\emptyset)=j_{1_{(0,+\infty)}}(s)=F_{\rm TW}(s)$ is the Tracy--Widom distribution in this case.
\end{example}

Next we give a second expression for the J\'anossy densities which is more directly parallel to equations~\eqref{eq:trace}, \eqref{eq:Stark}, and~\eqref{IDPII} for the case $m=0$.
This expression involves eigenfunctions of the Stark operator with a modified potential.

\begin{theoremintro}
	\label{thm:Janossyphi2}
	Let $\s$ satisfy Assumption~\ref{assumption:weak}.
	For all $s\in\R$ and all $\underline \nu=(\nu_1,\dots,\nu_m)$ with $\nu_i\not=\nu_j$ for all $i\not=j$, we have
	\be
	\label{eq:secondlogderJV}
	\pa_s^2\log j_\s(s|\underline \nu)=\int_\R\varphi_\s(\lambda;s|\underline \nu)^2\biggl(-\s'(\ll)+\sum_{i=1}^m\frac{2\bigl(1-\s(\ll)\bigr)}{\ll-\nu_i}\biggr)\d\ll
	\ee
	where $\varphi_\s(\lambda;s|\underline \nu)$ solves the Stark equation
	\begin{align}
		\label{StarkModified}
		\bigl(\pa_s^2+2v_\s(s|\underline \nu)-s\bigr)\varphi_\s(\lambda;s|\underline \nu)&=\lambda \varphi_\s(\ll;s|\underline \nu)
	\end{align}
	with potential
	\be
	\label{potentialModified}
	v_\s(s|\underline \nu):=\pa_s^2\log j_\s(s|\underline \nu).
	\ee
	Moreover, $\varphi_\s(\lambda;s|\underline\nu)$ can be expressed in terms of $\varphi_\s(\lambda;s|\emptyset)$ and $\pa_s\varphi_\s(\lambda;s|\emptyset)$ as
	\be
	\label{ModifiedVarphiDet}
	\varphi_\s(\lambda;s|\underline\nu)=\frac{\det\left(\begin{array}{cccc}
			\varphi_\s(\lambda;s|\emptyset) & L_s^\s(\lambda,\nu_1) & \cdots & L_s^\s(\lambda,\nu_m) \\ 
			\varphi_\s(\nu_1,s|\emptyset) & L_s^\s(\nu_1,\nu_1) & \cdots & L_s^\s(\nu_1,\nu_m) \\
			\vdots & \vdots & \ddots & \vdots \\
			\varphi_\s(\nu_m,s|\emptyset) & L_s^\s(\nu_m,\nu_1) & \cdots & L_s^\s(\nu_m,\nu_m) 
		\end{array}\right)}{\det\left(\begin{array}{ccc}
			L_s^\s(\nu_1,\nu_1) & \cdots & L_s^\s(\nu_1,\nu_m) \\
			\vdots & \ddots & \vdots \\
			L_s^\s(\nu_m,\nu_1) & \cdots & L_s^\s(\nu_m,\nu_m) 
		\end{array}\right)}\,,
	\ee
	so that, in particular, $\varphi_\s(\ll,s|\underline\nu)$ has zeros at $\lambda=\nu_1,\dots,\nu_m$.
\end{theoremintro}

The proof is given in Section~\ref{sec:proofs12}.
In the case $m=0$, the eigenfunctions $\varphi_\s(\ll;s|\emptyset)$ of this theorem reduce to what we denoted just $\varphi_\s(\ll;s)$ up to this point; for the sake of clarity, we will from now on use the notation $\varphi_\s(\ll;s|\emptyset)$.
Similarly said for the notation $v_\s(s)=v_\s(s|\emptyset).$

In other words, Theorem~\ref{thm:Janossyphi2} states that the Stark equation with potential $\pa_s^2\log j_\s(s|\underline\nu)$ is obtained by {\it Darboux transformations} (in the original spirit of Darboux~\cite{Darboux}) of the Stark equation with potential $\pa_s^2\log j_\s(s|\emptyset)$.

Moreover, by the asymptotics $\varphi_\s(\lambda;s|\emptyset)\sim\Ai(\lambda+s)=\varphi_0(\lambda;s|\emptyset)$ as $s\to+\infty$, one obtains that the appropriate boundary condition for the solution to the Stark equation~\eqref{StarkModified} is $\varphi_\s(\lambda;s|\underline\nu)\sim\varphi_0(\lambda;s|\underline\nu)$ as $s\to+\infty$.
The function $\varphi_0(\lambda;s|\underline\nu)$ is explicit in terms of $\Ai$ and $\Ai'$, by~\eqref{ModifiedVarphiDet} and $L_s^0=K^\Ai_s$.

\begin{remark}
The description of J\'anossy densities of determinantal point processes on $\R$, restricted to bounded intervals $B$ (instead of thinned point processes), in terms of solutions to certain differential equations has been recently developed in \cite{Nishigaki}.
The author there proved that for determinantal point processes defined through kernels satisfying the Tracy--Widom criteria \cite{TW94}, the Tracy--Widom method allows to express not only the gap probability (as proved in \cite{TW94}) but also the J\'anossy densities of the process restricted to a bounded interval $B$,  in terms of solutions to a system of differential equations in the endpoints of $B$.
For kernels also enjoying the integrable structure of Its--Izergin--Korepin--Slavnov~\cite{IIKS} (e.g., Airy, Bessel, sine kernels), the gap probability can be characterized by a Riemann--Hilbert (RH) problem.
This provides an alternative approach to study underlying integrable differential equations, and a powerful tool to tackle their asymptotics.
In this work, we extend the RH approach to study J\'anossy densities of the thinned shifted Airy point process. We are confident that our method can also be applied to other determinantal point processes with integrable structure, like the ones associated to Bessel and sine kernels.
\end{remark}

It follows by combining~\eqref{eq:secondlogderJV}, \eqref{StarkModified}, and~\eqref{potentialModified}, that $\varphi_\s(\lambda;s|\underline \nu)$ satisfies a deformation of the integro-differential Painlev\'e~II equation:
\be
\pa_s^2\varphi_\s(\ll;s|\underline \nu)=\biggl(\ll+s+2\int_\R\varphi_\s(\mu;s|\underline \nu)^2\biggl(\s'(\mu)-\sum_{i=1}^m\frac{2\bigl(1-\s(\mu)\bigr)}{\mu-\nu_i}\biggr)\d\mu\biggr)\varphi_\s(\ll;s|\underline \nu).
\ee
A general framework for studying integro-differential equations related to a class of Fredholm determinants was developed in \cite{Krajenbrink}.

\begin{remark}
	\label{rem:jumpsing}
	Theorems~\ref{thm:Janossyphi} and~\ref{thm:Janossyphi2} hold true more generally for all $\sigma$ satisfying the decay condition of Assumption~\ref{assumption:weak} even if they are only piecewise smooth with a finite number of jump singularities.
	In this case, one has to add to $\s'(\lambda)\,\d\ll$ a discrete measure supported at the singularities of $\sigma$.
	The extension to this case can be done following~\cite{CafassoClaeysRuzza},  as we will discuss in Section \ref{section:discont}.
\end{remark}

To construct a family of solutions to the cKdV equation in terms of the J\'anossy densities, we restore the full dependence on~$X,T$ and we introduce
\be
V_\s(X,T|\underline \nu):=\pa_X^2\log J_\sigma(X,T|\underline \nu),\qquad X\in\R,\ T>0.
\ee

\begin{theoremintro}\label{theorem:cKdV}
	For all $\s$ satisfying Assumption~\ref{assumption:weak} and all $\underline \nu=(\nu_1,\dots,\nu_m)$ with $\nu_i\not=\nu_j$ for all $i\not=j$, the function $V=V_\s(X,T|\underline \nu)$ solves the cKdV equation,
	\be
	\pa_TV+\frac 1{12}\,\pa_X^3V+V\,\pa_XV +\frac{1}{2T}\,V=0\,,\ \ \mbox{for all $X\in\R,T>0$.}
	\ee
\end{theoremintro}

The proof is in Section~\ref{sec:proofcKdV}, see Corollary~\ref{corollarythmckdv}.

\begin{example}\label{example:soliton}
	When $\s=0$ we have $j_{0}(s|\emptyset)=1$ and $L_s^{0}(\ll,\mu)=K^{\rm Ai}(\ll+s,\mu+s)$, such that, according to~\eqref{thm1eq1}, 
	\be 
	j_0(s|\underline \nu)=\rho^\Ai_m(\nu_1+s,\ldots,\nu_m+s)
	\ee
	is the $m$-point correlation function in the shifted Airy ensemble~\eqref{eq:corr}. 
	The corresponding cKdV tau function is
	\be
	\label{eq:J0nu}
	J_0(X,T|\underline \nu)=\det\bigl(K^{\rm Ai}_{X,T}(\nu_i,\nu_j)\bigr)_{i,j=1}^m.
	\ee
	The associated cKdV solution $V_0(X,T|\underline \nu)=\pa_X^2\log\det\bigl(K^{\rm Ai}_{X,T}(\nu_i,\nu_j)\bigr)_{i,j=1}^m$
	is a special case of soliton-type solution~\cite{Nakamura}, cf. Figure~\ref{fig:sol}, exhibiting right tail decay, and left tail rapid oscillations with decaying amplitude.
We will show below, cf.~\eqref{eq:asympcorrAiry}, that
\be
\label{eq:puresolition-asDet}
\det\bigl(K^\Ai_{X,T}(\nu_i,\nu_j)\bigr)_{i,j=1}^m \sim 
2^{-m(2m+1)}\pi^{-m} X^{-m^2}\prod_{1\leq i\leq m}\e^{-\frac 43T^{-\frac 12}(\nu_i+X)^{\frac 32}}\prod_{1\leq i<j\leq k}(\nu_i-\nu_j)^2
\ee
as $XT^{-\frac 13}\to+\infty$, uniformly for $T\geq T_0$.
	Since these asymptotics are moreover valid uniformly for complex $X$ with $|\arg X|<\delta$ and $\delta>0$ small, we can differentiate the logarithm of this expression twice which yields
	\be
	\label{eq:soliton-as}
	V_0(X,T|\underline \nu)\sim -\frac{m}{\sqrt{XT}},\qquad \mbox{as $XT^{-\frac 13}\to +\infty$, uniformly for $T\geq T_0$.}
	\ee
	Similarly, after straightforward computations involving the asymptotic behavior of the Airy function and its derivative, we obtain for the left tail if $m=1$,  as $XT^{-\frac 13}\to -\infty$, $T\geq T_0$,
	\begin{align}
		J_0(X,T|\nu)
		&=\frac{1}{\pi}\sqrt{\frac{|X|-\nu}{T}}\left(1-\frac{\sqrt{T}}{4}(|X|-\nu)^{{-}\frac 32}\cos\biggl(\frac{4}{3\sqrt T}(|X|-\nu)^{\frac 32}\biggr)+O(X^{-3}T)\right),
		\\
		\label{eq:soliton-as1-}V_0(X,T|\nu)
		&=\frac{1}{\sqrt{T|X|}}\cos\biggl(\frac{4}{3\sqrt T}(|X|-\nu)^{\frac 32}\biggr)+O(|X|^{-\frac 32}T^{-\frac 12}).
	\end{align}
\end{example}

\begin{figure}[t]
	\centering
	\begin{subfigure}{.19\textwidth}
		\includegraphics[scale=.23]{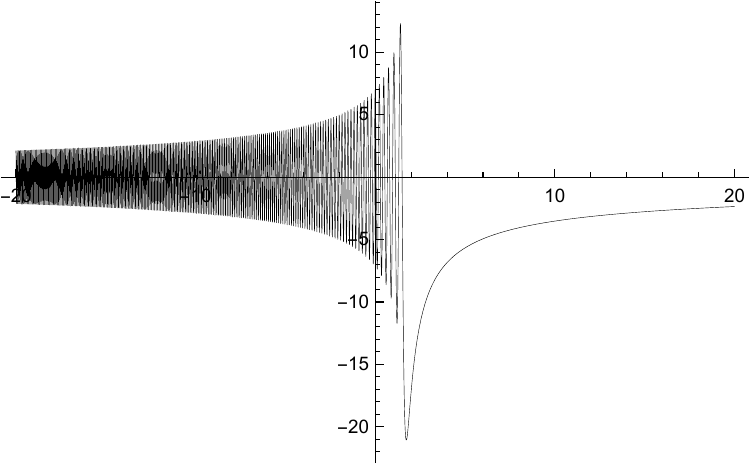}
		\caption*{$T=10^{-2}$}
	\end{subfigure}
	\begin{subfigure}{.19\textwidth}
		\includegraphics[scale=.23]{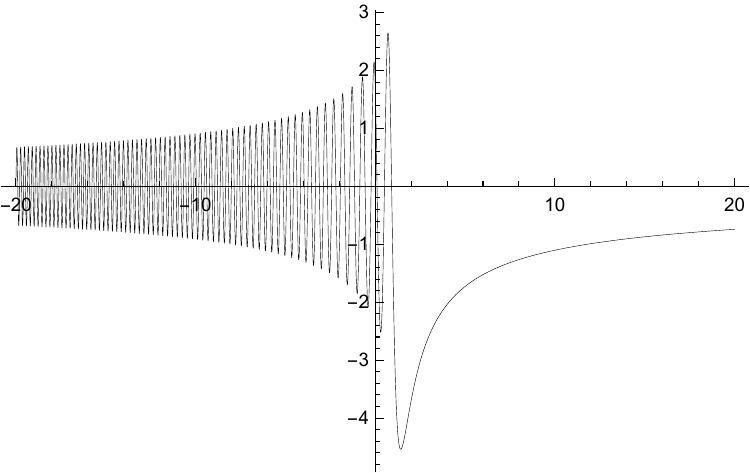}
		\caption*{$T=10^{-1}$}
	\end{subfigure}
	\begin{subfigure}{.19\textwidth}
		\includegraphics[scale=.23]{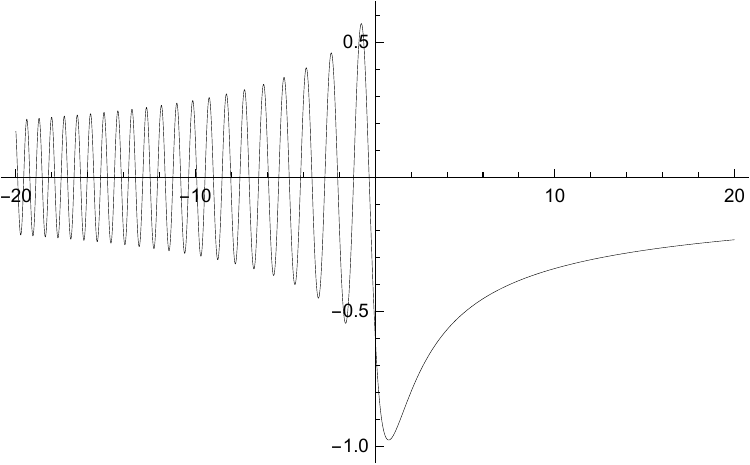}
		\caption*{$T=1$}
	\end{subfigure}
	\begin{subfigure}{.19\textwidth}
		\includegraphics[scale=.23]{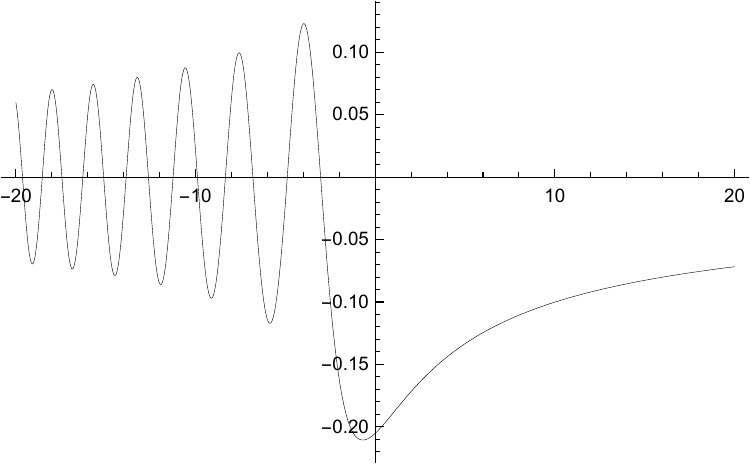}
		\caption*{$T=10$}
	\end{subfigure}
	\begin{subfigure}{.19\textwidth}
		\includegraphics[scale=.23]{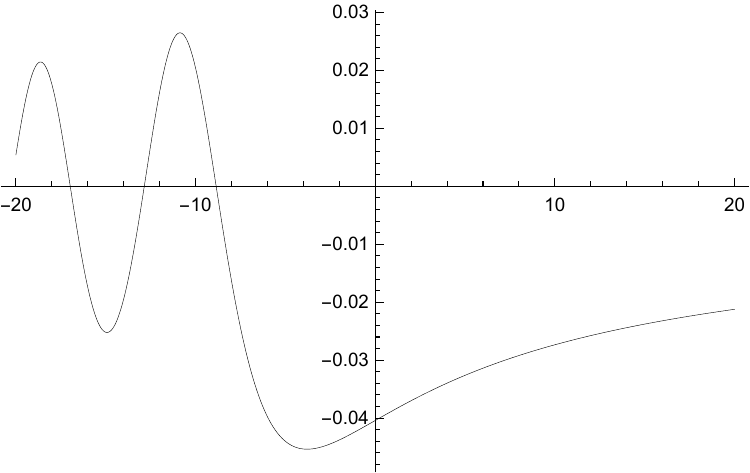}
		\caption*{$T=10^2$}
	\end{subfigure}
	\begin{subfigure}{.19\textwidth}
		\includegraphics[scale=.23]{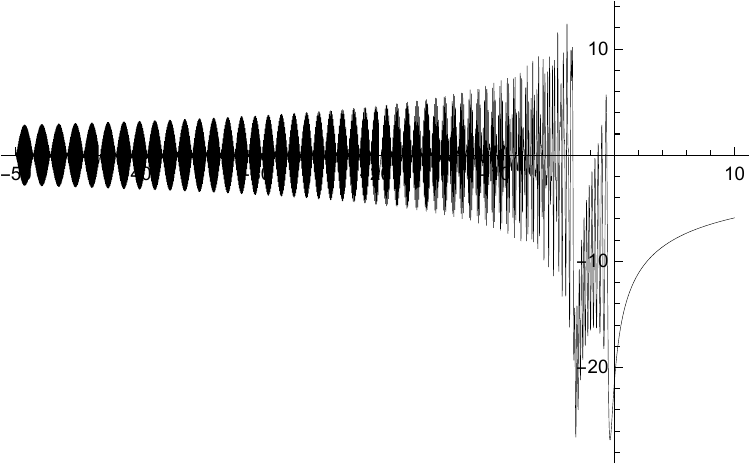}
		\caption*{$T=10^{-2}$}
	\end{subfigure}
	\begin{subfigure}{.19\textwidth}
		\includegraphics[scale=.23]{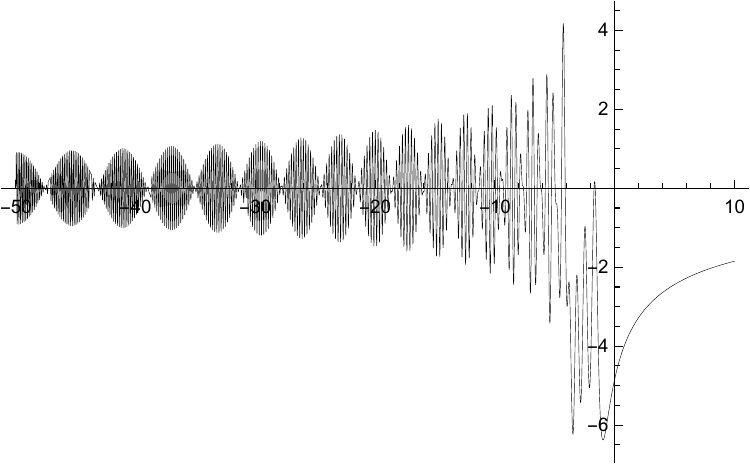}
		\caption*{$T=10^{-1}$}
	\end{subfigure}
	\begin{subfigure}{.19\textwidth}
		\includegraphics[scale=.23]{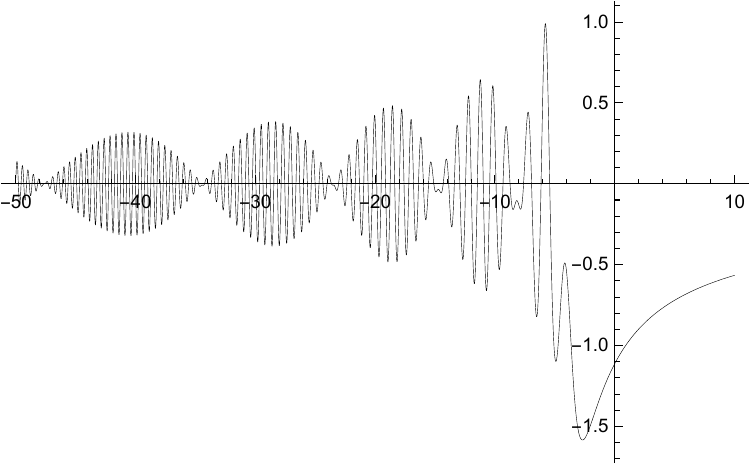}
		\caption*{$T=1$}
	\end{subfigure}
	\begin{subfigure}{.19\textwidth}
		\includegraphics[scale=.23]{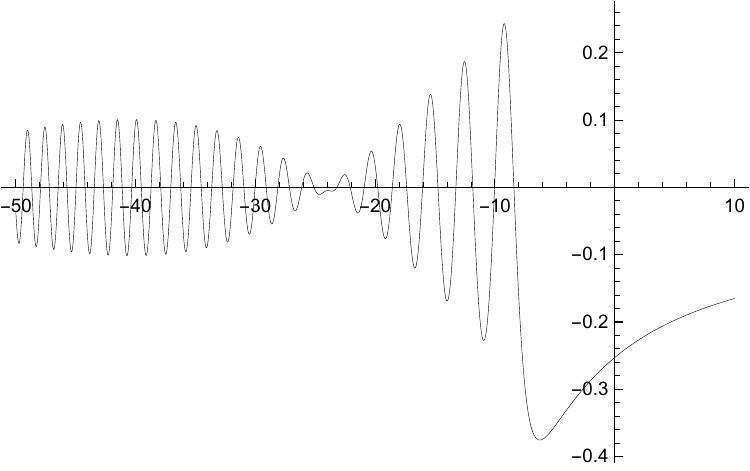}
		\caption*{$T=10$}
	\end{subfigure}
	\begin{subfigure}{.19\textwidth}
		\includegraphics[scale=.23]{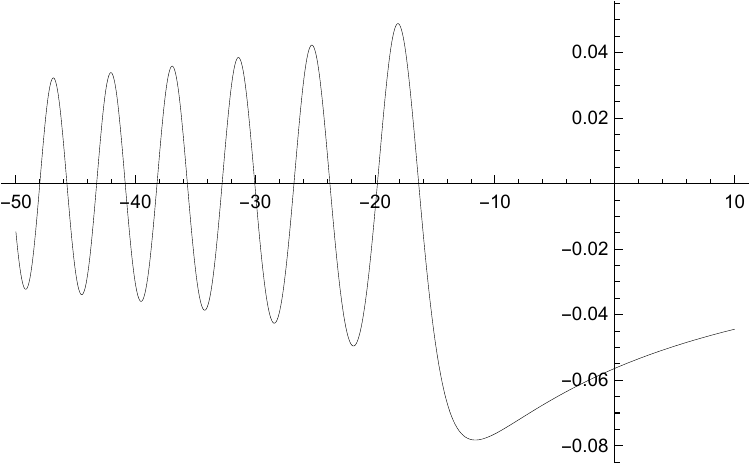}
		\caption*{$T=10^2$}
	\end{subfigure}
	\caption{First line: 1-soliton cKdV solution $V_{0}(X,T|\underline \nu)$ with $\underline \nu=(-2)$ as a function of $X$ for various values of $T$.
		Second line: 2-soliton cKdV solution $V_{0}(X,T|\underline \nu)$ with $\underline \nu=(0,3)$ as a function of $X$ for various values of $T$.}
	\label{fig:sol}
\end{figure}

\begin{remark}
The function $\sigma$ and the parameters $\underline\nu$ can be understood as scattering data for the cKdV solution under consideration. In analogy with \cite{GGJM, GGJMM}, we could interpret $\sigma$ as a function describing a gas of solitons.
\end{remark}

Our final result concerns the asymptotic behavior of the J\'anossy densities and associated cKdV solutions when $(X,T)$ goes to infinity in different regions of the $(X,T)$-plane, as illustrated in Figure \ref{figure: phase diagram}.
To this end, we need again to require stronger conditions on the functions $\s $, namely Assumption \ref{assumption:strong} (cf.~\cite{CafassoClaeysRuzza,CharlierClaeysRuzza}).

\begin{theoremintro}
\label{thm:asy}
Let $\underline\nu=(\nu_1,\dots,\nu_m)$ with $\nu_i\not=\nu_j$ for all $i\not=j$.

\noindent
{\bf (i)}~Let $\s$ satisfy Assumption~\ref{assumption:weak}.
For any $T_0>0$, there exists $c>0$ such that
\begin{align}
\label{eq:asy1LogJ}
J_\s(X,T|\underline\nu)&=\det\bigl(K^\Ai_{X,T}(\nu_i,\nu_j)\bigr)_{i,j=1}^m\bigl(1+O(\e^{-cXT^{-\frac 13}})\bigr),
\\
\label{eq:asy1V}
V_\s(X,T|\underline\nu)&=V_0(X,T|\underline\nu)+O\bigl(\e^{-cXT^{-\frac 13}}\bigr),
\end{align}
uniformly in $T\geq T_0$ as $XT^{-\frac 13}\to +\infty$. We recall that the asymptotics for $\det\bigl(K^\Ai_{X,T}(\nu_i,\nu_j)\bigr)_{i,j=1}^m$ and $V_0(X,T|\underline\nu)$ are given in~\eqref{eq:puresolition-asDet} and~\eqref{eq:soliton-as}.
	
\noindent {\bf (ii)}~Let $\s$ satisfy Assumption~\ref{assumption:strong}.
For any $T_0>0$, we have
\begin{align}
\label{eq:asy3LogJ}
J_\s(X,T|\underline\nu)&\sim \frac{|X|^{\frac m2}}{\pi^mT^{\frac m2}}J_\s(X,T|\emptyset)\prod_{j=1}^m 
\frac{1}{1-\sigma(\nu_j)},
\\
\nonumber
V_\s(X,T|\underline\nu)&=V_\s(X,T|\emptyset) \\
\label{eq:asy3V}
&\quad
+\frac{1}{\sqrt{|X|T}}\sum_{j=1}^m\cos\biggl(\frac{4|X|^{\frac 32}}{3T^{\frac 12}}(1+A_{X,T})-\frac{2|X|^{\frac 12}}{T^{\frac 12}}\nu_j(1+B_{X,T}(\nu_j))\biggr)+O(|X|^{-1}),
\end{align}
uniformly for $T\geq T_0$ as $\frac{X}{T\log^2|X|}\to -\infty$, where $A_{X,T}, B_{X,T}(\nu)$ converge to $0$ as $\frac{X}{T\log^2|X|}\to -\infty$.
We recall that the asymptotics for $\log J_\sigma(X,T|\emptyset)$ and $V_\sigma(X,T|\emptyset)$ are given in Theorem \ref{thm:CCR}.
\end{theoremintro}

Observe the specific structure of the $\nu_j$-dependence of the cKdV solution $V_\sigma(X,T|\underline\nu)$ and the J\'anossy density $J_\sigma(X,T|\underline\nu)$.
For~$XT^{-\frac 13}\to +\infty$, the leading order behavior of $J_\sigma(X,T|\underline\nu)$ depends on the positions $\nu_1,\ldots, \nu_m$ through their Vandermonde squared while the one of $V_\sigma(X,T|\underline\nu)$ depends on the number of points~$m$ but not explicitly on the positions $\nu_1,\ldots, \nu_m$.
On the other hand, for $\frac{X}{T\log^2 |X|}\to -\infty$, the effect of $\underline\nu$ is more prominent.
For~$J_\sigma(X,T|\underline\nu)$, it results in a product of factors depending on $\nu_j$, while for $V_\sigma(X,T|\underline\nu)$, the
presence of $\nu_1,\ldots, \nu_m$ results in a superposition of rapidly oscillating terms depending on $\nu_j$.

It is remarkable that both in the left and right tail asymptotics of~$V_\sigma(X,T|\underline\nu)$, we recognize (up to a sub-leading phase shift in the oscillatory terms) a superposition of $m$ 1-soliton solutions whose tail asymptotics are described in Example~\ref{example:soliton}, in addition to the leading order and $\underline\nu$-independent contribution coming from~$V_\s(X,T|\emptyset)$.

Recall that direct and inverse scattering theory for the cKdV equation has been established in~\cite{ItsSukhanov} for solutions $V(X,T)$ which are Schwartz class on the real line in $X$ when $T=1$. We emphasize however that the $X\to \pm\infty$ tails of $V_{\s}(X,T|\underline\nu)$ for fixed $T>0$ are too heavy to be included in this class.
Indeed, when $\underline\nu=\emptyset$, the asymptotic relation~\eqref{eq:heavytail} as $X\to+\infty$ shows that $V_{\s}(X,T|\emptyset)$ is not Schwartz class, despite exponential decay as $X\to-\infty$ (see part~(i) of Theorem~1.8 in~\cite{CafassoClaeysRuzza}).
When $\underline\nu\not=\emptyset$, even the fast decay as $X\to-\infty$ is destroyed, see~\eqref{eq:asy1V}.

\begin{remark}
Parts~(i) and (ii) of Theorem~\ref{thm:asy} will be proved by using respectively two different factorizations of the Janossy densities $J_\s(X,T|\underline \nu)$ given in Proposition \ref{prop:janossyfactorization}.
In particular, to achieve part~(ii), we need to understand the behavior of the kernel $\wh L_{X,T}^\s(\cdot,\cdot)$, which is the $T$-dependent counterpart of $L_s^\s(\cdot,\cdot)$, see~\eqref{eq:defLXT}. The control over the error term of the kernel diagonal turns out to be possible only by requiring $\frac{X}{T\log^2|X|}\to -\infty$, which is why this term appears in the final result. 
	A possible extension to  the full region $X/T\to -\infty$ is discussed in Remark~\ref{remark:extendas}.
\end{remark}

\begin{figure}[t]
	\centering
	\hspace{1cm}
	\begin{tikzpicture}[scale=1]
		\draw[->] (-9.5,0) -- (5.5,0) node[right] {$X$};
		\draw[-] (0,0) -- (0,3.3);
		\draw[->](0,3.9) -- (0,4.5) node[above] {$T$};
		\draw[dashed] (1.5,1.3) -- (5.5,1.3);
		\draw[dashed] (-2.5,1.3) -- (-9.5,1.3);
		\draw[domain=1:1.5, smooth, variable=\x, dashed] plot ({1.5*\x}, {1.3*\x*\x*\x}) node[above] {${\scriptstyle XT^{-\frac 13}=M}$};
		\draw[dashed] (-4.8,3.1) -- (2,3.1);
		\node at (-1.3,3.6) {intermediate regimes};
		\node at (3.7,2.7) {vanishing:};
		\node at (3.7,2.2) {\small $V_\s(X,T|\underline\nu)\sim-\frac m{\sqrt{XT}}$};
		\node at (-6.7,2.7) {superposition of};
		\node at (-6.7,2.2) {$V_\sigma(X,T|\emptyset)$ and oscillatory terms};
		\draw[domain=2.5:7, smooth, variable=\x, dashed] plot ({-\x},{1.3+23*(\x-2.5)/(ln(50*\x)*ln(50*\x))}) node[above] {${\scriptstyle K=\frac{|X|}{T\log^2|X|}}$};
	\end{tikzpicture}
	\caption{Phase diagram showing the different tail asymptotics for~$V_\sigma(X,T|\underline \nu)$, uniform in the indicated regions for fixed $M,K>0$.}
	\label{figure: phase diagram}
\end{figure}
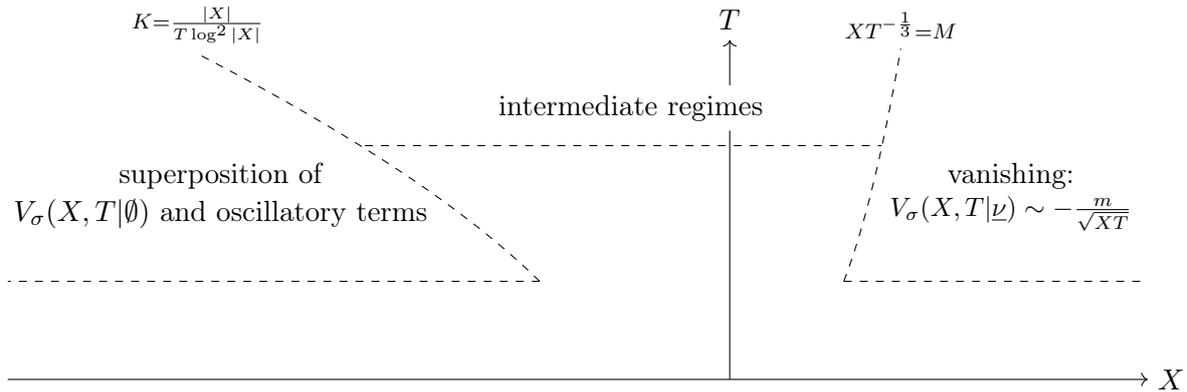

\subsubsection*{Methodology and outline}
In Section~\ref{sec:2}, we will gather several properties and identities for the J\'anossy densities on which we will rely later, and we will give a probabilistic interpretation to the kernel~$L_s^\sigma$. 
Of particular importance will be Proposition \ref{prop:janossyfactorization}, where we give two different factorizations of the J\'anossy densities. Both of them are relevant for two reasons. First, they have interesting probabilistic interpretations, as explained in Remark \ref{rem:prob interpr Janossy fact}. Secondly, they will be used both in Section \ref{sec:asymptotics} for the proof of Theorem \ref{thm:asy}, where part (i) and (ii) will be proved by using respectively the first and second factorization.

In Section \ref{sec:Rhcharact}, we will characterize the J\'anossy densities and other relevant quantities in terms of a {$2\times 2$ matrix-valued} Riemann--Hilbert (RH) problem, by relying on the Its--Izergin--Korepin--Slavnov method~\cite{IIKS}.
The main difference with the classical RH problems in the context of integrable operators will be the presence of a further condition, requiring specific singular behavior at the points $\nu_j\in \mathbb{R}$ ($j=1,\dots,m$) appearing in the definition of the J\'anossy densities~\eqref{def:Jsigma}.
	In this sense, this RH characterization shows strong similarities with the one from~\cite{BertolaCafasso}, where the authors studied general Schlesinger transformations on RH problems and the consequences of these transformations on associated quantities, like the Malgrange--Bertola differential and the tau function (or the Fredholm determinant in this case).
	We give a detailed comparison with the general methods of \cite{BertolaCafasso} in Section~\ref{sec:BC}.
Moreover, we establish a connection between the RH problem, the Stark boundary value problem~\eqref{eq:Stark}, and the cKdV equation~\eqref{cKdV}.
This will enable us to prove Theorem~\ref{thm:Janossyphi}, Theorem~\ref{thm:Janossyphi2}, and Theorem~\ref{theorem:cKdV}.
Section~\ref{sec:asymptotics} will be devoted to the asymptotic analysis of the RH problem from Section~\ref{sec:Rhcharact} and the related quantities $J_\s(X,T|\underline \nu), V_\s(X,T|\underline \nu)$.
	We will distinguish several regions in the $(X,T)$-plane which will require a different type of asymptotic analysis, and which lead to the results presented in Theorem~\ref{thm:asy}.
	We underline that the two parts of this result are obtained by different methods, in the following sense.
	As already mentioned, part~(i) relies on the first factorization of the J\'anossy densities from Proposition~\ref{prop:janossyfactorization} and thus its proof comes from classical asymptotic analysis of the Fredholm determinant factor in the right tail regime, see Lemma~\ref{lemma:H}. Part~(ii) relies instead on the second factorization from Proposition~\ref{prop:janossyfactorization} and its proof is derived by using previous results from~\cite{CafassoClaeysRuzza, CharlierClaeysRuzza} for the asymptotic behavior of the Fredholm determinant factor, summarized in Theorem~\ref{thm:CharlierClaeysRuzza}, and by exploiting the RH analysis of~\cite{CharlierClaeysRuzza} to obtain the asymptotic behavior of the finite determinant factor in the left-tail regime, achieved in Proposition~\ref{prop:decor}.
	Finally we conclude with some brief considerations about the various intermediate regimes.

\section{Preliminaries on J\'anossy densities}\label{sec:2}

In this section, we study in more detail the J\'anossy densities~$j_\s(s|\underline \nu)$ introduced in~\eqref{eq:reducedJanossy}.
The results could be easily translated into parallel results for~$J_\s(X,T|\underline \nu)$ by~\eqref{eq:Janossyreduction} and the observation that $\wt\s$ satisfies Assumption~\ref{assumption:weak} if $\s$ does, but we will omit the details for the sake of brevity. 

\subsection{Operator preliminaries}

For a given~$g\in L^\infty(\R)$, let~$\mathcal M_g$ be the multiplication operator on~$L^2(\R)$ defined by~$\mathcal M_{g}f=g f$ for all~$f\in L^2(\R)$, and let~$\mathcal K_s^\Ai$ be the operator acting on~$L^2(\R)$ through the shifted Airy kernel,
\be
(\mathcal K^\Ai_s f)(\ll) = \int_\R K_s^\Ai(\ll,\mu)f(\mu)\d\mu,\quad K_s^\Ai(\ll,\mu):=K^\Ai(\ll+s,\mu+s),\qquad f\in L^2(\R),
\ee
with $K^\Ai$ defined in~\eqref{eq:standardairykernel}.
It is worth recalling that $\mathcal K^\Ai_s$ is an orthogonal projector which can be represented as $\mathcal A_s \mathcal M_{1_{(0,+\infty)}}\mathcal A_s$ where $\mathcal A_s$ is the unitary involution of $L^2(\R)$ defined by
\be
(\mathcal A_sf)(\ll) = \int_{-\infty}^{+\infty}\Ai(\ll+\mu+s)f(\mu)\,\d \mu,\qquad f\in L^2(\R),
\ee
where the integral in the right-hand side is taken as an $L^2$-limit of $\int_{-\Lambda}^{+\infty}$ as $\Lambda\to+\infty$.

\begin{lemma}
\label{lemma1}
Let $\s$ satisfy Assumption~\ref{assumption:weak}.
The operator $\mathcal K_s^\s:=\mathcal M_{\sqrt\s}\mathcal K_s^\Ai\mathcal M_{\sqrt\s}$ is trace class on~$L^2(\R)$ and
\be
\label{eq:claimdet}
j_\s(s|\emptyset)=\det_{L^2(\R)}\bigl(1-\mathcal K_s^\s\bigr).
\ee
Moreover, $0<j_\s(s|\emptyset)\leq 1$ for all~$s\in\R$.
\end{lemma}
We denote the Fredholm determinant of a trace class perturbation of the identity by~$\displaystyle\det_{L^2(\R)}$.
\begin{proof}
We have $\mathcal K_s^\s=\mathcal H\mathcal H^\dagger$ where $\mathcal H:=\mathcal M_{\sqrt \s}\mathcal A_s 1_{(0,+\infty)}$.
It follows by the asymptotic properties of the Airy function at~$+\infty$ that $\mathcal H$ is Hilbert--Schmidt provided $\s$ satisfies Assumption~\ref{assumption:weak}.
Therefore, $\mathcal K_s^\s$ is the composition of two Hilbert--Schmidt operators, hence it is trace class on~$L^2(\R)$.
Then,~\eqref{eq:claimdet} follows by the classical formula for Fredholm determinants of operators with an integral kernel.
Next, since $\mathcal K_s^\Ai$ is an orthogonal projector and $0\leq\s\leq 1$, we have $(\mathcal K_s^\s)^2\leq\mathcal K_s^\s$ because
\be
(\mathcal K_s^\s)^2=\mathcal M_{\sqrt\s}\mathcal K_s^\Ai\mathcal M_\s \mathcal K_s^\Ai\mathcal M_{\sqrt\s}\leq\mathcal M_{\sqrt\s}(\mathcal K^\Ai_s)^2\mathcal M_{\sqrt\s}=\mathcal M_{\sqrt\s}\mathcal K^\Ai_s\mathcal M_{\sqrt\s}=\mathcal K_s^\s,
\ee
such that $0\leq\mathcal K_s^\s\leq 1$ and so $0\leq j_\s(s|\emptyset)\leq 1$. It remains to show that $j_\s(s|\emptyset)\not=0$, or, equivalently, that $1$ is not an eigenvalue of $\mathcal K_s^\s$.
For, assume $f\in L^2(\R)$ is such that $\mathcal K_s^\s f=f$. Setting $g:=\mathcal K_s^\Ai(\sqrt\s f)$, we have $\sqrt\s g=f$ and so, since $\mathcal K^\Ai_s$ is an orthogonal projector,
\be
\|g\|_2=\|\mathcal K^\Ai_s (\sqrt\s f)\|_2\leq \|\sqrt \s f\|_2=\|\s g\|_2.
\ee
Therefore, since $0\leq\s\leq 1$, we have $(\s-1)g=0$ almost everywhere on $\R$.
Since $\s\to 0$ at $+\infty$ (cf.~Assumption~\ref{assumption:weak}), $g$ has to vanish on some open set of $\R$. On the other hand, $g$ is the restriction to the real line of an entire function, as it follows from the fact that $\mathcal A_sh$ is entire for all~$h$ with support bounded below by standard properties of the Airy function. Therefore, $g$ is identically zero, so is $f$, and $1$ is not an eigenvalue of $\mathcal K_s^\s$. 
\end{proof}

\begin{remark}
\label{remark:numberofparticles}
$\mathcal M_{\sqrt\s}\mathcal K^\Ai_s\mathcal M_{\sqrt\s}$ acts on~$L^2(\R)$ through the kernel $\sqrt{\s(\ll)\s(\mu)}K^\Ai(\ll+s,\mu+s)$, which is a correlation kernel for the $\s$-thinned shifted Airy point process.
It follows from Lemma~\ref{lemma1} and from the general theory of determinantal point processes~\cite[Theorem~4]{Soshnikov} that the $\s$-thinned shifted and dilated Airy point process has almost surely a finite number of particles.
On the other hand, since $\mathcal K^\Ai$ is not trace class, the Airy point process has almost surely an infinite number of particles; it is however trace class once restricted to half-lines~$(t,+\infty)$ so that the Airy point process has almost surely a largest particle.
\end{remark}

\subsection{Conditional ensembles}\label{sec:conditional}

According to Lemma~\ref{lemma1}, the operator $1-\mathcal K_s^\s$ is invertible, and, therefore, so is $1-\mathcal M_\s\mathcal K^\Ai_s$.
Thus, it makes sense to introduce
\be
\label{eq:Lresolvent}
\mathcal L_s^\s:=\mathcal K_s^\Ai(1-\mathcal M_\s\mathcal K_s^\Ai)^{-1}=\mathcal K_s^\Ai+\mathcal K_s^\Ai\mathcal M_{\sqrt\s}(1-\mathcal K_s^\s)^{-1}\mathcal M_{\sqrt\s}\mathcal K_s^\Ai.
\ee
As we shall review below following the Its--Izergin--Korepin--Slavnov method~\cite{IIKS}, $\mathcal L_s^\s$ is an integral kernel operator, whose kernel we denote by~$L_s^\s(\cdot,\cdot)$.
It has been proved by the first two authors of this paper~\cite{ClaeysGlesner}, building on \cite{Bufetov1, Bufetov2, BufetovQiuShamov}, that this kernel induces a determinantal point process defined as follows.
Consider the shifted Airy process and construct a $\s$-marked point process by assigning to each point $\ll$ in a random configuration, independently, a mark~$1$ with probability~$\s(\ll)$ or a mark~$0$ with probability~$1-\s(\ll)$.
Conditioning the marked point process on the event that there are no points with mark~$1$, it is shown in op.~cit. that the resulting conditional ensemble is determinantal, with correlation kernel with respect to the {\it deformed} reference measure $\bigl(1-\s(\ll)\bigr)\d\ll$ given precisely by~$L_s^\s(\cdot,\cdot)$.
 
Let us introduce the following notation.
Given vectors $\underline u=(u_1,\dots,u_m)\in\R^m$ and $\underline w=(w_1,\dots,w_n)\in\R^n$, introduce the $m\times n$ matrix~$K_s^\Ai(\underline u,\underline w)\in\R^{m\times n}$ with entries
\be\label{eq:defKmatrix}
\left(K_s^\Ai(\underline u,\underline w)\right)_{i,j} := K_s^\Ai(u_i,w_j),\qquad 1\leq i\leq m,\ 1\leq j\leq n.
\ee
\begin{lemma}
\label{lemma2}
For any vector $\underline \nu=(\nu_1,\dots,\nu_m)$, with $\nu_i\not=\nu_j$ for all $i\not=j$, $K_s^\Ai(\underline \nu,\underline \nu)$ is positive-definite.
\end{lemma}
\begin{proof}
According to~\eqref{eq:standardairykernel} and to~\eqref{eq:defKmatrix}, we can rewrite $K^\Ai_s(\underline \nu,\underline \nu)$ as a Gram matrix
\be
\bigl(K^\Ai_s(\underline \nu,\underline \nu)\bigr)_{i,j}=\int_0^{+\infty}\Ai(\nu_i+\eta+s)\Ai(\nu_j+\eta+s)\,\d\eta = \bigl\langle \Ai(\nu_i+\cdot),\Ai(\nu_j+\cdot)\bigr\rangle_{L^2(s,+\infty)}.
\ee
Hence, it suffices to show that the $m$ vectors $\Ai(\nu_i+\cdot)\in L^2(s,+\infty)$, for $1\leq i\leq m$, are linearly independent when the points $\nu_1,\dots,\nu_m$ are distinct.
In order to obtain a contradiction, let us assume that the linear span of these $m$ vectors is $k$-dimensional with $k<m$ and, without loss of generality, that $\Ai(\nu_i+\cdot)$ for $i=1,\dots, k$ form a basis.
Then, there exists $c_1,\dots, c_k\in\C$ such that $\Ai(\nu_m+t)=\sum_{i=1}^kc_i\Ai(\nu_i+t)$ identically in~$t$.
Subtract $(t+\nu_m)$ times this relation from the second derivative of this relation in $t$ to get, using the Airy equation, $0=\sum_{i=1}^k(\nu_i-\nu_m)c_i\Ai(\nu_i+t)$.
Hence $c_i=0$ for all $1\leq i\leq k$ because $\nu_m\not=\nu_i$ for all $i\not=m$, and so $\Ai(\nu_m+t)=0$ identically in $t$, a contradiction.
\end{proof}

According to Lemma~\ref{lemma2}, for distinct points $\nu_1,\ldots, \nu_m\in\R$, collected into a vector $\underline \nu:=(\nu_1,\dots,\nu_m)$, we can introduce the integral kernel operator $\mathcal H_s^{\underline \nu}$ acting on $L^2(\R)$ through the kernel
\be
\label{eq:Palmkernel}
H_s^{\underline \nu}(\ll,\mu)
:=\frac{\det K_s^{\Ai}\bigl((\ll,\underline \nu),(\mu,\underline \nu)\bigr)}{\det K_s^{\Ai}(\underline \nu,\underline \nu)}
=K_s^\Ai(\ll,\mu)-K_s^\Ai(\ll,\underline \nu)K_s^\Ai(\underline \nu,\underline \nu)^{-1} K_s^\Ai(\underline \nu,\mu),
\ee
where the second equality stems from the well-known formula
\be
\label{eq:blockdet}
\det\left(\begin{array}{c|c} A&B\\ \hline C&D\end{array}\right)=\det(D)\det \left(A-BD^{-1}C\right),\qquad \mbox{if }\det D\not=0,
\ee
for the determinant of a block matrix with {lower-right} corner invertible.
It follows from the results in \cite{ShiraiTakahashi} that $H_s^{\underline \nu}(\ll,\mu)$ is the kernel of the {\it reduced Palm measure} of the shifted Airy process at (distinct) points $\nu_1,\ldots,\nu_m$, which can be interpreted as the shifted Airy {point} process conditioned on configurations containing points at~$\nu_1,\ldots,\nu_m$ and then removing the points~$\nu_1,\ldots,\nu_m$ from the configuration.

\subsection{Factorizations of J\'anossy densities}

We can factorize the J\'anossy densities~$j_\s(s|\underline \nu)$ in two different ways: the first {one} utilizes the Palm kernels $H_s^{\underline \nu}$, the second one involves the kernels $L_{s}^\s$ of the conditional ensembles.
It is convenient to introduce notations similar to~\eqref{eq:defKmatrix} for these kernels, namely, given vectors $\underline u=(u_1,\dots,u_m)\in\R^m$ and $\underline w=(w_1,\dots,w_n)\in\R^n$, we introduce matrices $L_s^\s(\underline u,\underline w)\in\R^{m\times n}$ and $H_s^{\underline \nu}(\underline u,\underline w)\in\R^{m\times n}$ with entries
\be
\label{eq:defLHmatrix}
\left(L_s^\s(\underline u,\underline w)\right)_{i,j}=L_s^\s(u_i,w_j),\quad
\left(H_s^{\underline \nu}(\underline u,\underline w)\right)_{i,j}=H_s^{\underline \nu}(u_i,w_j),\qquad 1\leq i\leq m,\ 1\leq j\leq n.
\ee
\begin{proposition}\label{prop:janossyfactorization}
For all~$\s$ satisfying Assumption~\ref{assumption:weak} and all $\underline \nu=(\nu_1,\dots,\nu_m)$ with $\nu_i\not=\nu_j$ for $i\not=j$, we have the identities
\begin{align}
\label{eq:idJanossy1}
j_\s(s|\underline \nu)&=\det\left(K_s^\Ai(\underline \nu,\underline \nu)\right)\det_{L^2(\R)}\left(1-\mathcal M_{\sqrt\s}\mathcal H_s^{\underline \nu}\mathcal M_{\sqrt\s}\right),
\\
\label{eq:idJanossy2}
j_\s(s|\underline \nu)&=\det\left(L_s^\sigma(\underline \nu,\underline \nu)\right)\det_{L^2(\R)}\left(1-\mathcal M_{\sqrt\s}\mathcal K_s^{\Ai}\mathcal M_{\sqrt\s}\right)=\det\left(L_s^\sigma(\underline \nu,\underline \nu)\right)\,j_\s(s|\emptyset).
\end{align}
\end{proposition} 
\begin{proof}
We start by rewriting~\eqref{def:Jsigma}:
\begin{align}
\nonumber
j_\s(s|\underline \nu)&=\sum_{n\geq 0}\frac{(-1)^n}{n!}\int_{\R^n}\det\left(K_s^\Ai\bigl((\underline\ll,\underline \nu),(\underline\ll,\underline \nu)\bigr)\right)\prod_{i=1}^n\s(\ll_i)\d\ll_i
\\
&=\det\left(K_s^\Ai(\underline \nu,\underline \nu)\right)\sum_{n=0}^{\infty}\frac{(-1)^n}{n!}\int_{\mathbb R^n}\det\left(H_s^{\underline \nu}(\underline\ll,\underline\ll)\right)\prod_{i=1}^n\sigma(\ll_i)\d\ll_i,
\end{align}
where we denote $\underline\ll=(\ll_1,\dots,\ll_n)$, $(\underline\ll,\underline\nu)=(\ll_1,\dots,\ll_n,\nu_1,\dots,\nu_m)$, and we manipulate the determinant of block matrices using~\eqref{eq:blockdet} and Lemma~\ref{lemma2} as
\begin{align}
\nonumber
\det\left(K_s^\Ai\bigl((\underline\ll,\underline \nu),(\underline\ll,\underline \nu)\bigr)\right)&=
\det\left(\begin{array}{c|c}
K_s^\Ai(\underline \ll,\underline\ll) & K_s^\Ai(\underline \ll,\underline \nu) \\ \hline
K_s^\Ai(\underline \nu,\underline\ll) & K_s^\Ai(\underline \nu,\underline \nu) 
\end{array}\right)
\\ \nonumber
&=\det\left(K_s^\Ai(\underline \nu,\underline \nu)\right)\det\left(K_s^\Ai(\underline \ll,\underline \ll)-K_s^\Ai(\underline \ll,\underline \nu)K_s^\Ai(\underline \nu,\underline \nu)^{-1}K_s^\Ai(\underline \nu,\ll)\right)
\\
&=\det\left(K_s^\Ai(\underline \nu,\underline \nu)\right)\det\left(H_s^{\underline \nu}(\underline \ll,\underline \ll)\right).
\end{align}
Hence, \eqref{eq:idJanossy1} is established.
Next, let us introduce the operator $\mathcal N:=\mathcal M_{\sqrt\s}(\mathcal K_s^\Ai-\mathcal H_s^{\underline \nu})\mathcal M_{\sqrt\s}$ such that
\be
\det_{L^2(\R)}(1-\mathcal M_{\sqrt\s}\mathcal H_s^{\underline \nu}\mathcal M_{\sqrt\s})=
\det_{L^2(\R)} (1-\mathcal K_s^\s)\,
\det_{L^2(\R)}\bigl(1+(1-\mathcal K_s^\s)^{-1}\mathcal N\bigr).
\ee
From~\eqref{eq:Palmkernel} we know that the kernel of~$\mathcal N$ is $N(\ll,\mu)=\sqrt{\s(\ll)\s(\mu)}K_s^\Ai(\ll,\underline \nu) K_s^\Ai(\underline \nu,\underline \nu)^{-1} K_s^\Ai(\underline \nu,\mu)$, such that the kernel of $(I-\mathcal K_s^\s)^{-1}\mathcal N$ is
\be
\wt L_s^\s(\lambda,\underline \nu)\bigl(K_s^\Ai(\underline \nu,\underline \nu)\bigr)^{-1}K_s^\Ai(\underline \nu,\mu)\sqrt{\s(\mu)}
\ee
where $\wt L_s^\s(\cdot,\cdot)$ is the kernel of $(1-\mathcal K_s^\s)^{-1}\mathcal M_{\sqrt\s}\mathcal K_s^\Ai$.
By the general formula for the Fredholm determinant of a finite-rank perturbation of the identity, cf.~\cite[Theorem~3.2]{GGK}, we obtain ($I_m$ denotes the $m\times m$ identity matrix)
\begin{align}
\nonumber
\det_{L^2(\R)}\bigl(1+(1-\mathcal K_s^\s)^{-1}\mathcal N\bigr)&=\det\left(I_m+K_s^\Ai(\underline \nu,\underline \nu)^{-1}\int_\R K_s^\Ai(\underline \nu,\ll)\sqrt{\s(\ll)}\wt L_s^\s(\lambda,\underline \nu)\d\ll\right)
\\
&=\frac{\det\bigl(K_s^\Ai(\underline \nu,\underline \nu)+\int_\R K_s^\Ai(\underline \nu,\ll)\sqrt{\s(\ll)}\wt L_s^\s(\lambda,\underline \nu)\d\ll\bigl)}{\det\bigl(K_s^\Ai(\underline \nu,\underline \nu)\bigl)}
=\frac{\det\bigl(L_s^\s(\underline \nu,\underline \nu)\bigl)}{\det\bigl(K_s^\Ai(\underline \nu,\underline \nu)\bigl)}
\end{align}
where we use the second identity in~\eqref{eq:Lresolvent}. Finally,~\eqref{eq:idJanossy2} follows from~\eqref{eq:idJanossy1}.
\end{proof}
\begin{remark}
\label{rem:prob interpr Janossy fact}
Both factorizations~\eqref{eq:idJanossy1} and~\eqref{eq:idJanossy2} have a natural probabilistic interpretation as products of an $m$-point correlation function with a gap probability.
In the first factorization, we have the $m$-point correlation function in the shifted and rescaled Airy point process, multiplied with the gap probability in the $\s$-thinning of the Palm measure at points $\nu_1,\ldots,\nu_m$ associated to the shifted and rescaled Airy point process. In the second factorization, we have the $m$-point correlation function in the conditional ensemble associated to the shifted and rescaled Airy point process introduced above, multiplied with the gap probability in the $\s$-thinning of the thinned shifted and rescaled Airy point process.
In the first factorization, the correlation function is simpler, but the gap probability is on the other hand simpler in the second factorization.
\end{remark}

Using the above result, it is now easy to show that J\'anossy densities are strictly positive for all distinct $\nu_1,\dots,\nu_m$.

\begin{proposition}
\label{prop:sec2}
For all~$\s$ satisfying Assumption~\ref{assumption:weak} and all $\underline \nu=(\nu_1,\dots,\nu_m)$ with $\nu_i\not=\nu_j$ for $i\not=j$, we have $\det L_s^\s(\underline \nu,\underline \nu)>0$ and $j_\s(s|\underline \nu)>0$.
\end{proposition}
\begin{proof}
The operator $\mathcal N:=\mathcal M_{\sqrt\s}(\mathcal K_s^\Ai-\mathcal H_s^{\underline \nu})\mathcal M_{\sqrt\s}$ is nonnegative-definite.
Indeed, for all~$\phi\in L^2(\R)$,
\be
\left\langle\mathcal N\phi,\phi\right\rangle=\underline h^\dagger\bigl(K_s^\Ai(\underline \nu,\underline \nu)\bigr)^{-1}\underline h\geq 0\,,
\ee
where $\underline h=\int_\R\sqrt\s(\ll)\phi(\ll)K_s^\Ai(\underline \nu,\ll)\d\ll$.
Hence, we have proved that
\be
1-\mathcal M_{\sqrt\s}\mathcal H_s^{\underline \nu}\mathcal M_{\sqrt\s}\geq 1-\mathcal K^\s_s.
\ee
Since $1-\mathcal K^\s_s$ is (strictly) positive-definite by Lemma~\ref{lemma1}, the operator $1-\mathcal M_{\sqrt\s}\mathcal H_s^{\underline \nu}\mathcal M_{\sqrt\s}$ is also positive-definite, hence invertible.
Therefore, $j_\s(s|\underline \nu)>0$ by Lemma~\ref{lemma2} and the first factorization of J\'anossy densities~\eqref{eq:idJanossy1}, and therefore~$\det\bigl(L_s^\s(\underline \nu,\underline \nu)\bigr)>0$ by the second one~\eqref{eq:idJanossy2}.
\end{proof}

\section{RH characterization of J\'anossy densities}\label{sec:Rhcharact}

The aim of the section is to give RH characterizations of J\'anossy densities, in order to prove Theorems~\ref{thm:Janossyphi}, \ref{thm:Janossyphi2}, and~\ref{theorem:cKdV}.

\subsection{RH problems}
\label{subsec:RHpbs}
The operator $\mathcal M_{\s}\mathcal K_s^\Ai$ is {\it integrable} in the sense of Its--Izergin--Korepin--Slavnov (IIKS)~\cite{IIKS}, namely it is a kernel operator whose kernel can be expressed as
\be
\label{eq:sigmaKintegrable}
\frac{\mathbf f(\lambda;s)\mathbf h(\mu;s)}{\lambda-\mu},\quad
\mathbf f(\lambda;s):=\s(\lambda)\begin{pmatrix}
-\i\,\Ai'(\lambda+s) \\ \Ai(\lambda+s) 
\end{pmatrix},\
\mathbf h(\mu;s):=\begin{pmatrix}
-\i\,\Ai(\mu+s) \\ \Ai'(\mu+s) 
\end{pmatrix}.
\ee
Therefore, according to op.~cit., the resolvent operator~$(1-\mathcal M_\s\mathcal K_s^\Ai)^{-1}-1$ can be characterized in terms of the following RH problem (see proof of Proposition~\ref{prop:dlogJ} below).

\subsubsection*{RH problem for $Y_\s$}
\begin{itemize}
	\item[(a)] $Y_\s(\cdot;s): \mathbb{C}\setminus \mathbb{R} \to \mathbb{C}^{2\times 2}$ is analytic for all $s\in\R$. 
	\item[(b)] The boundary values of $Y_\s(\cdot;s)$ are continuous on $\mathbb{R}$ and are related by
	\be
	Y_{\s,+}(\lambda;s) = Y_{\s,-}(\lambda;s) \bigl(I-2\pi\i\,\mathbf f(\lambda;s)\mathbf h^\top(\lambda;s)\bigr), \qquad \lambda \in \R,
	\ee
 	where the subscript $+$ (respectively, $-$) indicates the boundary value from above (respectively, below) the real axis.
	\item[(c)] As $\lambda \to \infty$, we have
	\be
	\label{eq:Yasympinfty}
	Y_\s(\lambda;s) = I + \frac 1\lambda\begin{pmatrix}
		\beta_\s(s) & \i\eta_\s(s) \\
		\i\alpha_\s(s) & -\beta_\s(s)
	\end{pmatrix}+O(\lambda^{-2}),
	\ee
for some $\alpha_\s(s)$, $\beta_\s(s)$, and $\eta_\s(s)$.
\end{itemize}

The following result has been proven in \cite{CafassoClaeysRuzza}.
For the reader's convenience, we offer a direct proof based on the IIKS method.

\begin{proposition}
\label{prop:dlogJ}
The RH problem for~$Y_\s$ has a unique solution for all~$s\in\R$ and we have
\be
\label{eq:dLogJ}
\pa_s\log j_\s(s|\emptyset)=-\alpha_\s(s),
\ee
where $\alpha_\s$ is given in \eqref{eq:Yasympinfty}.
\end{proposition}
\begin{proof}
The upshot of IIKS theory~\cite{IIKS} is that the RH problem for $Y_\s$ is uniquely solvable if and only if $1-\mathcal M_s^\s\mathcal K_s^\Ai$ is invertible.
The latter condition holds true by Lemma~\ref{lemma1}.
Moreover, in this case, the resolvent operator $(1-\mathcal M_\s\mathcal K_s^\Ai)^{-1}-1$ is also {an integral operator} with kernel
\be
\label{eq:kernelresolvent}
\frac{\mathbf f^\top(\lambda;s) Y^\top_\s(\lambda;s) Y^{-\top}_\s(\mu;s)\mathbf h(\mu;s)}{\lambda-\mu}.
\ee
Therefore, using Jacobi variational formula and the identity
\be
\pa_sK_s^\Ai(\lambda,\mu)=-\Ai(\lambda+s)\Ai(\mu+s),
\ee
which follows directly from~\eqref{eq:standardairykernel}, we compute $\frac\pa{\pa s}\log j_\s(s|\emptyset)$ as
\begin{align}
\nonumber
&-\tr\left((1-\mathcal M_\s\mathcal K_s^\Ai)^{-1}\mathcal M_{\s}\pa_s\mathcal K_s^\Ai\right)
\\
\nonumber
&=-\tr\left(\bigl((1-\mathcal M_\s\mathcal K_s^\Ai)^{-1}-1\bigr)\mathcal M_\s\pa_s\mathcal K_s^\Ai\right)-\tr\left(\mathcal M_\s\pa_s\mathcal K_s^\Ai\right)
\\
\nonumber
&=\int_\R\int_\R\frac{\mathbf f^\top(\lambda;s)Y_\s^\top(\lambda;s)Y_\s^{-\top}(\mu;s)\mathbf h(\mu;s)}{\lambda-\mu}\s(\mu)\Ai(\mu+s)\Ai(\lambda+s)\d\lambda\d\mu+\int_\R\sigma(\mu)\Ai(\mu+s)^2\d\mu
\\
\nonumber
&=\int_\R\biggl[(\i,0)\int_\R\frac{\mathbf h(\lambda;s)\mathbf f^\top(\lambda;s)Y_\s^\top(\lambda;s)}{\lambda-\mu}\d\lambda\biggr]Y_\s^{-\top}(\mu;s)\mathbf h(\mu;s)\sigma(\mu)\Ai(\mu+s)\d\mu+\int_\R\sigma(\mu)\Ai(\mu+s)^2\d\mu
\\
\nonumber
&\mathop{=}^{(*)}\int_\R (\i,0)\bigl(I-Y_\s^\top(\mu;s)\bigr)Y_\s^{-\top}(\mu;s)\mathbf h(\mu;s)\s(\mu)\Ai(\mu+s)\d\mu+\int_\R\sigma(\mu)\Ai(\mu+s)^2\d\mu
\\
&=\int_\R (\i,0)Y_\s^{-\top}(\mu;s)\mathbf h(\mu;s)\mathbf f^\top(\mu;s)\begin{pmatrix} 0 \\ 1 \end{pmatrix}\d\mu,
\label{eq:susp}
\end{align}
where we use the expressions of $\mathbf f,\mathbf h$ given in \eqref{eq:sigmaKintegrable}, and in the equality $(*)$ we use the identity
\be
Y_\s^\top(\mu)=I-\int_{\R}\frac{\mathbf h(\lambda;s)\mathbf f^\top(\lambda;s)Y_\s^\top(\lambda;s)}{\lambda-\mu}\d\lambda,
\ee
which follows from the RH problem satisfied by $Y_\s$ and the Sokhotski--Plemelj formula.
Finally, \eqref{eq:susp} can be simplified by a residue computation:
\begin{align}
\nonumber
\int_\R (\i,0)Y_\s^{-\top}(\mu;s)\mathbf h(\mu;s)\mathbf f^\top(\mu;s)\begin{pmatrix} 0 \\ 1 \end{pmatrix}\d\mu 
&=\i\left(\frac 1{2\pi\i}\int_\R \bigl(Y^{-\top}_{\s,+}(\mu;s)-Y^{-\top}_{\s,-}(\mu;s)\bigr)\d\mu\right)_{1,2}
\\
&=\i\left(\res{\mu=\infty}Y_\s^{-\top}(\mu;s)\right)_{1,2}=-\alpha_\s(s),
\end{align}
using \eqref{eq:Yasympinfty}.
\end{proof}

Next we consider the following RH problem which depends on~$s\in\R$ and on a finite number of distinct points $\nu_1,\dots,\nu_m$, $\nu_i\not=\nu_j$ for $i\not=j$, as usual collected into the vector~$\underline \nu:=(\nu_1,\dots,\nu_m)$. 
When $m=0$, condition (c) below is empty.

\subsubsection*{RH problem for $\Psi_\s$}
\begin{itemize}
	\item[(a)] $\Psi_\s(\cdot;s|\underline \nu): \mathbb{C}\setminus \mathbb{R} \to \mathbb{C}^{2\times 2}$ is analytic for all $s\in\R$ and all $\underline \nu$. 
	\item[(b)] The boundary values of $\Psi_\s(\cdot;s|\underline \nu)$ are continuous on $\mathbb{R}\setminus\{\nu_1,\dots,\nu_m\}$ and are related by
	\be
	\Psi_{\s,+}(\ll;s|\underline \nu) = \Psi_{\s,-}(\ll;s|\underline \nu) \begin{pmatrix}
		1 & 1-\s(\ll) \\
		0 & 1
	\end{pmatrix}, \qquad \ll \in \R, \ \ll\not=\nu_1,\dots,\nu_m.
	\ee
 	\item[(c)] For all $i=1,\dots,m$, as $\lambda \to \nu_i$ from either side of the real axis we have 
	\be
	\label{eq:PsiasympV}
	\Psi_\s(\lambda;s|\underline \nu) (\lambda-\nu_i)^{-\sigma_3}=O(1).
	\ee
	\item[(d)] As $\lambda \to \infty$, we have
	\be
	\label{eq:Psiasympinfty}
	\Psi_\s(\lambda;s|\underline \nu) = \biggl( I + \frac 1\ll\begin{pmatrix}
		q_\s(s|\underline \nu) & \i r_\s(s|\underline \nu) \\
		\i p_\s(s|\underline \nu) & -q_\s(s|\underline \nu)
	\end{pmatrix} +O(\lambda^{-2}) \biggr) \lambda^{\frac 14\sigma_{3}}G\e^{\left(-\frac{2}{3}\lambda^{\frac 32}-s\lambda^{\frac 12}\right)\sigma_{3}} C_\delta
	\ee
for any $\delta \in (0,\frac{\pi}{2})$. Here we take the principal branches of $\lambda^{\frac 14\sigma_3}$ and $\lambda^{\frac 12}$, analytic in $\C \setminus \left(-\infty,0\right]$ and positive for $\lambda>0$, and
		\be
	\label{eq:G}
	\s_3 := \begin{pmatrix}
		1 & 0 \\
		0 & -1
	\end{pmatrix}, \qquad G := \frac{1}{\sqrt{2}} \begin{pmatrix}
		1 & -\i \\ -\i & 1
	\end{pmatrix},\qquad C_\delta:=\begin{cases}
		I, & |\arg \lambda | < \pi -\delta, \\
		\begin{pmatrix}
			1 & 0 \\
			\pm 1 & 1
		\end{pmatrix}, & \pi - \delta < \pm \arg \lambda < \pi.
	\end{cases}
	\ee
\end{itemize}

\begin{remark}
We shall explain in detail in Section~\ref{sec:IS} the {relation of this RH problem with} the one in~\cite{ItsSukhanov} related to the inverse scattering for the cKdV equation.
\end{remark}

\begin{remark}
\label{rem:uniquenessPsiV}
The solution to this RH problem is unique by a standard argument in RH problems based on Liouville and Morera theorems.
Moreover, as we will show, the solution exists and can be constructed in terms of the solution to the RH problem for~$Y_\s$ (by an Airy dressing) and of a suitable matrix-valued rational function (by a Schlesinger transformation~\cite{BertolaCafasso}).
\end{remark}

We first recall the case $m=0$, which has already been considered in \cite{CafassoClaeysRuzza}.
To this end we introduce the Airy model RH problem in the following form.

\subsubsection*{RH problem for $\Phi^\Ai$}
\begin{itemize}
\item[(a)] $\Phi^\Ai$ is analytic in $\mathbb C\setminus \mathbb R$.
\item[(b)] The boundary values of $\Phi^\Ai$ are continuous on $\R$ and are related by
\be
	\Phi_+^\Ai(\lambda)=\Phi_-^\Ai(\lambda)\begin{pmatrix}
								1 & 1 \\
								0 & 1
							\end{pmatrix}, \qquad \lambda \in \mathbb R.
\ee
\item[(c)] As $\lambda\to\infty$, $\Phi^\Ai$ has the asymptotic behavior
\be
\Phi^\Ai(\lambda) = \left( I +\lambda^{-1}\begin{pmatrix}0&\frac {7\i}{48}\\
0&0\end{pmatrix} +\mathcal O(\lambda^{-2}) \right) \lambda^{\frac 14\s_3}G\e^{-\frac 23\lambda^{\frac 32}\s_3}C_\delta,
\ee
for any $0<\delta<\pi/2$ where $G,C_\delta$ are given in \eqref{eq:G} and the branches of $\lambda^{\frac 14\s_3}$ and $\lambda^{\frac 12}$ are as in~\eqref{eq:Psiasympinfty}.
\end{itemize}
The (unique) solution can be expressed in terms of the Airy function as
\be
\label{eq:defPhi}
\Phi^\Ai(\lambda):=
\begin{cases}
-\sqrt{2 \pi }
\begin{pmatrix} \Ai'(\lambda) &-\e^{\frac{2\i\pi}3}\Ai'(\e^{\frac{-2\i\pi}3}\lambda) \\ \i\,\Ai(\lambda) & -\i\e^{\frac{-2\i\pi}3}\Ai(\e^{\frac{-2\i\pi}3} \lambda)  \end{pmatrix}, 
&\mbox{\rm if }\Im \lambda>0,
\\[15pt]
-\sqrt{2 \pi }
\begin{pmatrix} \Ai'(\lambda) &\e^{\frac{-2\i\pi}3}\Ai'(\e^{\frac{2\i\pi}3}\lambda) \\ \i\,\Ai(\lambda) & \i\e^{\frac{2\i\pi}3}\Ai(\e^{\frac{2\i\pi}3}\lambda)  \end{pmatrix},
&\mbox{\rm if }\Im \lambda<0.
\end{cases}
\ee

\begin{proposition}
\label{prop:Psiempty}
When $m=0$, the RH problem for $\Psi_\s$ has a unique solution for all $s\in\R$ which can written as
\be
\label{eq:PsiEmptyformula}
\Psi_\s(\lambda;s|\emptyset)=\begin{pmatrix}
	1&\frac{\i s^2}{4}\\
	0&1
\end{pmatrix}Y_\s(\lambda;s)\Phi_s^\Ai(\lambda)
\ee
where $\Phi_s^\Ai(\lambda):=\Phi^\Ai(\lambda+s)$.
Moreover,
\be
\label{eq:pa}
p_\s(s|\emptyset)=\alpha_\s(s)+\frac{s^2}4,
\ee
and the kernel $L_s^\sigma(\lambda,\mu)$ of the operator $\mathcal L_s^\s:=\mathcal K_s^\Ai(1-\mathcal M_\s\mathcal K_s^\Ai)^{-1}$ can be written as
\be
\label{eq:Lkernel}
L_s^\s(\lambda,\mu) = 
\frac{\left(\Psi_\s(\mu;s|\emptyset)^{-1}\Psi_\s(\lambda;s|\emptyset)\right)_{2,1}}{2\pi\i(\lambda-\mu)}.
\ee
\end{proposition}
\begin{proof}
{As explained in Remark~\ref{rem:uniquenessPsiV}, uniqueness of the solution follows from standard arguments,} so it suffices to verify that \eqref{eq:PsiEmptyformula} solves the RH problem.
Condition (a) is easily checked, while for condition (b) we use the identity
\be
I-2\pi\i\,\mathbf f(\lambda;s)\mathbf h^\top(\lambda;s) =\Phi_{s,-}^\Ai(\lambda) \begin{pmatrix}1 & 1-\sigma(\lambda) \\ 0 & 1 \end{pmatrix}\Phi_{s,+}^\Ai(\lambda)^{-1},
\ee
which follows directly from the identities
\be
\label{eq:fgPhiAi}
\mathbf f(\lambda;s)=\frac{\i\,\sigma(\lambda)}{\sqrt{2\pi}}\Phi^\Ai_s(\lambda)\begin{pmatrix}
1 \\ 0
\end{pmatrix},\qquad
\mathbf h(\lambda;s)=-\frac 1{\sqrt{2\pi}}\,\Phi^\Ai_s(\lambda)^{-\top}\begin{pmatrix}
0 \\ 1
\end{pmatrix},
\ee
and from condition (b) in the RH problem for $\Phi^\Ai$.
Finally, combining conditions (c) in the RH problems for $Y$ and $\Phi^\Ai$, we obtain that as $\lambda\to\infty$ we have
\be
Y_\s(\lambda;s)\Phi_s^\Ai(\lambda)=\left(I+\lambda^{-1}\begin{pmatrix}
\beta(s) & \i\bigl(\eta(s)+\frac 7{48}\bigr) \\ \i\alpha(s) & -\beta(s)\end{pmatrix}+O(\lambda^{-2})\right)(\lambda+s)^{\frac 14\s_3}G\e^{-\frac 23(\lambda+s)^{\frac 32}\s_3}C_\delta
\ee
and expanding for~$\lambda$ large and~$s$ fixed we verify condition~(d) in the RH problem for $\Psi_\s$ along with the claimed relation \eqref{eq:pa}.
Finally, \eqref{eq:Lkernel} follows directly from the expression~\eqref{eq:kernelresolvent} for the kernel of $(1-\mathcal M_\s\mathcal K_s^\Ai)^{-1}-1=\mathcal M_\s\mathcal L_s^\s$, along with the identities~\eqref{eq:PsiEmptyformula} and~\eqref{eq:fgPhiAi}
\end{proof}
 
\begin{proposition}
\label{prop:PsiV}
The RH problem for~$\Psi_\s$ has a unique solution for all~$s\in\R$ and all $\underline \nu=(\nu_1,\dots,\nu_m)$ with $\nu_i\not=\nu_j$ for~$i\not=j$, which can be expressed as
\be
\label{eq:PsiVformula}
\Psi_\s(\lambda;s|\underline \nu)=M(\lambda;s|\underline \nu) \Psi_\s(\lambda;s|\emptyset),
\ee
where $M$ is a rational function of $\ll$, with poles at $\lambda=\nu_1,\dots,\nu_m$ only, given by
\be
\label{eq:M}
M(\lambda;s|\underline \nu)= I-\frac{1}{2\pi\i}\sum_{i,j=1}^{m}\frac{\bigl(L_s^\s(\underline \nu,\underline \nu)^{-1}\bigr)_{j,i}}{\lambda-\nu_j}\Psi_\s(\nu_i;s|\emptyset)\begin{pmatrix} 0 & 1 \\ 0 & 0\end{pmatrix} \Psi_\s^{-1}(\nu_j;s|\emptyset),
\ee
where we use the notation~\eqref{eq:defLHmatrix}.
\end{proposition}
\begin{proof}
By the conditions in the RH problem for $\Psi_\s$ it is straightforward to verify that $M(\lambda;s|\underline \nu):=\Psi_\s(\lambda;s|\underline \nu)\Psi_\s(\lambda;s|\emptyset)^{-1}$ is a rational matrix with simple poles at $\nu_1,\dots,\nu_m$ only and $M(\lambda;s|\underline \nu)\to I$ as $\lambda\to\infty$.
Hence we write
\be
\label{eq:formM}
M(\lambda;s|\underline \nu) = I +\sum_{j=1}^m\frac{M_j(s|\underline \nu)}{\lambda-\nu_j}.
\ee
Condition (c) in the RH problem for $\Psi_\s$ then translates to the condition
\be
\label{eq:conditionMj}
M(\lambda;s|\underline \nu)\Psi_\s(\lambda;s|\emptyset)(\lambda-\nu_j)^{-\sigma_3}=O(1),\quad\mbox{as }\lambda\to \nu_j,\qquad j=1,\dots,m,
\ee
and we claim that this condition uniquely determines the coefficients $M_j(s|\underline \nu)$. 
Indeed, the expansion at $\lambda\to \nu_j$ of the left-hand side of \eqref{eq:conditionMj} gives
\be
\biggl(I+\frac{M_j(s|\underline \nu)}{\lambda-\nu_j} +
\! \sum_{\begin{smallmatrix}1\leq i\leq m \\ i\neq j\end{smallmatrix}}\!\frac{M_i(s|\underline \nu)}{\nu_j-\nu_i} 
+ O(\lambda-\nu_j)\biggr)
\biggl(\Psi_\s(\nu_j;s|\emptyset) + \Psi'_\s(\nu_j;s|\emptyset)(\lambda-\nu_j)+O((\lambda-\nu_j)^2)\biggr)(\lambda-\nu_j)^{-\s_3}
\ee
where $\Psi'_\s(\lambda;s|\underline \nu):=\pa_\lambda\Psi_\s(\lambda;s|\underline \nu)$.
Vanishing of singular terms in this Laurent series yields
\begin{align}
\label{eq:Mjdef1}
M_j(s|\underline \nu)\Psi_\s(\nu_j;s|\emptyset)\begin{pmatrix}
1\\0
\end{pmatrix}&=\begin{pmatrix}
0\\0
\end{pmatrix},
\\
\label{eq:Mjdef2}
\biggl(\biggl(I + \sum_{\begin{smallmatrix}1\leq i\leq m \\ i\neq j\end{smallmatrix}}\frac{M_i(s|\underline \nu)}{\nu_j-\nu_i}\biggr) \Psi_\s(\nu_j;s|\emptyset) + M_j(s|\underline \nu) \Psi'_\s(\nu_j;s|\emptyset)\biggr)\begin{pmatrix}
1\\0
\end{pmatrix}
&= \begin{pmatrix}
0\\0
\end{pmatrix}.
\end{align}
Equation \eqref{eq:Mjdef1} means that the first column of $M_j(s|\underline \nu)\Psi_\s(\nu_j;s|\emptyset)$ vanishes. Denoting by $\mathbf a_j=\mathbf a_j(s|\underline \nu)\in\C^2$ the second column of $M_j(s|\underline \nu)\Psi_\s(\nu_j;s|\emptyset)$, we get
\be
\label{eq:Mjfunctfj}
M_j(s|\underline \nu)=\mathbf a_j (0,1) \Psi^{-1}_\s(\nu_j;s|\emptyset).
\ee
Plugging \eqref{eq:Mjfunctfj} into \eqref{eq:Mjdef2} using \eqref{eq:Lkernel} we get
\begin{align}\nonumber
\Psi_\s(\nu_j;s|\emptyset) \begin{pmatrix}
	1\\0
\end{pmatrix}+\sum_{\begin{smallmatrix}1\leq i\leq m \\ i\neq j\end{smallmatrix}}\mathbf a_i\underbrace{\frac{(0,1)\Psi^{-1}_\s(\nu_i;s|\emptyset)\Psi_\s(\nu_j;s|\emptyset)\begin{pmatrix}
	1\\0
\end{pmatrix}}{\nu_j-\nu_i}}_{=\,2\pi \i\bigl(L_s^\s(\underline \nu,\underline \nu)\bigr)_{j,i}}
\\
+\,\mathbf a_j\underbrace{(0,1)\Psi^{-1}_\s(\nu_j;s|\emptyset) \Psi'_\s(\nu_j;s|\emptyset) \begin{pmatrix}
	1\\0
\end{pmatrix}}_{=\,2\pi \i\bigl(L_s^\s(\underline \nu,\underline \nu)\bigr)_{j,j}}
=&\begin{pmatrix}
	0\\0
\end{pmatrix}
\end{align}
and so, cf. Proposition~\ref{prop:sec2},
	\be
	\Psi_\s(\nu_j;s|\emptyset) \begin{pmatrix}
		1\\0
	\end{pmatrix}+2\pi\i \sum_{i=1}^{m}\mathbf a_i\bigl(L_s^\s(\underline \nu,\underline \nu)\bigr)_{j,i} = \begin{pmatrix}
		0\\0
	\end{pmatrix}
	\ \Rightarrow\
	\mathbf a_j=-\frac{1}{2\pi \i}\sum_{i=1}^{m}\bigl(L_s^\s(\underline \nu,\underline \nu)^{-1}\bigr)_{j,i}\Psi_\s(\nu_i;s|\emptyset) \begin{pmatrix}
		1\\0
	\end{pmatrix},
	\ee
and {by \eqref{eq:Mjfunctfj}, we finally get} \eqref{eq:M}. 
\end{proof}

\subsection{Stark equation}

It is convenient to introduce the following {variant} of $\Psi_\s$, namely
\be
\label{eq:Theta}
\Theta_\s(\lambda;s|\underline \nu):=
\begin{pmatrix}
	1 & p_\s(s|\underline \nu) \\ 0 & 1
\end{pmatrix}\e^{\frac{\i\pi}4\s_3}\Psi_\s(\lambda;s|\underline \nu)\e^{-\frac{\i\pi}4\s_3}.
\ee
The RH conditions on $\Psi_\s$ imply that $\Theta_\s$ is the unique solution to the following RH problem.

\subsubsection*{RH problem for $\Theta_\s$}
\begin{itemize}
	\item[(a)] $\Theta_\s(\cdot;s|\underline \nu): \C\setminus \R \to \C^{2\times 2}$ is analytic for all $s\in\R$ and all finite $\underline \nu\subset\R$. 
	\item[(b)] The boundary values of~$\Theta_\s(\cdot;s|\underline \nu)$ are continuous on $\mathbb{R}\setminus\underline \nu$ and are related by
	\be
	\label{eq:ThetaJump}
	\Theta_{\s,+}(\lambda;s|\underline \nu) = \Theta_{\s,-}(\lambda;s|\underline \nu) \begin{pmatrix}
	1 & \i(1-\sigma(\lambda)) \\
	0 & 1
	\end{pmatrix},  \qquad \ll \in \R, \ \ll\not=\nu_1,\dots,\nu_m.
	\ee
	\item[(c)] For all $i=1,\dots,m$, as $\lambda \to \nu_i$ from either side of the real axis we have 
	\be
	\label{eq:ThetaasympV}
	\Theta_\s(\lambda;s|\underline \nu) (\lambda-\nu_i)^{-\sigma_3}=O(1).
	\ee
	\item[(d)] As~$\lambda \to \infty$, we have
	\be
	\label{eq:Thetaasympinfty}
	\Theta_\s(\lambda;s|\underline \nu)= \begin{pmatrix}1 & p \\ 0 & 1
	\end{pmatrix} \left( I + \lambda^{-1}\begin{pmatrix}q & -r \\ p & -q \end{pmatrix} +O(\lambda^{-2}) \right) \lambda^{\frac 14\s_3}\frac{\begin{pmatrix} 1 & 1 \\ -1 & 1
\end{pmatrix}}{\sqrt 2}
	\e^{\left(-\frac{2}{3}\lambda^{\frac 32}-s\lambda^{\frac 12}\right)\s_3}\wh C_\delta
	\ee
	for any $\delta \in (0,\frac{\pi}{2})$; here $p=p_\s(s|\underline \nu)$, $q=q_\s(s|\underline \nu)$, and $r=r_\s(s|\underline \nu)$ are the same as in~\eqref{eq:Psiasympinfty}, $\wh C_\delta:=\e^{\frac{\i\pi}4\s_3}C_\delta\e^{-\frac{\i\pi}4\s_3}$ where $C_\delta$ is defined in \eqref{eq:G}, and the branches of $\lambda^{\frac 14\s_3}$ and $\lambda^{\frac 12}$ are taken as in~\eqref{eq:Psiasympinfty}.
\end{itemize}

The formula~\eqref{eq:Lkernel} is equivalent to
\be
\label{eq:LkernelTheta}
L_s^\s(\lambda,\mu)=\frac{\left(\Theta_\s(\mu;s|\emptyset)^{-1}\Theta_\s(\lambda;s|\emptyset)\right)_{2,1}}{2\pi(\lambda-\mu)}.
\ee

\begin{proposition}
\label{prop:equationins}
For any $\lambda\in\mathbb C\setminus\mathbb R$ and for any $\underline \nu=(\nu_1,\dots,\nu_m)$ with $\nu_i\not=\nu_j$ for~$i\not=j$, $\Theta_\s(\lambda;s|\underline \nu)$ is differentiable in $s$, and 
\be
\label{eq:dsTheta}
\pa_s\Theta_\s(\lambda;s|\underline \nu)=\begin{pmatrix} 0 & \lambda+2\pa_s p_\s(s|\underline \nu) \\ 1 & 0
\end{pmatrix}\Theta_\s(\lambda;s|\underline \nu)\,,
\ee
where $p_\s(s|\underline\nu)$ appears in~\eqref{eq:Psiasympinfty}.
\end{proposition}
\begin{proof}
The differentiability of $\Theta_\s$ in $s$, and the fact \eqref{eq:Thetaasympinfty} continues to hold after differentiating formally in $s$, can be proved using standard techniques from RH theory, and we refer the reader to \cite[Section 3]{CafassoClaeysRuzza} for details.
The matrix function $A(\lambda;s|\underline \nu):=\pa_s\Theta_\s(\lambda;s|\underline\nu)\Theta_\s(\lambda;s|\underline\nu)^{-1}$ is entire in $\lambda$; indeed it has no jump across the real axis and no singularities at $\underline \nu$ because of the RH conditions (b) and (c) for $\Theta_\s$.
Moreover, condition (d) in the RH problem for $\Theta_\s$ implies that
\be
\label{eq:expA}
A(\lambda;s|\underline \nu) = \begin{pmatrix} 0 & \lambda+p^2+2q+\pa_sp \\ 1 & 0
\end{pmatrix}+\lambda^{-1} \begin{pmatrix} \star & \star \\ -p^2-2q+\pa_sp & \star\end{pmatrix}+ O(\lambda^{-2}),\qquad\lambda\to\infty,
\ee
where $p=p_\s(s|\underline \nu)$ and $q=q_\s(s|\underline \nu)$ are as in \eqref{eq:Thetaasympinfty} and $\star$ denote expressions which are not relevant to us now.
Since $A(\lambda;s|\underline \nu)$ is entire, {Liouville's theorem implies that $A(\lambda;s|\underline \nu)$ coincides with the linear and constant terms in the} Laurent series \eqref{eq:expA} and that higher order terms vanish. This yields
\be
\label{eq:idpq}
p_\s(s|\underline \nu)^2+2q_\s(s|\underline \nu)=\pa_sp_\s(s|\underline \nu),
\ee
and the proof is complete.
\end{proof}

From equation \eqref{eq:dsTheta} it follows that
\be
\label{eq:ThetaStructure}
\Theta_\s(\lambda;s|\underline \nu)=-\sqrt{2\pi}\begin{pmatrix}
\pa_s\varphi_\s(\lambda;s|\underline \nu) & \pa_s\chi_\s(\lambda;s|\underline \nu)
\\	
\varphi_\s(\lambda;s|\underline \nu) & \chi_\s(\lambda;s|\underline \nu)
\end{pmatrix},
\ee
where either $f=\varphi_\s(\lambda;s|\underline \nu)$ or $f=\chi_\s(\lambda;s|\underline \nu)$ solves
\be
\label{eq: Schroedinger phi,varphi}
\biggl(\pa_s^2-2\bigl(\pa_sp_\s(s|\underline\nu)\bigr)\biggr)f=\lambda f
\ee
{which will yield the \emph{Stark equation}~\eqref{StarkModified}.}

\begin{proposition}
We have
\be
\label{eq:dLs}
\pa_sL_s^\s(\ll,\mu)=-\varphi_\s(\lambda;s|\emptyset)\varphi_\s(\mu;s|\emptyset)
\ee
\end{proposition}
\begin{proof}
We use~\eqref{eq:LkernelTheta} to compute
\begin{align}
\nonumber
\pa_s L_s^\s(\lambda,\mu)&=
\tr\pa_s\biggl(\frac{\Theta_\s(\lambda;s|\emptyset){\rm E}_{12}\Theta_\s(\mu;s|\emptyset)^{-1}}{2\pi(\lambda-\mu)}\biggr)
\\
\nonumber
&=\tr\frac{\bigl(A(\lambda;s|\emptyset)-A(\mu;s|\emptyset)\bigr)\Theta_\s(\lambda;s|\emptyset){\rm E}_{12}\Theta_\s(\mu;s|\emptyset)^{-1}}{2\pi(\lambda-\mu)}
\\
\label{eq:prooffinal}
&=\tr\frac{{\rm E}_{12}\Theta_\s(\lambda;s|\emptyset){\rm E}_{12}\Theta_\s(\mu;s|\emptyset)^{-1}}{2\pi}\,,
\end{align}
where we used the cyclic property of the trace and Proposition~\ref{prop:equationins} and we denoted
\be
\label{eq:notationE12}
{\rm E}_{12}:=\begin{pmatrix}
0 & 1\\ 0 & 0
\end{pmatrix},\qquad
A(\lambda;s|\emptyset):=\begin{pmatrix} 0 & \lambda+2\pa_sp_\s(s|\emptyset) \\ 1 & 0
\end{pmatrix}.
\ee
Finally, it suffices to insert \eqref{eq:ThetaStructure} into \eqref{eq:prooffinal}.
\end{proof}

We can finally characterize the J\'anossy densities in terms of the RH problem for~$\Psi_\s$.

\begin{proposition}
\label{prop:dlogJV}
For all $s\in\R$ and all finite sets $\underline \nu=(\nu_1,\dots,\nu_m)$ with $\nu_i\not=\nu_j$ for all $i\not=j$, we have
\be
\pa_s\log j_\sigma(s|\underline \nu) =\frac{s^2}4-p_\s(s|\underline \nu)
\ee
where $p_\s(s|\underline \nu)$ appears in \eqref{eq:Psiasympinfty}.
\end{proposition}
\begin{proof}
Using Proposition~\ref{prop:PsiV} we get
\be
\label{meh}
\i p_\s(s|\underline \nu)-\i p_\s(s|\emptyset)=\bigl(M^1_\infty(s|\underline \nu)\bigr)_{2,1}
\ee
where $M(\lambda;s|\underline \nu)=I+\lambda^{-1}M^1_\infty(s|\underline \nu)+O(\lambda^{-2})$ as $\lambda\to\infty$.
Using~\eqref{eq:M} we compute
\be
M^1_\infty(s|\underline \nu)=-\frac{1}{2\pi \i}\sum_{i,j=1}^{m}\bigl(L_s^\s(\underline \nu,\underline \nu)^{-1}\bigr)_{j,i}\Psi_\s(\nu_i;s|\emptyset)\begin{pmatrix} 0 & 1 \\ 0 & 0\end{pmatrix} \Psi^{-1}_\s(\nu_j;s|\emptyset)
\ee
and so, using~\eqref{eq:Theta} and~\eqref{eq:ThetaStructure}, we obtain
\be
\bigl(M^1_\infty(s|\underline \nu)\bigr)_{2,1}=\i \sum_{i,j=1}^{m}\bigl(L_s^\s(\underline \nu,\underline \nu)^{-1}\bigr)_{j,i}\varphi_\s(\nu_i;s|\emptyset)\varphi_\s(\nu_j;s|\emptyset)
\ee
On the other hand, by~\eqref{eq:dLs} we have
\be
\pa_s\log\det L_s^\s(\underline \nu,\underline \nu)=\sum_{i,j=1}^{m}\bigl(L_s^\s(\underline \nu,\underline \nu)^{-1}\bigr)_{j,i}\frac{\pa\bigl(L_s^\s(\underline \nu,\underline \nu)\bigr)_{i,j}}{\pa s}=-\sum_{i,j=1}^{m}\bigl(L_s^\s(\underline \nu,\underline \nu)^{-1}\bigr)_{j,i}\varphi_\s(\nu_i;s|\emptyset)\varphi_\s(\nu_j;s|\emptyset)
\ee
and the proof now follows from~\eqref{meh} because $\log j_\s(s|\underline \nu)=\log\det L_s^\s(\underline \nu,\underline \nu)+\log j_\s(s|\emptyset)$ by~\eqref{eq:idJanossy2} and because 
\be
\label{eq:dlogJintermsofp}
\frac\pa{\pa s}\log j_\s(s|\emptyset)=\frac{s^2}4-p_\s(s|\emptyset)
\ee
by~\eqref{eq:dLogJ} and \eqref{eq:pa}.
\end{proof}

\subsection{Asymptotics as \texorpdfstring{$s\to+\infty$}{s goes to plus infinity}}\label{sec:asymptoticssplusinfty}

The jump matrix of condition~(b) in the RH problem for $Y$ can be rewritten, thanks to~\eqref{eq:fgPhiAi}, as 
\be
\label{eq:Yjump with PhiAi}
I-2\pi\i\,\mathbf f (\ll;s)\mathbf h^\top(\ll;s)=I-\Phi_s^\Ai(\ll)\begin{pmatrix}0 & \sigma(\ll) \\ 0 & 0
\end{pmatrix}\Phi_s^\Ai(\ll)^{-1}.
\ee
We now show that this jump matrix is close to the identity in the appropriate norms in order to apply the standard {\it small-norm} RH theory~\cite{Its}.
To this end, we introduce the following notation for a measurable matrix-valued function $X:\R\to\C^{m\times n}$:
\be
\label{pnorm}
\|X\|_p:=\begin{cases}
\displaystyle\max_{1\leq i\leq m,\ 1\leq j\leq n}\left\lbrace\left(\int_\R|X_{i,j}(\mu)|^p\,\d\mu\right)^{\frac 1p}\right\rbrace, &p\in[1,+\infty), 
\\[15pt]
\displaystyle\max_{1\leq i\leq m,\ 1\leq j\leq n}\left\lbrace\esssup_{\mu\in\R}|X_{i,j}(\mu)|\right\rbrace, & p=\infty.
\end{cases}
\ee
\begin{lemma}
\label{lemma:old}
Let $\s$ satisfy Assumption~\ref{assumption:weak}. Then, with the same $\kappa>0$ as in Assumption~\ref{assumption:weak}, we have
\be
\left\|\Phi_s^\Ai\begin{pmatrix}0 & \sigma \\ 0 & 0
\end{pmatrix}(\Phi_s^\Ai)^{-1}\right\|_p=
O(s^{-\kappa}),\quad\mbox{as }s\to+\infty,\qquad p=1,2,\infty.
\ee
\end{lemma}

\begin{proof}
The entries in $\Phi_s^\Ai(\ll)\begin{pmatrix}0 & \sigma(\ll) \\ 0 & 0
\end{pmatrix}\Phi_s^\Ai(\ll)^{-1}$ are (possibly, up to a sign) $\s(\ll)\mathcal B(\ll+s)$ where $\mathcal B(\ll)$ is one of the functions $\Ai(\ll)^2$, $\Ai(\ll)\Ai'(\ll)$, or $\Ai'(\ll)^2$.
By Assumption~\ref{assumption:weak}, there are $\Lambda,C_1>0$ such that $\lambda<-\Lambda$ implies $\s(\ll)<C_1|\ll|^{-\frac 32-\kappa}$. Assuming $s\geq 2\Lambda$, we have, using standard asymptotic properties of~$\Ai$ and $\Ai'$:
\begin{itemize}
\item for $\ll\geq -s/2$, $\s(\ll)\leq 1$ and $|\mathcal B(\ll+s)|\leq C_2\exp(-\ll-s)$ for some $C_2>0$,
\item for $\ll\leq -s/2$, $\s(\ll)\leq C_1|\ll|^{-\frac 32-\kappa}$ and $|\mathcal B(\ll+s)|\leq C_3|\ll|^{\frac 12}$ for some $C_3>0$.
\end{itemize}
Therefore, as~$s\to+\infty$ we get:
\begin{align}
\|\s(\ll)\mathcal B(\ll+s)\|_1&\leq\int_{-\infty}^{-s/2}C_1C_3|\lambda|^{-\kappa-1}\d\ll+\int_{-s/2}^{+\infty}C_2\e^{-\ll-s}\d\ll=O(s^{-\kappa}),\\
\|\s(\ll)\mathcal B(\ll+s)\|_2^2&\leq\int_{-\infty}^{-s/2}(C_1C_3)^2|\lambda|^{-2\kappa-2}\d\ll+\int_{-s/2}^{+\infty}C_2^2\e^{-\ll-s}\d\ll=O(s^{-2\kappa-1}),\\
\|\s(\ll)\mathcal B(\ll+s)\|_\infty&\leq \max\left\lbrace C_1C_3\left(\frac s2\right)^{-\kappa-1},C_2\e^{-s/2}\right\rbrace=O(s^{-\kappa-1}),
\end{align}
and the lemma is proved.
\end{proof}

\begin{proposition}
\label{prop:asympwavefunction}
If $\s$ satisfies Assumption~\ref{assumption:weak}, we have
\be
\label{eq:phiasymplarges}
\varphi_\s(\lambda;s|\emptyset)=\Ai(\lambda+s)\bigl(1+O(s^{-\frac 12-\kappa})\bigr),\
\pa_s\varphi_\s(\lambda;s|\emptyset)=\Ai'(\lambda+s)\bigl(1+O(s^{-\frac 12-\kappa})\bigr),\quad s\to+\infty,
\ee
for all $\lambda\in\R$.
\end{proposition}
\begin{proof}
Using equation \eqref{eq:Yjump with PhiAi}, we can rewrite condition (b) in the RH problem for $Y$ as
	\be
	Y_{\s,+}(\lambda;s)-Y_{\s,-}(\lambda;s)=Y_{\s,-}(\lambda;s)\Phi_s^\Ai(\lambda)\begin{pmatrix}0 & -\sigma(\lambda) \\ 0 & 0
	\end{pmatrix}\Phi_s^\Ai(\lambda)^{-1},\qquad \lambda\in\R,
	\ee
and so standard RH theory implies that we can write
\begin{align}
\nonumber
Y_\s(\lambda;s)
&=I-\frac 1{2\pi\i}\int_\R (Y_{\s,-}(\mu;s)-I)\Phi^\Ai_s(\mu)\begin{pmatrix}
0 & \s(\mu) \\ 0 & 0
\end{pmatrix}\Phi_s^\Ai(\mu)^{-1}\frac {\d \mu}{\mu-\lambda}
\\
\label{YSP}
&\qquad\qquad\qquad\qquad
-\frac 1{2\pi\i}\int_\R \Phi^\Ai_s(\mu)\begin{pmatrix}
0 & \s(\mu) \\ 0 & 0
\end{pmatrix}\Phi_s^\Ai(\mu)^{-1}\frac {\d \mu}{\mu-\lambda}.
\end{align}
Moreover, by Lemma~\ref{lemma:old} and standard small-norm RH theory~\cite{Its}, $Y_-(\cdot;s)-I$ is in $L^2$ (entry-wise) for all $s\in\R$ and satisfies
\be
\|Y_{\s,-}(\cdot;s)-I\|_2=O(s^{-\kappa}),\qquad s\to+\infty.
\ee
By~\eqref{eq:defPhi}, \eqref{eq:PsiEmptyformula}, \eqref{eq:Theta}, and~\eqref{eq:ThetaStructure}, we have 
\begin{align}
\nonumber
&
\begin{pmatrix}
1 & -\alpha_\sigma(s) \\ 0 & \i
\end{pmatrix}
\begin{pmatrix}
\pa_s\varphi_\s(\lambda;s|\emptyset) \\ \varphi_\s(\lambda;s|\emptyset)
\end{pmatrix}-
\begin{pmatrix}\Ai'(\lambda+s) \\ \i\,\Ai(\lambda+s)\end{pmatrix}
\\
&=
-\frac 1{\sqrt{2\pi}}\int_\R (Y_{\s,-}(\mu)-I)\Phi^\Ai_s(\mu)\begin{pmatrix} 1 \\ 0 \end{pmatrix}\sigma(\mu)K_s^\Ai(\lambda,\mu)\d\mu
-\frac 1{\sqrt{2\pi}}\int_\R \Phi^\Ai_s(\mu)\begin{pmatrix} 1 \\ 0 \end{pmatrix}\sigma(\mu)K_s^\Ai(\lambda,\mu)\d\mu
\end{align}
where $\alpha_\sigma(s)$ is given in~\eqref{eq:Yasympinfty}, cf.~\eqref{eq:pa}.
Hence, for some $C>0$ and $s$ sufficiently large, we have
\be
\label{lastinequality}
\biggl|\frac{\varphi_\s(\lambda;s|\emptyset)}{\Ai(\lambda+s)}-1\biggr|\leq C\left(\left\|Y_{\s,-}-I\right\|_2\frac{\left\|\Phi_s^\Ai\begin{pmatrix}1\\ 0\end{pmatrix}\s K_s^\Ai(\cdot,\lambda)\right\|_2}{\Ai(\lambda+s)}+\frac{\left\|\Phi_s^\Ai\begin{pmatrix}1\\ 0\end{pmatrix}\s K_s^\Ai(\cdot,\lambda)\right\|_1}{\Ai(\lambda+s)}\right)\,,
\ee
and
\begin{align}
\nonumber
&\biggl|\frac{\pa_s\varphi_\s(\lambda;s|\emptyset)-\alpha_\sigma(s)\varphi_\sigma(\lambda;s|\emptyset)}{\Ai'(\lambda+s)}-1\biggr|
\leq 2\,\biggl|\frac{\pa_s\varphi_\s(\lambda;s|\emptyset)-\alpha_\sigma(s)\varphi_\sigma(\lambda;s|\emptyset)-\Ai'(\lambda+s)}{\sqrt s\,\Ai(\lambda+s)}\biggr| 
\\
\label{lastlastinequality}
&\qquad\leq 2\,C\, s^{-\frac 12}\left(\left\|Y_{\s,-}-I\right\|_2\frac{\left\|\Phi_s^\Ai\begin{pmatrix}1\\ 0\end{pmatrix}\s K_s^\Ai(\cdot,\lambda)\right\|_2}{\Ai(\lambda+s)}+\frac{\left\|\Phi_s^\Ai\begin{pmatrix}1\\ 0\end{pmatrix}\s K_s^\Ai(\cdot,\lambda)\right\|_1}{\Ai(\lambda+s)}\right)\,.
\end{align}
where we noted that $\sqrt s\,\Ai(\lambda+s)\leq 2| \Ai'(\lambda+s)|$ for $s$ sufficiently large.
Therefore, denoting $\mathcal A$ either $\Ai$ or $\Ai'$, we need to estimate the $L^p(\R,\d\mu)$-norm (for $p=1,2$)  of
\be
a(\mu):=\mathcal A(\mu+s)\sigma(\mu)\frac{K_s^\Ai(\ll,\mu)}{\Ai(\ll+s)}
\ee
as $s\to+\infty$, for fixed $\lambda$.
We can assume $s$ is sufficiently large such that $s>2|\ll|$ and $\Ai(\ll+s)\leq|\Ai'(\ll+s)|$.
\begin{itemize}
\item When $\mu\leq -s/2$, we have $|\mathcal A(\mu+s)|=O(|\mu|^{\frac 14})$, $\sigma(\mu)=O(|\mu|^{-\frac 32-\kappa})$, and
\be
\label{eq:estimate1}
\left|\frac{K^\Ai_s(\ll,\mu)}{\Ai(\ll+s)}\right|\leq\frac{|\Ai'(\ll+s)|}{\Ai(\ll+s)}\frac{(|\Ai(\mu+s)|+|\Ai'(\mu+s)|)}{|\lambda-\mu|}=O(s^{-\frac 12}|\mu|^{\frac 14})
\ee
hence
\be
a(\mu)=O(|\mu|^{-1-\kappa}s^{-\frac 12}).
\ee
\item When $\mu\geq -s/2$, we have $|\mathcal A(\mu+s)|=O(\e^{-\mu-s})$, $\sigma(\mu)=O(1)$, and
\be
\label{eq:estimate2}
\left|\frac{K^\Ai_s(\ll,\mu)}{\Ai(\ll+s)}\right|\leq\int_s^{+\infty}\frac{\Ai(\ll+\eta)}{\Ai(\ll+s)}\Ai(\mu+\eta)\d\eta\leq
\int_s^{+\infty}\Ai(\mu+\eta)\d\eta=O(\e^{-\mu-s}).
\ee
\end{itemize}
Hence $\|a\|_{L^1(\R)}=O(s^{-\frac 12-\kappa})$ and $\|a\|_{L^2(\R)}=O(s^{-1-\kappa})$, so that resuming from~\eqref{lastinequality} and~\eqref{lastlastinequality}, we get
\be
\biggl|\frac{\varphi_\s(\lambda;s|\emptyset)}{\Ai(\lambda+s)}-1\biggr|=O(s^{-\frac 12-\kappa}),
\quad
\biggl|\frac{\pa_s\varphi_\s(\lambda;s|\emptyset)-\alpha_\sigma\varphi_\s(\lambda;s|\emptyset)}{\Ai'(\lambda+s)}-1\biggr|=O(s^{-1-\kappa}),
\ee
proving the first asymptotic relation in the statement, {while the second one then follows from the triangular inequality
\be
\biggl|\frac{\pa_s\varphi_\s(\lambda;s|\emptyset)}{\Ai'(\lambda+s)}-1\biggr|\leq\biggl|\frac{\pa_s\varphi_\s(\lambda;s|\emptyset)-\alpha_\sigma(s)\varphi_\sigma(\lambda;s|\emptyset)}{\Ai'(\lambda+s)}-1\biggr|+\biggl|\frac{\alpha_\sigma(s)\varphi_\s(\lambda;s|\emptyset)}{\Ai'(\lambda+s)}\biggr|
\ee
along with the estimates $\alpha_\sigma(s)=O(s^{-\kappa})$ (stemming from Lemma~\ref{lemma:old} and standard small-norm RH theory) and $\varphi_\s(\lambda;s|\emptyset)/\Ai'(\lambda+s)=O\bigl(\Ai(\lambda+s)/\Ai'(\lambda+s)\bigr)=O(s^{-\frac 12})$ as $s\to+\infty$.}
\end{proof}

\begin{corollary}
\label{corollary:LkernelIntegral}
We have 
\be
\label{eq:LkernelIntegral}
L_s^\s(\lambda,\mu)=\int_s^{+\infty} \varphi_\s(\lambda;r|\emptyset)\varphi_\s(\mu;r|\emptyset)\d r.
\ee
\end{corollary}
\begin{proof}
Follows directly by integrating \eqref{eq:dLs} from $s$ to $+\infty$ thanks to \eqref{eq:phiasymplarges}.
\end{proof}

\begin{corollary}
\label{cor:uniqueness}
The function $\varphi(\lambda;s)=\varphi_\s(\lambda;s|\emptyset)$ is the unique solution to the boundary value problem~\eqref{eq:Stark}.
\end{corollary}
\begin{proof}
Let $\omega(\lambda;s)$ also satisfy
\be
\bigl(\pa_s^2+2v_\s(s)-s\bigr)\omega(\lambda;s)=\ll\,\omega(\ll;s),\qquad\omega(\ll;s)=\bigl(1+o(1)\bigr)\Ai(\ll+s),\ \ s\to+\infty.
\ee
Then, $\varphi(\lambda;s)\,\pa_s\bigl(\omega(\lambda;s)\bigr)-\omega(\lambda;s)\,\pa_s\bigl(\varphi(\lambda;s)\bigr)=c$ for some $c=c(\lambda)$ not depending on $s$.
Assume $c\not=0$ for the sake of obtaining a contradiction.
Since $\Ai(\lambda+s)$ is nonzero for $s$ sufficiently large, we get
\be
\label{eq:temporarylabel}
\bigl(1+o(1)\bigr)\pa_s\bigl(\omega(\lambda;s)\bigr)-\bigl(1+o(1)\bigr)\pa_s\bigl(\varphi(\lambda;s|\emptyset)\bigr)=\frac{c}{\Ai(\lambda+s)},\quad s\to+\infty.
\ee
By Proposition~\ref{prop:asympwavefunction}, we obtain
\be
\pa_s\bigl(\omega(\lambda;s)\bigr)=\bigl(1+o(1)\bigr)\Ai'(\lambda+s)+\bigl(1+o(1)\bigr)\frac{c}{\Ai(\lambda+s)} =\bigl(1+o(1)\bigr)\frac{c}{\Ai(\lambda+s)},\quad s\to+\infty,
\ee
which can be integrated between $s$ and $2s$ to give
\be
\bigl(1+o(1)\bigr)c\int_s^{2s}\frac{\d s'}{\Ai(\lambda+s')}=\omega(\lambda;2s)-\omega(\lambda;s)=\bigl(1+o(1)\bigr)\Ai(\lambda+2s)-\bigl(1+o(1)\bigr)\Ai(\lambda+s),\quad s\to+\infty.
\ee
By standard asymptotic properties of the Airy function this is a contradiction, proving that $c=0$ whence the desired uniqueness $\omega(\lambda;s)=\varphi(\lambda;s)$.
\end{proof}

\subsection{Proofs of Theorems~\ref{thm:Janossyphi} and~\ref{thm:Janossyphi2}}\label{sec:proofs12}

\begin{proof}[Proof of Theorem~\ref{thm:Janossyphi}]
The first relation~\eqref{thm1eq1} is {nothing else than} \eqref{eq:idJanossy2}.

The first equality in~\eqref{thm1eq2} is~\eqref{eq:LkernelIntegral} while the second one is a rewriting of~\eqref{eq:LkernelTheta} using~\eqref{eq:ThetaStructure}.

That $\varphi_\s$ solves the Stark boundary value problem~\eqref{eq:Stark} with potential $v_\s(s|\emptyset):=\pa_s^2\log j_\s(s|\emptyset)$ follows from~\eqref{eq: Schroedinger phi,varphi},~\eqref{eq:dlogJintermsofp}, and~\eqref{eq:phiasymplarges}.

Finally, in order to prove \eqref{thm1eq3}, we first consider the following chain of equalities, where we use \eqref{eq:dlogJintermsofp} and the asymptotics as $s\to+\infty$ of Section~\ref{sec:asymptoticssplusinfty}:
\begin{align}
\nonumber
\log j_\s(s|\emptyset) = -\int_s^{+\infty}\pa_r\log j_\s(r|\emptyset)\d r=\int_s^{+\infty}(r-s)\pa_r^2\log j_\s(r|\emptyset)\d r
=\int_s^{+\infty}(r-s)v_\s(r|\emptyset)\d r.
\end{align}
In the first step we use $\lim_{r\to+\infty}j_\s(r|\emptyset)=\lim_{r\to+\infty}\det_{L^2(\R)}(1-\mathcal K_s^\s)=1$ because $\mathcal K_s^\s$ converges to the zero operator in trace-norm when $s\to+\infty$.
Indeed, $\mathcal K_s^\s$ is a non-negative trace-class operator with (jointly) continuous integral kernel so that its trace-norm is
\be
\int_\R\s(\ll)K_s^\Ai(\ll,\ll)\d\ll=\int_{-\infty}^{-s/2}\underbrace{\s(\ll)K_s^\Ai(\ll,\ll)}_{O(|\ll|^{-1-\kappa})}\d\ll+\int_{-s/2}^{+\infty}\underbrace{\s(\ll)K_s^\Ai(\ll,\ll)}_{O(\exp(-\ll-s))}=O(s^{-\kappa})
\ee
as $s\to+\infty$; here we use that as $s\to+\infty$ we have $\sigma(\ll)=O(\ll^{-\frac 32-\kappa})$ and $K_s^\Ai(\ll,\ll)=O(|\ll|^{\frac 12})$ for $\ll<-s/2$, and $\s(\ll)=O(1)$ and $K_s^\Ai(\ll,\ll)=O(\exp(-\ll-s))$ for $\ll>-s/2$ (cf. Assumption~\ref{assumption:weak}).

The identity $\int\varphi^2(\lambda;s|\emptyset)\d\s(\lambda)=-v_\s(s|\emptyset)$, proved in \cite[Proposition~4.1]{CafassoClaeysRuzza}, completes the proof.
This identity also follows by setting $\underline \nu=\emptyset$ in the more general identity~\eqref{eq:moregeneral} below, which will be shown  in the proof of Theorem~\ref{thm:Janossyphi2} by an adaptation of the argument in loc. cit. (and not relying on the case $\underline\nu=\emptyset$).
\end{proof}

\begin{proof}[Proof of Theorem~\ref{thm:Janossyphi2}]
Let us introduce
\be
\Xi(\ll;s|\underline\nu):=\Theta_\s(\ll;s|\underline\nu)\xi(\lambda|\underline\nu)^{-\s_3},\qquad
\xi(\lambda|\underline\nu):=\prod_{i=1}^m(\lambda-\nu_i).
\ee
As it follows from conditions~(b) and~(c) in the RH problem for~$\Theta_\s$, $\Xi(\ll;s|\underline\nu)$ is a sectionally analytic matrix-valued function of~$\ll$ satisfying a jump condition across the real axis of the form
\be
\label{eq:whThetaJump}
\Xi_{+}(\lambda;s|\underline \nu) =\Xi_{-}(\lambda;s|\underline \nu) \begin{pmatrix}
1 & \i(1-\sigma(\lambda))\xi(\lambda|\underline\nu)^2 \\
0 & 1
\end{pmatrix}, \qquad \lambda \in \mathbb{R}.
\ee
It follows that $C(\lambda;s|\underline\nu):=\bigl(\pa_\lambda\Xi(\lambda;s|\underline \nu)\bigr)\Xi(\lambda;s|\underline\nu)^{-1}$  is also a sectionally analytic matrix-valued function of~$\ll$ satisfying a jump condition across the real axis of the form
\be
C_+(\lambda;s|\underline\nu)-C_-(\lambda;s|\underline\nu)=\Xi(\lambda;s|\underline\nu)\begin{pmatrix}
0 & \i\xi(\lambda|\underline\nu)^2\bigl((1-\s(\lambda))2\pa_\lambda\log\xi(\lambda|\underline\nu)-\s'(\ll)\bigr) \\ 0 & 0
\end{pmatrix}\Xi(\ll;s|\underline\nu)^{-1},
\ee
for all~$\ll\in\R$.
In the right-hand side of this equation we omit the choice of boundary values for $\Xi$ as the expression is independent from this choice, as it can be shown by~\eqref{eq:whThetaJump}.
It therefore follows {from} a contour deformation argument that
\be
\label{deformationargument}
\int_{\R}\Xi(\lambda;s|\underline\nu)\biggl(\begin{smallmatrix}
0 & \i\xi(\lambda|\underline\nu)^2\bigl((1-\s(\lambda))2\pa_\lambda\log\xi(\lambda|\underline\nu)-\s'(\ll)\bigr) \\ 0 & 0
\end{smallmatrix}\biggr)\Xi(\ll;s|\underline\nu)^{-1}\frac{\d\ll}{2\pi\i}=
\lim_{R\to+\infty}\oint_{c_R}C(\ll;s|\underline\nu)\frac{\d\ll}{2\pi\i}\,,
\ee
where $c_R$ is the {\it clock-wise} oriented circle $|\ll|=R$.
By the identity
\be
C:=\bigl(\pa_\ll\Xi\bigr)\Xi^{-1}=\bigl(\pa_\ll\Theta_\s\bigr)\Theta^{-1}_\s-(\pa_\ll\log\xi)\Theta_\s\s_3\Theta^{-1}_\s
\ee
and the asymptotic relation~\eqref{eq:Thetaasympinfty} we obtain that, as $\ll\to\infty$ uniformly in the complex plane, we have
\be
\bigl(C(\ll;s|\underline\nu)\bigr)_{2,1}=1+\ll^{-1}\left(\frac s2-\pa_s p_\s(s|\underline\nu)\right)+O(\ll^{-\frac 32}).
\ee
Here we also use~\eqref{eq:idpq}.
Taking the $(2,1)$-entry of~\eqref{deformationargument}, we obtain, also using Proposition~\ref{prop:dlogJV},
\be
\label{eq:moregeneral}
-\int_\R\varphi_\s(\ll;s|\underline\nu)^2\left((1-\s(\lambda))2\pa_\lambda\log\xi(\lambda|\underline\nu)-\s'(\ll)\right)\d\ll=\pa_s p_\s(s|\underline\nu)-\frac s2=-\pa_s^2\log j_\s(s|\underline\nu).
\ee
Taking into account that $\xi(\lambda|\underline\nu)=\prod_{i=1}^m(\ll-\nu_i)$, the identity~\eqref{eq:moregeneral} we just proved is~\eqref{eq:secondlogderJV}.

Next, the Stark equation~\eqref{StarkModified} follows from~\eqref{eq: Schroedinger phi,varphi} and Proposition~\ref{prop:dlogJV}.

Finally, extracting the $(2,1)$-entry in~\eqref{eq:PsiVformula}, using~\eqref{eq:M} and~\eqref{eq:ThetaStructure},
\begin{align}
\nonumber
\varphi_\s(\ll;s|\underline\nu)&=\left(1-\sum_{i,j=1}^m\frac{\bigl(L_s^\s(\underline\nu,\underline\nu)^{-1}\bigr)_{j,i}}{\ll-\nu_j}\varphi_\s(\nu_i;s|\emptyset)\pa_s\varphi_\s(\nu_j;s|\emptyset)\right)\varphi_\s(\ll;s|\emptyset)
\\
\nonumber
&\qquad+
\left(\sum_{i,j=1}^m\frac{\bigl(L_s^\s(\underline\nu,\underline\nu)^{-1}\bigr)_{j,i}}{\ll-\nu_j}\varphi_\s(\nu_i;s|\emptyset)\varphi_\s(\nu_j;s|\emptyset)\right)
\pa_s\varphi_\s(\ll;s|\emptyset)
\\
\displaybreak
\nonumber
&=
\varphi_\s(\ll;s|\emptyset)-\sum_{i,j=1}^m\bigl(L_s^\s(\underline\nu,\underline\nu)^{-1}\bigr)_{j,i}\varphi_\s(\nu_i;s|\emptyset) L_s^\s(\ll,\nu_i)
\\
&=
\frac 1{\det L_s^\s(\underline\nu,\underline\nu)}
\det\left(\begin{array}{cccc}
\varphi_\s(\lambda;s|\emptyset) & L_s^\s(\lambda,\nu_1) & \cdots & L_s^\s(\lambda,\nu_m) \\ 
\varphi_\s(\nu_1,s|\emptyset) & L_s^\s(\nu_1,\nu_1) & \cdots & L_s^\s(\nu_1,\nu_m) \\
\vdots & \vdots & \ddots & \vdots \\
\varphi_\s(\nu_m,s|\emptyset) & L_s^\s(\nu_m,\nu_1) & \cdots & L_s^\s(\nu_m,\nu_m) 
\end{array}\right)\,,
\end{align}
where in the second step we use~\eqref{thm1eq2} and in the third a standard manipulation of the determinant of a block matrix.
Therefore~\eqref{ModifiedVarphiDet} holds true and the proof is complete.
\end{proof}

\subsection{Comparison with inverse scattering for the Stark operator}\label{sec:IS}

We now comment on the connection between our probabilistic construction based on the $\sigma$-thinned (shifted) Airy process and the classical inverse scattering problem for the Stark operator, as described in~\cite{ItsSukhanov}, see also \cite{Santini1, Santini2}.
The latter can be formulated through the following RH problem, cf.~\cite[Definition~2.3]{ItsSukhanov}.

\subsubsection*{RH problem for $M$}
\begin{itemize}
	\item[(a)] $M(\cdot;\xi): \mathbb{C}\setminus \mathbb{R} \to \mathbb{C}^{2\times 2}$ is analytic for all $\xi\in\R$.
	\item[(b)] The boundary values of $M(\cdot;\xi)$ are continuous on $\mathbb{R}$ and are related by
	\be
	\label{eq:jumpM}
M_{+}(\mu;\xi) = \begin{pmatrix} 0 &-s(\mu) \\ \overline{s(\mu)}&1\end{pmatrix}
M_{-}(\mu;\xi), \qquad \mu \in \R.
	\ee
	\item[(c)] As $\mu\to \infty$, we have
\be
M(\mu;\xi)=M_\infty(\mu;\xi)\bigl(I+o(1)\bigr),\quad
M_\infty(\xi;\mu):=\begin{cases}
\begin{pmatrix}
-w_0(\xi-\mu) & -w_0'(\xi-\mu) \\
w_1(\xi-\mu) &w_1'(\xi-\mu)
\end{pmatrix},&\Im\mu>0,
	\\
	\begin{pmatrix}
w_2(\xi-\mu) & w_2'(\xi-\mu) \\
w_0(\xi-\mu) &w_0'(\xi-\mu)
\end{pmatrix},&\Im\mu<0,
\end{cases}
\ee
where
\be
w_0(k):=2\i\sqrt\pi\Ai(k),\quad
w_1(k):=2\sqrt\pi\e^{\frac{\pi\i}6}\Ai(\e^{\frac{2\pi\i}3}k),\quad
w_2(k):=2\sqrt\pi\e^{-\frac{\pi\i}6}\Ai(\e^{-\frac{2\pi\i}3}k).
\ee
\end{itemize}

In~\eqref{eq:jumpM}, $s(\mu)=a(\mu)\bigl(\overline{a(\mu)}\bigr)^{-1}$ (for $\mu\in\R$), and $a(\mu)$ is part of the scattering data for the Stark operator.
In particular, cf.~\cite[Theorem~2.2]{ItsSukhanov}, $a(\mu)$ is analytic and nonzero in the half-plane $\Im \mu<0$ and $a(\mu)=1+o(|\mu|^{-\frac 12})$ as $\mu\to\infty$ within $\Im\mu\leq 0$.
Then, the matrix
\be
\label{eq:purported}
\Psi(\lambda;s):=\frac 1{\sqrt 2}\begin{pmatrix}0&-\i\\1&0
\end{pmatrix}
\times
\begin{cases}
M^\top(-\lambda;s)\left(\begin{array}{cc}
0& a(-\lambda)^{-1}\\ -a(-\lambda)&0 
\end{array}\right),&\Im\lambda>0,
\\
M^\top(-\lambda;s)\left(\begin{array}{cc}
\overline{a(-\overline\lambda)} & 0\\ 0 & \left(\overline{a(-\overline\lambda)}\right)^{-1}
\end{array}\right),&\Im\lambda<0,
\end{cases}
\ee
essentially solves the RH problem for $\Psi_\s$, for $\sigma(\lambda)=1-|a(-\lambda)|^{-2}$ and $m=0$, with the caveat that it only satisfies a slightly weaker normalization at~$\lambda=\infty$ in which the sub-leading term is just $o(1)$ rather than $O(\lambda^{-1})$.

Indeed, the expression in the right-hand side of~\eqref{eq:purported} is analytic for $\lambda\in\C\setminus\R$ by the above mentioned properties of $a$, and a direct computation suffices to ascertain that
\be
\Phi^\Ai(\ll+s)=\frac 1{\sqrt 2}\begin{pmatrix}0&-\i\\1&0
\end{pmatrix}M^\top_\infty(-\lambda;s)\times\begin{cases}
\begin{pmatrix}
0&1\\-1&0
\end{pmatrix}
,&\Im\lambda>0,
\\
I,&\Im\lambda<0,
\end{cases}
\ee
such that the normalization at~$\infty$ of the two RH problems match (up to the order of the sub-leading contribution, as we already mentioned).
Moreover, a direct computation shows that the jump condition of the RH problem for $\Psi_\s$ is satisfied by the right-hand side of~\eqref{eq:purported}.

There is however an essential difference between our assumptions on~$\s$, and the assumptions in~\cite{ItsSukhanov}.
Whereas we consider functions~$\s$ converging to~$0$ at~$-\infty$, but not necessarily converging to~$0$ at~$+\infty$, cf.~Assumption~\ref{assumption:weak}, it is required in~\cite[Theorem~2.2(c)]{ItsSukhanov} that $\sigma(\lambda)=1-|a(-\lambda)|^{-2}\to 0$ as~$\lambda\to\pm\infty$.
Hence, the class of functions~$\sigma$ that we consider, is not included in the class of scattering data considered by classical inverse scattering theory for the Stark operator.

\subsection{Connection with the theory of Schlesinger transformations\texorpdfstring{~\cite{JimboMiwa,Bertola,BertolaCafasso}.}{.}}\label{sec:BC}
It is also worth to make a comparison of our setting with the general theory of Schlesinger transformations.
For a general RH problem depending on parameters, one can define the \emph{Malgrange--Bertola differential} on the space of parameters~\cite{Bertola}.
The general definition, applied to the RH problem for $\Gamma(\lambda;s):=\Psi_\sigma(\lambda-s;s|\emptyset)$, specializes to the following one-form in $s$:
\be
\Omega = \omega(s)\d s,\quad\omega(s) := \int_\R\tr\biggl[\Gamma^{-1}_-(\lambda;s)\frac{\d\Gamma_-(\lambda;s)}{\d\lambda}\frac{\d J_\G(\lambda;s)}{\d s}J^{-1}_\G(\lambda;s)\biggr]\frac{\d\lambda}{2\pi\i},
\ee
where $J_\G(\lambda;s)= \begin{pmatrix} 1 & 1-\sigma(\lambda-s) \\ 0 & 1\end{pmatrix}$.
Using the form of $J_\G$ and the jump condition $\Gamma_+(\lambda;s)=\Gamma_-(\lambda;s)J_\Gamma(\lambda;s)$, the integrand can be rewritten as
\be
\tr\biggl[\G^{-1}_-(\lambda;s)\frac{\d\G_-(\lambda;s)}{\d\lambda}\begin{pmatrix} 0 & \sigma'(\lambda-s) \\ 0 & 0 \end{pmatrix}\biggr] = \frac 12\tr\biggl(\G_-^{-1}(\lambda;s)\frac{\d^2\G_-(\lambda;s)}{\d\lambda^2}-\G_+^{-1}(\lambda;s)\frac{\d^2\G_+(\lambda;s)}{\d\lambda^2}\biggr)
\ee
and hence a residue computation gives
\be
\omega(s)=-\frac 12\res{\lambda=\infty}\tr\biggl(\G^{-1}(\lambda;s)\frac{\d^2\G(\lambda;s)}{\d\lambda^2}\biggr) = \frac{s^2}4-p_\s(s|\emptyset)=\pa_s\log j_\s(s|\emptyset).
\ee
The logarithmic potential $j_\sigma$ of $\Omega$ is then termed tau function of the RH problem~\cite{Bertola}.
Note that a tau function in this sense is defined only up to a multiplicative (integration) constant.
Accordingly, we can say that the Fredholm determinant $j_\s(s|\emptyset)$ is the tau function associated with the RH problem for $\Gamma(\lambda;s)$.

Pole insertion in a RH problem (\emph{Schlesinger transformation}) and its effect on $\Omega$ have been studied in depth in~\cite{BertolaCafasso} (expanding on \cite{JimboMiwa,Bertola}) for RH problems with identity normalization at infinity and insertion of poles off the jump contour.
The general results of op.~cit. formally match with our setting.
Namely, in our setting we consider the RH problem for $\G(\ll;s|\underline\nu):=\Psi(\ll-s;s|\underline\nu)$ which is obtained from the one for $\Gamma$ by inserting poles at $\nu_i+s$ such that $\Gamma(\ll;s|\underline\nu)=(\ll-\nu_i-s)^{-\s_3}O(1)$ for $i=1,\dots,m$.
The \emph{characteristic matrix} of~\cite[Definition~2.2]{BertolaCafasso}, such that the logarithmic differential of its determinant expresses the variation between the Malgrange--Bertola differential after pole insertion~\cite[Theorem~2.2 part~(3)]{BertolaCafasso}, would reduce in the present setting to
\be
\res{\lambda=\nu_i+s}\res{\mu=\nu_j+s}\frac{\bigl(\Gamma^{-1}(\ll;s)\G(\mu;s)\bigr)_{2,1}}{(\ll-\mu)(\ll-\nu_i-s)(\mu-\nu_j-s)}\d\ll=-2\pi\i \, L_s^\s(\nu_i,\nu_j),\qquad i,j=1,\dots,m.
\ee
Hence, the above-mentioned \cite[Theorem~2.2 part~(3)]{BertolaCafasso} predicts that the tau function associated with the RH problem for $\G(\ll;s|\underline\nu)$ is (within an absolute multiplicative constant) $j_\s(s|\emptyset)$ times the determinant of $L_s^\s(\underline\nu,\underline\nu)$, i.e. $j_\s(s|\underline\nu)$ by~\eqref{eq:idJanossy1}, as we showed in Proposition~\ref{prop:dlogJV}.

\subsection{Isospectral deformation and cKdV: proof of Theorem~\ref{theorem:cKdV}}\label{sec:proofcKdV}

As explained in Section \ref{sec:intro}, the connection with the cKdV equation is made by studying the shifted and dilated Airy kernel
\be
K_{X,T}^\Ai(\lambda,\mu):=T^{-\frac 13} K^\Ai\biggl(T^{-\frac 13}(\lambda+X), T^{-\frac 13}(\mu+X)\biggr),
\ee
where $X\in\R,\ T\geq 0$ are parameters.
The corresponding J\'anossy density $J_\s(X,T|\underline\nu)$, defined in~\eqref{def:Jsigma}, is recovered from $j_\s(s|\underline\nu)$ by~\eqref{eq:Janossyreduction}.
In view of Proposition~\ref{prop:dlogJV}, we have
\be
\label{reformulate}
\pa_X\log J_\s(X,T|\underline\nu)=T^{-\frac 13}\left.\pa_s\log j_{\wt\s}(s|T^{-\frac 13}\underline\nu)\right|_{s=XT^{-\frac 13}}=\frac {X^2}{4T}-T^{-\frac 13}p_{\wt\s}(XT^{-\frac 13};T^{-\frac 13}\underline\nu).
\ee
Throughout this section we use the notation $\wt\s(\ll)=\s(T^{\frac 13}\ll)$, as in~\eqref{eq:Janossyreduction}.
{It is straightforward to verify by the RH problem for $\Psi_\s$ that the matrix
\be
\wh\Psi_\s(\ll;X,T|\underline\nu):=T^{\frac 1{12}\s_3}\Psi_{\wt\s}(\ll T^{-\frac 13};XT^{-\frac 13}| T^{-\frac 13}\underline\nu)
\ee
is the (unique) solution to the following RH problem.}

\subsubsection*{RH problem for $\wh\Psi_\s$}
\label{RHPsi X,T}
\begin{itemize}
	\item[(a)] $\wh\Psi_\s(\cdot;X,T|\underline \nu): \mathbb{C}\setminus \mathbb{R} \to \mathbb{C}^{2\times 2}$ is analytic for all~$X\in\R$, $T>0$, and all $\underline \nu$. 
	\item[(b)] The boundary values of $\wh\Psi_\s(\cdot;X,T|\underline \nu)$ are continuous on $\mathbb{R}\setminus\{\nu_1,\dots,\nu_m\}$ and are related by
	\be
	\wh\Psi_{\s,+}(\ll;X,T|\underline \nu) = \wh\Psi_{\s,-}(\ll;X,T|\underline \nu) \begin{pmatrix}
		1 & 1-\s(\ll) \\
		0 & 1
	\end{pmatrix}, \qquad \ll \in \R, \ \ll\not=\nu_i.
	\ee
 	\item[(c)] For all $i=1,\dots,m$, as $\lambda \to \nu_i$ from either side of the real axis we have 
	\be
	\label{eq:PsiXTasympV}
	\wh\Psi_\s(\lambda;X,T|\underline \nu) (\lambda-\nu_i)^{-\sigma_3}=O(1).
	\ee
	\item[(d)] As $\lambda \to \infty$, we have
	\be
	\label{eq:PsiXTasympinfty}
	\wh\Psi_\s(\lambda;X,T|\underline \nu) = \biggl( I + \frac 1\ll\begin{pmatrix}
\wh q_\s(X,T|\underline\nu) & \i \wh r_\s(X,T|\underline\nu) \\	
\i\wh p_\s(X,T|\underline\nu) & -\wh q_\s(X,T|\underline\nu)	
	\end{pmatrix}+O(\lambda^{-2}) \biggr) \lambda^{\frac 14\s_3}G\e^{-T^{-\frac 12}\left(\frac{2}{3}\lambda^{\frac 32}+X\lambda^{\frac 12}\right)\sigma_{3}} C_\delta
	\ee
for any $\delta \in (0,\frac{\pi}{2})$.
Here we take principal branches of the roots of $\ll$ as explained after~\eqref{eq:Psiasympinfty}, and $G,C_\delta$ are as in~\eqref{eq:G}. 
Moreover, the coefficients in the sub-leading term are related to the ones in~\eqref{eq:Psiasympinfty} by
\begin{align}
\nonumber
\wh q_\s(X,T|\underline\nu)&=T^{\frac 13}q_{\wt\s}(XT^{-\frac 13}|T^{-\frac 13}\underline \nu),\qquad
\wh r_\s(X,T|\underline\nu)=T^{\frac 12} r_{\wt\s}(XT^{-\frac 13}|T^{-\frac 13}\underline \nu),\\
\label{eq:oldandnewp}
\wh p_\s(X,T|\underline\nu)&=T^{\frac 16}p_{\wt\s}(XT^{-\frac 13}|T^{-\frac 13}\underline \nu).
\end{align}
\end{itemize}

It is convenient to reformulate~\eqref{reformulate} using~\eqref{eq:oldandnewp}, as
\be
\label{pdlogjXT}
\pa_X\log J_\s(X,T|\underline\nu)=\frac{X^2}{4T}-T^{-\frac 12}\wh p_\s(X,T|\underline\nu).
\ee
Introduce now, cf.~\eqref{eq:Theta} and \eqref{eq:ThetaStructure},
\begin{align}
\nonumber
\wh\Theta_\s(\ll;X,T|\underline\nu):={}&T^{\frac 1{12}\s_3}\Theta_{\wt\s}(\ll T^{-\frac 13};XT^{-\frac 13}|T^{-\frac 13}\underline\nu) 
\\
\nonumber
={}&\begin{pmatrix}
1 & \wh p_\s(X,T|\underline\nu) \\ 0 & 1
\end{pmatrix}
\e^{\frac{\i\pi}4\s_3}\wh\Psi_\s(\lambda;X,T|\underline \nu)\e^{-\frac{\i\pi}4\s_3}
\\
\label{eq:whTheta}
={}&-\sqrt{2\pi}
\begin{pmatrix}
T^{\frac 12}\pa_X\wh\varphi_\s(\ll;X,T|\underline\nu)&T^{\frac 12}\pa_X\wh\chi_\s(\ll;X,T|\underline\nu)
 \\ \wh\varphi_\s(\ll;X,T|\underline\nu)&\wh\chi_\s(\ll;X,T|\underline\nu)
\end{pmatrix}
\end{align}
where we define, cf.~\eqref{eq:ThetaStructure},
\be
\label{eq:newwave}
\wh\varphi_\s(\ll;X,T|\underline\nu)=T^{-\frac 1{12}}\varphi_{\wt\s}(\ll T^{-\frac 13};XT^{-\frac 13}|T^{-\frac 13}\underline\nu),\quad
\wh\chi_\s(\ll;X,T|\underline\nu)=T^{-\frac 1{12}}\chi_{\wt\s}(\ll T^{-\frac 13};XT^{-\frac 13}|T^{-\frac 13}\underline\nu).
\ee

\begin{proposition}
Let $V_\s(X,T|\underline\nu):=\pa_X^2\log J_\s(X,T|\underline\nu)$ and $f$ be in the linear span of $\wh\varphi_\s(\ll;X,T|\underline\nu)$ and $\wh\chi_\s(\ll;X,T|\underline\nu)$. 
We have the ``Lax pair''
\be
\label{eq:Lax}
\mathscr Lf=\ll f,\qquad \mathscr Af=\pa_Tf,
\ee
where
\be
\mathscr L:=T\pa_X^2+2TV_\s(X,T|\underline\nu)-X,\quad
\mathscr A:=-\frac 13\pa_X^3-V_\s(X,T|\underline\nu)\pa_X-\frac 12\pa_XV_\s(X,T|\underline\nu)+\frac 13T^{-1}-\frac 1{12}T^{-\frac 32}.
\ee
\end{proposition}
\begin{proof}
Although the first equation $\mathscr L f=\ll f$ follows directly from~\eqref{eq: Schroedinger phi,varphi} by using~\eqref{eq:newwave}, it is convenient to deduce it again; doing so will provide us with additional information useful in the derivation of the second equation.

We start by noting that the matrix function $A(\ll;X,T|\underline\nu):=\bigl(\pa_X\wh\Theta_\s(\ll;X,T|\underline\nu)\bigr)\wh\Theta_\s(\ll;X,T|\underline\nu)^{-1}$ has no jump across the real axis because the jump condition for $\wh\Theta_\s$ across the real axis does not depend on $X$ as if follows from~\eqref{eq:whTheta} along with condition~(b) in the RH problem for~$\wh\Psi_\s$.
Once more, we refer the reader to \cite[Section 3]{CafassoClaeysRuzza} for the rigorous justification of the differentiability of the RH solution.

Moreover, it is readily checked, cf.~\eqref{eq:Thetaasympinfty}, that as~$\ll\to\infty$ we have
\begin{align}
\nonumber
\Theta_\s(\lambda;s|\underline \nu)&= \begin{pmatrix}1 & \wh p_\s(X,T|\underline\nu) \\ 0 & 1
	\end{pmatrix}\biggl( I +\frac 1 \lambda\begin{pmatrix}\wh q_\s(X,T|\underline\nu) & -\wh r_\s(X,T|\underline\nu) \\ \wh p_\s(X,T|\underline\nu) & -\wh q_\s(X,T|\underline\nu) \end{pmatrix}
	\\
	\label{eq:asympwhTheta}
	&\qquad+\frac 1{\lambda^2}\begin{pmatrix}\star & \star \\ \wh n_\s(X,T|\underline\nu) & \star \end{pmatrix}+O(\lambda^{-3}) \biggr)
	\lambda^{\sigma_{3}/4}\frac{\begin{pmatrix} 1 & 1 \\ -1 & 1
\end{pmatrix}}{\sqrt 2}	
	\e^{-T^{-\frac 12}\left(\frac{2}{3}\lambda^{\frac 32}+X\lambda^{\frac 12}\right)\sigma_{3}} C_\delta
	\end{align}
for any $\delta \in (0,\frac{\pi}{2})$.
The notations are as in~\eqref{eq:Thetaasympinfty} but now we also need to explicitly record the $(2,1)$-entry of the second sub-leading term in the asymptotic series, denoted $\wh n_\s$.

Next, from~\eqref{eq:asympwhTheta}, we deduce that $A$ has an expansion for large $\ll$ of the form
\begin{align}
\nonumber
A&=T^{-\frac 12}\begin{pmatrix}
0& \ll+\wh p_\s^2+2\wh q_\s+T^{\frac 12}\pa_X\wh p_\s\\
1 & 0 
\end{pmatrix}
\\
&\quad
+\ll^{-1}T^{-\frac 12}\begin{pmatrix}
\star & \star
\\
T^{\frac 12}\pa_X\wh p_\s-\wh p_\s^2-2\wh q_\s&
\wh n_\s+\wh r_\s+\wh p_\s^3+3\wh p_\s\wh q_\s-\frac 12T^{\frac 12}\pa_X\left(\wh p^2_\s+2\wh q_\s\right)  
\end{pmatrix}+O(\ll^{-2}).
\end{align}
Liouville's theorem guarantees then that $A$ is a polynomial in $\ll$.
Consequently, the higher-order Laurent coefficient in this expansion must vanish; we do not need the information coming from the first row (and we have accordingly omitted these terms), while from the second row at order $\ll^{-1}$ we obtain
\be
\label{eq:useful}
\wh p_\s^2+2\wh q_\s=T^{\frac 12}\pa_X\wh p_\s,\qquad
\wh n_\s+\wh r_\s+\wh p_\s^3+3\wh p_\s\wh q_\s=\frac 12T\pa_X^2\wh p_\s.
\ee
Summarizing, also thanks to~\eqref{pdlogjXT}, we have
\be
A=T^{-\frac 12}\begin{pmatrix}
0& \ll+2T^{\frac 12}\pa_X\wh p_\s\\
1 & 0 
\end{pmatrix}=
T^{-\frac 12}\begin{pmatrix}
0& \ll-2V_\s+X\\
1 & 0 
\end{pmatrix}.
\ee
Comparing with~\eqref{eq:whTheta} we obtain $\mathscr L f=\ll f$ whenever $f$ is in the linear span of $\wh\varphi_\s,\wh\chi_\s$.

Next, the matrix function $B(\ll;X,T|\underline\nu):=\bigl(\pa_T\wh\Theta_\s(\ll;X,T|\underline\nu)\bigr)\wh\Theta_\s(\ll;X,T|\underline\nu)^{-1}$ has no jump across the real axis because the jump condition for $\wh\Theta_\s$ across the real axis does not depend on $T$, and so $B$ is entire in $\lambda$. 
It then follows from an application of Liouville theorem that $B$ is a polynomial in $\ll$.
In particular, from~\eqref{eq:asympwhTheta}, the $(2,1)$-entry of $B$ is expressed as
\be
3T^{\frac 32}B_{2,1}=-\ll+\left(\wh p_\s^2+2 \wh q_\s-\frac{3X}2\right)=-\ll+T^{\frac 12}\pa_X\wh p_\s-\frac 32X=-\ll-T V_\s-X,
\ee
and, similarly, the $(2,2)$-entry of $B$ as
\be
3T^{\frac 32}B_{2,2}=-\left(\wh n_\s+\wh r_\s+\wh p_\s^3+3\wh p_\s\wh q_\s\right)=-\frac 12T\pa_X^2\wh p_\s=\frac 12T^{\frac 32}\pa_X V_\s-\frac 14.
\ee
We have used~\eqref{eq:useful} and~\eqref{pdlogjXT} to simplify these expressions.
Comparing with~\eqref{eq:whTheta}, we must have
\be
\pa_Tf = -\frac{\ll+T V_\s+X}{3T}\pa_Xf+\frac{2\pa_X V_\s-T^{-\frac 32}}{12}f,
\ee
for $f$ equal to either $\wh\varphi_\s$ or $\wh\chi_\s$, and hence for any $f$ in their linear span.
By the relation $\mathscr L f=\ll f$ obtained above we can rewrite the last relation by using $\lambda \pa_Xf=\pa_X(\mathscr L f)$ which finally yields $\pa_Tf=\mathscr A f$.
\end{proof}
\begin{corollary}\label{corollarythmckdv}
The function $V_\s(X,T|\underline\nu):=\pa_X^2\log J_\s(X,T|\underline\nu)$ satisfies the cKdV equation~\eqref{cKdV}.
\end{corollary}
\begin{proof}
This is a classical argument~\cite{Lax}.
From the compatibility condition of~\eqref{eq:Lax} we obtain
\be
\label{compatibility}
\bigl(\pa_T\mathscr L+[\mathscr L,\mathscr A]\bigr)f=0.
\ee
A direct computation gives that $\pa_T\mathscr L+[\mathscr L,\mathscr A]$ is the operator of multiplication with the function
\be
\label{eq:ckdvproof}
V_\s(X,T|\underline\nu)+2T\pa_TV_\s(X,T|\underline\nu)+2T V_\s(X,T|\underline\nu)\pa_X V_\s(X,T|\underline\nu)+\frac 16T\pa_X^3V_\s(X,T|\underline\nu).
\ee
The equation~\eqref{compatibility} must be true for any $f$ in the linear span of $\wh\varphi_\s,\wh\chi_\s$.
Since $\det\wh\Theta_\s=1$ identically in all variables $\ll,X,T$, the functions $\wh\varphi_\s$ and $\wh\chi_\s$ never vanish simultaneously, hence~\eqref{eq:ckdvproof} must vanish identically.
\end{proof}

This proves Theorem~\ref{theorem:cKdV}.

\subsection{Generalization to discontinuous \texorpdfstring{$\sigma$'s}{sigma}}
\label{section:discont}

In this section we briefly explain how to extend the results to a broader class of functions $\sigma$, including in particular $\sigma=1_{(0,+\infty)}$.

\begin{assumption}
\label{assumption:veryweak}
The function $\s:\mathbb R\to[0,1]$ is such that $\sigma=\s_0+\sum_{j=1}^f w_j1_{(\xi_j,+\infty)}$ for some (finite) integer $f\geq 0$, some $w_1,\ldots,w_f>0$ and some $\xi_1,\ldots,\xi_f\in\R$, and a smooth function $\sigma_0$ such that $\sigma_0(\ll)=O(|\ll|^{-\frac 32-\kappa})$ as $\ll\to-\infty$ for some $\kappa>0$. 
\end{assumption}

These are the assumptions made in~\cite{CafassoClaeysRuzza}, to which we refer for more details, and they include the setting of~\cite{CD} which corresponds to the case $\s_0=0$.
Under these more general assumptions, the RH problems for~$Y_\s$ and $\Psi_\s$ have to be complemented with the condition that $Y_\s$ and $\Psi_\s$ have, at worst, logarithmic singularities at~$\xi_j$.

Theorem~\ref{thm:Janossyphi} holds true verbatim except for~\eqref{thm1eq3}, which is to be replaced by
\be
j_\s(s|\emptyset)=\exp\left(-\int_s^{+\infty}(r-s)\left(\int_\R\varphi_\s(\ll;r|\emptyset)^2\s_0'(\ll)\d\ll+\sum_{j=1}^f\Delta_j\varphi_\s(\xi_j;r|\emptyset)^2\right)\right)
\ee
where $\Delta_j:=w_j-w_{j-1}$ for $2\leq j\leq f$ and $\Delta_1:=w_1$.
This follows directly from~\cite[equations~(1.8) and~(1.26)]{CafassoClaeysRuzza}.

Moreover, Theorem~\ref{thm:Janossyphi2} holds true verbatim except for~\eqref{eq:secondlogderJV}, which is to be replaced by
\be
\pa_s^2\log j_\s(s|\underline \nu)=\int_\R\varphi_\s(\lambda;s|\underline \nu)^2\biggl(-\s'_0(\ll)+\sum_{i=1}^m\frac{2\bigl(1-\s(\ll)\bigr)}{\ll-\nu_i}\biggr)\d\ll-\sum_{j=1}^f\Delta_j\varphi_\s(\xi_j;s|\underline\nu)^2,
\ee
and Theorem~\ref{theorem:cKdV} holds true verbatim.
{These two generalizations are obtained by  studying the local behavior of $\Psi_\s$ near the logarithmic singularities at the points $\xi_j$, as is done in the end of the proof of~\cite[Proposition~4.1]{CafassoClaeysRuzza}, cf. equations~(4.5) and~(4.6) there.}

\section{Asymptotics}\label{sec:asymptotics}

\subsection{Outline}\label{section:41}

The goal of this section is to prove Theorem~\ref{thm:asy}. 

The proof of part (i) of Theorem \ref{thm:asy} will rely on elementary operator estimates, starting from the analogue of the factorization \eqref{eq:idJanossy1} in the cKdV variables,
\begin{equation}\label{eq:factorization1XT}
J_\s(X,T|\underline \nu)=\det\left(K_{X,T}^{\rm Ai}(\underline \nu,\underline \nu)\right)\det_{L^2(\R)}\left(1-\mathcal M_{\sqrt\s}\widehat{\mathcal H}_{X,T}^{\underline \nu}\mathcal M_{\sqrt\s}\right),
\end{equation}
where
$\widehat{\mathcal H}_{X,T}^{\Ai}$ is the integral operator with kernel, similarly to \eqref{eq:Palmkernel},
\be\label{eq:PalmXT}\wh H_{X,T}^{\underline \nu}(\ll,\mu)
:=\frac{\det K_{X,T}^{\Ai}\bigl((\ll,\underline \nu),(\mu,\underline \nu)\bigr)}{\det K_{X,T}^{\Ai}(\underline \nu,\underline \nu)}
=K_{X,T}^\Ai(\ll,\mu)-K_{X,T}^\Ai(\ll,\underline \nu)K_{X,T}^\Ai(\underline \nu,\underline \nu)^{-1} K_{X,T}^\Ai(\underline \nu,\mu).\ee
Consequently, by \eqref{eq:Tpotential}, we also have
\be
\label{eq:decompositionasy1}
V_\s(X,T|\underline\nu)=V_{\s=0}(X,T|\underline \nu)+\pa_X^2\log\det_{L^2(\R)}\left(1-\mathcal M_{\sqrt\s}\widehat{\mathcal H}_{X,T}^{\underline \nu}\mathcal M_{\sqrt\s}\right).
\ee
We will prove in Section~\ref{section:righttail} that the second factor in \eqref{eq:factorization1XT} is close to $1$ and that the second term in \eqref{eq:decompositionasy1} is close to $0$, and this will result in part (i) of Theorem \ref{thm:asy}.

\medskip

For part (ii) of Theorem \ref{thm:asy}, we will instead use the analogue of \eqref{eq:idJanossy2} in the cKdV variables $X,T$.
Using~\eqref{eq:Janossyreduction}, \eqref{thm1eq1}, and \eqref{thm1eq2}, we obtain the identity
\be
\label{eq:factorizationasy}
J_\s(X,T|\underline\nu)=\det\bigl(\wh L_{X,T}^\s(\nu_i,\nu_j)\bigr)_{i,j=1}^m\,J_\s(X,T|\emptyset)
\ee
where
\be
\label{eq:defLXT}
\wh L_{X,T}^\s(\ll,\mu) := T^{-\frac 13} L_{XT^{-\frac 13}}^{\wt\s}(T^{-\frac 13}\ll,T^{-\frac 13}\mu),
\ee
with $\wt\s(\ll)=\s(T^{\frac 13}\ll)$ as in~\eqref{eq:Janossyreduction}, which can be rewritten by~\eqref{thm1eq2} and~\eqref{eq:whTheta}--\eqref{eq:newwave} as
\begin{align}
\nonumber
\wh L_{X,T}^\s(\ll,\mu)
&=T^{\frac 12}\frac{\wh\varphi_\s(\ll;X,T|\emptyset)\pa_X\wh\varphi_\s(\mu;X,T|\emptyset)-\wh\varphi_\s(\mu;X,T|\emptyset)\pa_X\wh\varphi_\s(\ll;X,T|\emptyset)}{\ll-\mu}\\
&=\label{eq:kernelXTTheta}
\frac{\left(\wh\Theta_\s(\mu;X,T|\emptyset)^{-1}\wh\Theta_\s(\lambda;X,T|\emptyset)\right)_{2,1}}{2\pi(\lambda-\mu)}.
\end{align}
It also follows that
\be
\label{eq:decompositionasy}
V_\s(X,T|\underline\nu)=V_\s(X,T|\emptyset)+\pa_X^2\log\det\bigl(\wh L_{X,T}^\s(\nu_i,\nu_j)\bigr)_{i,j=1}^m.
\ee

The asymptotic behavior of the second factor in \eqref{eq:factorizationasy} and of the first term in \eqref{eq:decompositionasy} has been established in~\cite{CafassoClaeysRuzza,CharlierClaeysRuzza}, and can be summarized as follows in the cases where $XT^{-\frac 12}\to -\infty$.

\begin{theorem}\label{thm:CharlierClaeysRuzza}
Let $\s$ satisfy Assumption~\ref{assumption:strong}.
For any $T_0>0$ there exists $K>0$ such that
\begin{align}
\label{eq:asyempty3LogJ}
\log J_\s(X,T|\emptyset)&=\rho^3T^2\left(-\frac{4}{15}\left(1-\xi\right)^{\frac 52}+\frac 4{15}-\frac 23\xi+\frac 12\xi^2\right)+O\left(|X|^{\frac 32}T^{-\frac 12}\right),\\
\label{eq:asyempty3V}
V_\s(X,T|\emptyset)&=\rho\left(1-\sqrt{1-\xi}\,\right)+O\left(|X|^{-\frac 12}T^{-\frac 12}\right),
\end{align}
where $\rho:=c_+^2/\pi^2$ and $\xi:=X/(\rho T)$, uniformly for $X\leq -KT^{\frac 12}$ ant $T\geq T_0$.
\end{theorem}

This result is contained in~\cite[Theorem~1.3]{CharlierClaeysRuzza}.

Therefore, in order to prove {part (ii) of} Theorem~\ref{thm:asy} we only need to study, in Section~\ref{section:lefttail}, the additional contributions to~\eqref{eq:factorizationasy} and~\eqref{eq:decompositionasy} coming from $\det\bigl(\wh L_{X,T}^\s(\nu_i,\nu_j)\bigr)_{i,j=1}^m$.

\subsection{Right tail\texorpdfstring{: $XT^{-\frac 13}\to\infty$}{}}
\label{section:righttail}

We start with a Fredholm determinant estimate for the operator $\mathcal M_{\sqrt\s}\widehat{\mathcal H}_{X,T}^{\Ai}\mathcal M_{\sqrt\s}$, not only valid for large positive $X$, but also for complex large $X$ with $\arg X$ sufficiently small.

\begin{lemma}\label{lemma:H}
Let $\s$ satisfy Assumption~\ref{assumption:weak} and let $\underline\nu=(\nu_1,\dots,\nu_m)$ with $\nu_i\not=\nu_j$ for all $i\not=j$.
For any $T_0>0$ there exist $M,c,\delta>0$ such that
\be
\det_{L^2(\R)}\left(1-\mathcal M_{\sqrt\s}\widehat{\mathcal H}_{X,T}^{\underline\nu}\mathcal M_{\sqrt\s}\right)=1+O\bigl(\e^{-cXT^{-\frac 13}}\bigr),
\ee
uniformly for $|X|\geq MT^{\frac 13}$, $|\arg X|<\delta$, and $T\geq T_0$.
\end{lemma}
\begin{proof}
Using the integral representation for the Airy kernel in  \eqref{eq:standardairykernel}, \eqref{eq:shifteddilatedAirykernel}, and the asymptotic behavior for the Airy function, it is straightforward to verify that
\be
\left|K_{X,T}^{\rm Ai}(\lambda,\mu)\right|=O\left(|\lambda\mu|^{\frac 14}\e^{-cT^{-\frac 13}\Re(\lambda+X)_+}\e^{-cT^{-\frac 13}\Re(\mu+X)_+}\right),
\ee
uniformly for $|X|\geq MT^{\frac 13}$, $|\arg X|<\delta$, and $T\geq T_0$, where $R_+=\max\{R,0\}$, for any $c>0$, and uniformly for $\lambda,\mu\in\mathbb R$. 
Hence, by \eqref{eq:PalmXT},
\be
\left|\wh H_{X,T}^{\underline\nu}(\lambda,\mu)\right|=O\left(|\lambda\mu|^{\frac 14}\e^{-cT^{-\frac 13}\Re(\lambda+X)_+}\e^{-cT^{-\frac 13}\Re(\mu+X)_+}\right),
\ee
uniformly for the same values of $X,T,\lambda,\mu$.
We can now use the triangular inequality in the Fredholm series
\be
\det_{L^2(\R)}\left(1-\mathcal M_{\sqrt\s}\widehat{\mathcal H}_{X,T}^{\underline\nu}\mathcal M_{\sqrt\s}\right)-1=\sum_{n=1}^\infty\frac{(-1)^n}{n!}\int_{\mathbb R^n}\det\bigl(\wh H^{\underline\nu}_{X,T}(\lambda_i,\lambda_j)\bigr)_{i,j=1}^n \prod_{j=1}^n\s(\lambda_j)\d\lambda_j,
\ee
in order to obtain
\begin{align}
\nonumber
&\left|\det_{L^2(\R)}\left(1-\mathcal M_{\sqrt\s}\widehat{\mathcal H}_{X,T}^{\underline\nu}\mathcal M_{\sqrt\s}\right)-1\right|\\
\nonumber
&\qquad\qquad\leq \sum_{n=1}^\infty\frac{1}{n!}\int_{\mathbb R^n}\det\bigl(O(1)\bigr)_{i,j=1}^n\prod_{j=1}^n\e^{-2cT^{-\frac 13}\Re(\lambda_j+X)_+}|\lambda_j|^{\frac 12}\s(\lambda_j)\d\lambda_j\\
\nonumber
&\qquad\qquad=O\!\left(\sum_{n=1}^\infty\frac{n^{\frac n2}\e^{-ncT^{-\frac 13}\Re X}}{n!}\left(\int_{\mathbb R}\e^{-cT^{-\frac 13}\Re(\lambda+X)_+}|\lambda|^{\frac 12}\s(\lambda)\d\lambda\right)^n\right)
=O\!\left(\sum_{n=1}^\infty\frac{n^{\frac n2}}{n!}\xi^{n}\right)
\end{align}
where Hadamard's inequality guarantees that $\det\bigl(O(1)\bigr)_{i,j=1}^n=O(n^{\frac n2})$, and in the last step we set $\xi:=\e^{-cT^{-\frac 13}\Re X}\int_{\mathbb R}\e^{-cT^{-\frac 13}\Re(\lambda+X)_+}|\lambda|^{\frac 12}\s(\lambda)\d\lambda$.
Finally, $\sum_{n=1}^\infty\frac{n^{\frac n2}}{n!}\xi^{n}$ is a power series in~$\xi$ with infinite radius of convergence, hence $\sum_{n=1}^\infty\frac{n^{\frac n2}}{n!}\xi^{n}=O(\xi)$ when $\xi\to 0$; since
\be
|\xi|\leq\e^{-cT^{-\frac 13}\Re X}\int_{\mathbb R}|\lambda|^{\frac 12}\s(\lambda)\d\lambda,
\ee
and {since} the integral on the right-hand side is finite by Assumption~\ref{assumption:weak}, we have $\xi=O(\e^{-cT^{-\frac 13}\Re X})$ and the proof is complete.
\end{proof}

Next, we need to study the asymptotics for the finite determinant in~\eqref{eq:factorization1XT}.
To this end, it is convenient to first state and prove a general lemma.

\begin{lemma}
\label{generallemma}
Let $F:\R\to\C$ and $G:\R\to\C$ be smooth functions admitting asymptotic expansions
\be
\label{eq:expFG}
F(z)\sim z^{-\frac 14}\biggl(1+\sum_{i\geq 1}f_iz^{-\frac i2}\biggr),\quad G(z)\sim -z^{\frac 14}\biggl(1+\sum_{i\geq 1}g_iz^{-\frac i2}\biggr)
\ee
as $z\to+\infty$.
Let
\be
H(\lambda,\mu)=
\begin{cases}
\displaystyle\frac{F(\lambda)G(\mu)-G(\lambda)F(\mu)}{\lambda-\mu},&\mbox{if }\lambda\not=\mu,\\
\frac{\d F}{\d\lambda}(\lambda)G(\lambda)-\frac{\d G}{\d\lambda}(\lambda)F(\lambda),&\mbox{if }\lambda=\mu.
\end{cases}
\ee
Then, for all $m\geq 1$ and $\nu_1,\dots,\nu_m\in\R$ distinct,
\be
\label{eq:finalasymprel}
\det\bigl( H(\nu_i+s,\nu_j+s)\bigr)_{i,j=1}^m\sim 2^{-m(2m-1)}s^{-m^2}\prod_{1\leq i<j\leq m}(\nu_i-\nu_j)^2,\quad s\to+\infty.
\ee
If $F,G$ admit analytic extensions to a sector in the complex plane including the positive real half-line, and if the asymptotic expansions~\eqref{eq:expFG} hold true as $z\to\infty$ in this sector, then also the asymptotic relation~\eqref{eq:finalasymprel} holds true as $s\to\infty$ in this sector.
\end{lemma}
\begin{proof}
By the properties of $F$ and $G$, we have the asymptotic expansion
\be
H(\lambda+s,\mu+s)\sim\sum_{m\geq 2}s^{-\frac m2} h^{[m]}(\lambda,\mu),\qquad s\to+\infty.
\ee
Here, $h^{[m]}(\lambda,\mu)=\sum_{a,b\geq 0}c_{a,b}^{[m]}\lambda^a\mu^b$ are symmetric polynomials of $\lambda,\mu$, with coefficients $c_{a,b}^{[m]}$ which are polynomial functions of $f_i,g_i$.
The polynomials $h^{[m]}(\lambda,\mu)$ have degree $\lfloor\frac m2\rfloor-1$, i.e., $c_{a,b}^{[m]}=0$ when $\frac m2<a+b+1$.
Moreover, when $m$ is even, the top-degree part in $\lambda,\mu$ is independent of $f_i,g_i$, namely, for all $a,b\geq 0$ we have
\be
c_{a,b}^{[2(a+b+1)]}=c_{a,b}^{[2(a+b+1)]}\bigr|_{f_i=0,g_i=0}.
\ee
Therefore, we can write the expansion of the $m\times m$ matrix with entries $H(\nu_i+s,\nu_j+s)$ as
\be
\label{eq:expHinfinitematrices}
\bigl(H(\nu_i+s,\nu_j+s)\bigr)_{i,j=1}^m \sim \begin{pmatrix}
1 & \nu_1 & \nu_1^2 & \cdots \\
\vdots & \vdots & \vdots& \cdots \\
1 & \nu_m & \nu_m^2 & \cdots \\
\end{pmatrix}
\begin{pmatrix}
C_{0,0}(s) & C_{0,1}(s) & C_{0,2}(s) & \cdots \\
C_{1,0}(s) & C_{1,1}(s) & C_{1,2}(s) & \cdots \\
C_{2,0}(s) & C_{2,1}(s) & C_{2,2}(s) & \cdots \\
\vdots & \vdots & \vdots& \ddots
\end{pmatrix}
\begin{pmatrix}
1 & \cdots & 1 \\
\nu_1 & \cdots & \nu_m \\
\nu_1^2 & \cdots & \nu_m^2\\
\vdots & \vdots & \vdots
\end{pmatrix}
\ee
where we express the formal series representing the asymptotic expansion on the right as a product of matrices of size $m\times\infty$, $\infty\times\infty$, and $\infty\times m$, respectively, with the formal series $C_{a,b}(s)$ given by
\be
C_{a,b}(s)=\sum_{m\geq 2(a+b+1)}s^{-\frac{m}2}c^{[m]}_{a,b}=c^{[2(a+b+1)]}_{a,b}s^{-a-b-1}+O(s^{-a-b-\frac 32})\qquad a,b\geq 0.
\ee
Although these matrices are infinite,~\eqref{eq:expHinfinitematrices} makes sense to express the large-$s$ expansion, since, due to the structure of the middle matrix, only finitely many terms contribute to each power of $s^{-\frac 12}$.

Writing the expansion~\eqref{eq:expHinfinitematrices} more compactly as $\bigl(H(\nu_i+s,\nu_j+s)\bigr)_{i,j=1}^m\sim \mathcal V\cdot\mathcal  C(s)\cdot\mathcal V^\top$, we can use the Binet--Cauchy identity to obtain
\be
\det\bigl(H(\nu_i+s,\nu_j+s)\bigr)_{i,j=1}^m\sim\sum_{S,S'}\det(\mathcal V_{[m],S})\det(\mathcal C_{S,S'})\det(\mathcal V_{[m],S'})
\ee
where the sum runs over pairs of subsets $S,S'\subset\lbrace 1,2,3,\dots\rbrace$ of size $m$ and, in general, for sets $I,J$ we denote $\mathcal A_{I,J}$ the submatrix of $\mathcal A$ with row and column indices in $I$ and $J$ respectively, and we also use the notation $[m]=\lbrace 1,\dots,m\rbrace$.
Extracting the leading order we obtain (by simplifying the Vandermonde determinant)
\begin{align}
\nonumber
\det\bigl(H(\nu_i+s,\nu_j+s)\bigr)_{i,j=1}^m&=\biggl(\det\bigl(\nu_a^{b-1}\bigr)_{a,b=1}^m\biggr)^2\det\bigl(c_{a,b}^{[2(a+b+1)]}s^{-a-b-1}\bigr)_{a,b=0}^{m-1}\,+\, O(s^{-m^2-\frac 12})
\\
&=s^{-m^2}d_m\prod_{1\leq i<j\leq m}(\nu_i-\nu_j)^2 +O(s^{-m^2-\frac 12})
\end{align}
for a constant $d_m=\det\bigl(c_{a,b}^{[2(a+b+1)]}\bigr)_{a,b=0}^{m-1}$ independent of $f_i,g_i$.
As such, the constant can be computed by setting $f_i=0,g_i=0$ for all $i\geq 1$.
Therefore, let us denote $H_0=H\bigr|_{f_i=0,g_i=0}$, namely
\be
H_0(\lambda,\mu)=\frac{\lambda^{-\frac 14}\mu^{-\frac 14}}{\lambda^{\frac 12}+\mu^{\frac 12}}.
\ee
An application of Cauchy determinant formula gives
\be
\det\bigl(H_0(\nu_i+s,\nu_j+s)\bigr)
=\frac{\prod_{1\leq i<j\leq m}\bigl((\nu_i+s)^{\frac 12}-(\nu_j+s)^{\frac 12}\bigr)^2}{2^m\prod_{i=1}^m(\nu_i+s)\prod_{1\leq i<j\leq m}\bigl((\nu_i+s)^{\frac 12}+(\nu_j+s)^{\frac 12}\bigr)^2}
\sim \frac{\prod_{1\leq i<j\leq m}(\nu_i-\nu_j)^2}{2^{m(2m-1)}s^{m^2}},
\ee
as $s\to+\infty$, yielding $d_m=2^{-m(2m-1)}$ and completing the proof.
\end{proof}

By standard asymptotic properties of the Airy function and its derivative, the functions
\be 
F(z)=2\sqrt\pi \Ai(z)\exp\biggl(\frac 23 z^{\frac 32}\biggr),\quad
G(z)=2\sqrt\pi\Ai'(z)\exp\biggl(\frac 23 z^{\frac 32}\biggr)
\ee
satisfy the assumptions of Lemma~\ref{generallemma}, and in such case
\be
H(\lambda,\mu)=4\pi K^\Ai(\lambda,\mu)\exp\biggl(\frac 23\bigl( \lambda^{\frac 32}+\mu^{\frac 32}\bigr)\biggr)
\ee
which implies
\be
K_{X,T}^\Ai(\lambda,\mu)=\frac 1{4\pi} T^{-\frac 13}H(T^{-\frac 13}(\lambda+X),T^{-\frac 13}(\mu+X))\exp\biggl(-\frac 23T^{-\frac 12}\bigl((\lambda+X)^{\frac 32}+(\mu+X)^{\frac 32}\bigr)\biggr)\,.
\ee
By identifying $s$ with $XT^{-\frac 13}$ and replacing $\nu_i\mapsto T^{-\frac 13}\nu_i$ in Lemma~\ref{generallemma}, we obtain the following asymptotics for the correlation functions of the shifted and dilated Airy process:
\be
\label{eq:asympcorrAiry}
\det\bigl(K^\Ai_{X,T}(\nu_i,\nu_j)\bigr)_{i,j=1}^m \sim 
2^{-m(2m+1)}\pi^{-m} X^{-m^2}\prod_{1\leq i\leq m}\e^{-\frac 43T^{-\frac 12}(\nu_i+X)^{\frac 32}}\prod_{1\leq i<j\leq k}(\nu_i-\nu_j)^2
\ee
as $XT^{-\frac 13}\to+\infty$, uniformly for $T\geq T_0$.
Here we also use that the asymptotic relation~\eqref{eq:finalasymprel} is uniform in $c$ when $\nu_i\mapsto c\nu_i$ as long as $c>0$ is bounded above, as it follows from the proof of Lemma~\ref{generallemma}.
This proves~\eqref{eq:puresolition-asDet}; this argument (along with standard properties of the Airy function) also shows that the same asymptotics holds true as $|X|T^{-\frac 13}\to+\infty$, uniformly for $T\geq T_0$ and $|\arg X|<\delta$ for some $\delta>0$, which is needed to obtain~\eqref{eq:soliton-as}.

\medskip

Taking logarithms on both sides in \eqref{eq:factorization1XT}, we obtain
\be
\label{eq:factlog}
\log J_\s(X,T|\underline \nu)=\log \det\left(K_{X,T}^{\rm Ai}(\underline \nu,\underline \nu)\right)+\log\det_{L^2(\R)}\left(1-\mathcal M_{\sqrt\s}\widehat{\mathcal H}_{X,T}^{\Ai}\mathcal M_{\sqrt\s}\right),
\ee
and it follows from Lemma \ref{lemma:H} (for real $X$) that
\be
\det_{L^2(\R)}\left(1-\mathcal M_{\sqrt\s}\widehat{\mathcal H}_{X,T}^{\Ai}\mathcal M_{\sqrt\s}\right)=1+O(\e^{-cXT^{-\frac 13}})
\ee
as $X,T\to\infty$ uniformly for $X\geq MT^{\frac 13}$ and $T\geq T_0$.
This proves~\eqref{eq:asy1LogJ}.

\medskip

Taking the second logarithmic $X$-derivative in \eqref{eq:asy1LogJ}, we obtain
\be\label{eq:factlog2}V_\s(X,T|\underline \nu)=V_0(X,T|\underline \nu)+\partial_X^2\log\det_{L^2(\R)}\left(1-\mathcal M_{\sqrt\s}\widehat{\mathcal H}_{X,T}^{\Ai}\mathcal M_{\sqrt\s}\right).\ee
Since the estimate from Lemma \ref{lemma:H} holds uniformly for $|X|\geq MT^{\frac 13}$, $|\arg X|<\delta$, and $T\geq T_0$, we can use Cauchy's integral formula for the second derivative to obtain
\be
\pa_X^2\log\det_{L^2(\R)}\left(1-\mathcal M_{\sqrt\s}\widehat{\mathcal H}_{X,T}^{\Ai}\mathcal M_{\sqrt\s}\right)=O(\e^{-cXT^{-\frac 13}}),
\ee
and thus we prove \eqref{eq:asy1V}, so the proof of part (i) of Theorem \ref{thm:asy} is concluded.

\subsection{Left tail\texorpdfstring{: $X/T\to -\infty$}{}}\label{section:lefttail}

In this section we use the results of~\cite{CharlierClaeysRuzza}; {the latter rely on Assumption~\ref{assumption:strong}, which we assume throughout this section.}
We recall the transformation $x=-XT^{-\frac 12}$ and $t=T^{-\frac 12}$ between the cKdV variables of the present paper and the KdV variables of~\cite{CharlierClaeysRuzza}. 
For the ease of notations, we will denote $\wh\Theta(\lambda):=\wh\Theta_\s(\lambda;X,T|\emptyset)$ throughout this section for the function defined in \eqref{eq:whTheta}.
Let us now assume that, for an arbitrary~$T_0>0$ and for a sufficiently large~$K>0$, we have~$X\leq -K T$ and~$T\geq T_0$.
In this regime (in fact, in the larger regime $X\leq -K T^{\frac 12}$), the relevant asymptotics have been studied in~\cite{CharlierClaeysRuzza} via a RH analysis involving a series of transformations which we can condense in the relation
\begin{align}
\nonumber
\wh\Theta_\pm(|X|w)&=\begin{pmatrix}
1 & \wh p \\ 0 & 1
\end{pmatrix}
\e^{\frac{\i\pi}4\s_3}
|X|^{\frac 14\s_3}
\begin{pmatrix}
1 & -\i |X|^{\frac 32}T^{-\frac 12}g_1 \\ 0 & 1
\end{pmatrix}
R(w)
\\
\label{eq:ccr4}
&\quad
\times
(w-a)^{\frac 14\s_3}
G
\begin{pmatrix}
1 & 0 \\ \pm\e^{-|X|^{\frac 32}T^{-\frac 12}\phi_\pm(w)} & 1
\end{pmatrix}
\e^{|X|^{\frac 32}T^{-\frac 12}(g_\pm(w)-g_0)\s_3}
\e^{-\frac{\i\pi}4\s_3},
\end{align}
for $w$ sufficiently close to $0$.
Here, cf.~\cite[equations (4.15), (4.17)]{CharlierClaeysRuzza}, with principal branches for the roots,
\begin{align}
\nonumber
g(w)&=\int_a^w g'(s)\d s,&
g'(w)&=-(w-a)^{\frac 12}\biggl(1+\frac{T^{\frac 12}}{2\pi|X|^{\frac 12}}\int_{-\infty}^a\frac {\s'(|X|s)}{1-\sigma(|X|s)}\frac 1{\sqrt{a-s}}\frac{\d s}{s-w}\biggr),\\
\label{alotofstufffromccr4}
g_0&=\frac {T^{\frac 12}\log\bigl(1-\s(|X|a)\bigr)}{2|X|^{\frac 32}},&
\phi(w)&=2\bigl(g(w)-g_0\bigr)+\frac {T^{\frac 12}\log\bigl(1-\s(|X|w)\bigr)}{|X|^{\frac 32}}.
\end{align}
Moreover, $G$ is given in~\eqref{eq:G} and $\wh p=\wh p_\s(X,T|\emptyset)$.
The value of $a=a(X,T)$ is implicitly defined by the \emph{endpoint condition}, namely $g'_+(w)-g_-'(w)=O\bigl((w-a)^{\frac 12}\bigr)$ as $w\to a$ with $w<a$, cf.~\cite[equation (4.3)]{CharlierClaeysRuzza}.
For $X\leq -K T$ and $T\geq T_0$, $a$ is bounded away from zero and infinity, and by~\cite[Lemma~4.5]{CharlierClaeysRuzza},  
\be
\label{as:Rleft}
R(w)=I+O(|X|^{-\frac32}T^{\frac12}),\qquad \pa_w R(w)=O(|X|^{-\frac32}T^{\frac12}),
\ee 
uniformly in $w$, $T\geq T_0$, $X\leq -KT$.
Finally, the value of $g_1$ is given explicitly in~\cite[equation (4.16)]{CharlierClaeysRuzza} but it is not needed for our current purposes.

As explained in Section \ref{section:41}, in order to describe the behavior of $J_\s(X,T|\underline\nu), V_\s(X,T|\underline\nu)$ in this regime we will use equations \eqref{eq:factorizationasy} and \eqref{eq:decompositionasy}. Thus what is fundamental to understand is the behavior of $\wh L_{X,T}^\s(\nu_i,\nu_j)$. The following two lemmas will show how the kernel behaves on and off the diagonal. 

We are interested in values of $\wh\Theta$ at $\nu=|X|w$, hence we assume throughout this section that $w$ is real and small.
Therefore,~\eqref{eq:ccr4} implies
\be
\wh\Theta_\pm(|X|w)=
\begin{pmatrix}
1& \wh p+\frac{X^2}{T^{\frac 12}}g_1
\\ 0 & 1
\end{pmatrix}|X|^{\frac 14\s_3}O(1)
\begin{pmatrix}
\e^{|X|^{\frac 32}T^{-\frac 12}(g_\pm(w)-g_0)} & 0 \\ \mp\i\e^{|X|^{\frac 32}T^{-\frac 12}(g_\mp(w)-g_0)} & \e^{-|X|^{\frac 32}T^{-\frac 12}(g_\pm(w)-g_0)}
\end{pmatrix}
\ee
where we also use the identity $\phi_\pm=g_\pm-g_\mp$~\cite[equation below (4.17)]{CharlierClaeysRuzza}.
From~\cite[Proposition~4.7 and equation~(1.32)]{CharlierClaeysRuzza} we have
\be 
\label{estimateforlaterref}
\wh p+X^2T^{-\frac 12}g_1=O(|X|^{-1}T^{\frac 12}),
\ee
such that the previous relation implies
\be
\label{eq:resumingsec4}
\wh\Theta_\pm(|X|w)=
\begin{pmatrix}
1& O(|X|^{-1}T^{\frac 12})
\\ 0 & 1
\end{pmatrix}|X|^{\frac 14\s_3}O(1)
\begin{pmatrix}
\e^{|X|^{\frac 32}T^{-\frac 12}(g_\pm(w)-g_0)} & 0 \\ \mp\i\e^{|X|^{\frac 32}T^{-\frac 12}(g_\mp(w)-g_0)} & \e^{-|X|^{\frac 32}T^{-\frac 12}(g_\pm(w)-g_0)}
\end{pmatrix}.
\ee
Finally, we study the last factor, involving $g,g_0$.
By~\eqref{alotofstufffromccr4} and the Sokhotski--Plemelj formula, we have
\be
\label{eq:gprimeSP}
g_\pm'(w)=\mp\i\sqrt{a-w}\left(1+\frac{|X|^{-\frac 12}T^{\frac 12}}{2\pi}\mathrm{p.v.}\!\int_{-\infty}^a\frac {\s'(|X|s)}{1-\sigma(|X|s)}\frac 1{\sqrt{a-s}}\frac{\d s}{s-w}\pm\i\frac{T^{\frac 12}|X|^{-\frac 12}}{2\sqrt{a-w}}\frac {\s'(|X|w)}{1-\sigma(|X|w)}\right),
\ee
where $\mathrm{p.v.}\!\int$ is the principal value integral.
It follows that
\be
\Re g_\pm'(w)=\frac{T^{\frac 12}}{2|X|^{\frac 12}}\frac {\s'(|X|w)}{1-\sigma(|X|w)}
\ee
and so
\be
\label{eq:finalgsec4}
\Re \bigl(g_\pm(\zeta)-g_0)=\int_a^w \Re g'_\pm(s)\d s-g_0=
-\frac{T^{\frac 12}}{2|X|^{\frac 32}}
\log\bigl(1-\sigma(|X|w)\bigr).
\ee
Finally, let us set $w=|X|^{-1}\nu$, for a fixed $\nu$ and sufficiently large $|X|$.
It follows from the last estimates and~\eqref{eq:resumingsec4} that
\be
\wh\Theta_\pm(\nu)=\begin{pmatrix}
1& O(|X|^{-1}T^{\frac 12})
\\ 0 & 1
\end{pmatrix}|X|^{\frac 14\s_3}O(1)
\ee
because from~\eqref{eq:finalgsec4} we have
\be
\label{eq:Tbounded}
\biggl|\e^{|X|^{\frac 32}T^{-\frac 12}\bigl(g_\pm(\nu |X|^{-1})-g_0\bigr)}\biggr|=\frac 1{1-\s(\nu)}
\ee
which is bounded away from $0$ and $\infty$, uniformly in the regime under consideration.
In particular, by~\eqref{eq:whTheta},
\be
\label{eq:finalphiphiprimesec4}
\wh\varphi_\s(\nu;X,T|\emptyset)=O(|X|^{-\frac 14}),\qquad
T^{\frac 12}\pa_X\wh\varphi_\s(\nu;X,T|\emptyset)=O(|X|^{\frac 14}).
\ee

From~\eqref{eq:kernelXTTheta} and~\eqref{eq:finalphiphiprimesec4}, we immediately obtain boundedness of the kernel $\wh L_{X,T}^\s(\nu_1,\nu_2)$ for $\nu_1\neq \nu_2$.
Namely, we have proved the following lemma.

\begin{lemma}
\label{lemmadec1}
Let $T_0>0,\nu_1\not=\nu_2\in\R$ be fixed.
There exists $K>0$ such that
\be
\wh L_{X,T}^\s(\nu_1,\nu_2)=O(1)
\ee
uniformly for $X\leq-KT^{}$ and $T\geq T_0$.
\end{lemma}

On the other hand, we now show that on the diagonal, the kernel $\wh L_{X,T}^\s(\nu,\nu)$ grows.

\begin{lemma}
\label{lemmadec2}
Let $T_0>0,\nu\in\R$ be fixed.
We have
\be
\label{eq:onepoint}
\wh L_{X,T}^\s(\nu,\nu)\sim
\frac{|X|^{\frac 12}}{\pi T^{\frac 12}}\frac{1}{1-\sigma(\nu)},\ \ \mbox{as $\frac{X}{T\log^2 |X|}\to -\infty$,}
\ee
uniformly for $T\geq T_0$.
In particular, $\wh L_{X,T}^\s(\nu,\nu)^{-1}=O(|X|^{-\frac 12}T^{\frac 12})$.
\end{lemma}
\begin{proof}
We combine~\eqref{eq:kernelXTTheta} (in the confluent limit $\mu,\ll\to\nu$) with~\eqref{eq:ccr4} to get
\begin{align}
\nonumber
\wh L_{X,T}^\s(\nu,\nu)&=\frac 1{2\pi}\bigl(\wh\Theta_+^{-1}(\nu)\pa_\nu\wh\Theta_+(\nu)\bigr)_{2,1}
\\
\nonumber
&=\frac 1{2\pi\i}\biggl\lbrace\e^{-\chi(g_+(\frac\nu{|X|})-g_0)\s_3}\begin{pmatrix}1 & 0 \\ -\e^{-\chi\phi_+(\frac\nu{|X|})}&1\end{pmatrix}G^{-1}\bigl(-\tfrac\nu{|X|}-a\bigr)^{-\frac 14\s_3}R^{-1}(\tfrac\nu{|X|})
\\
\nonumber
&\qquad\qquad\qquad
\pa_\nu\biggl[R(\tfrac\nu{|X|})\bigl(-\tfrac\nu{|X|}-a\bigr)^{\frac 14\s_3}G\begin{pmatrix}1 & 0 \\ \e^{-\chi\phi_+(\frac\nu{|X|})}&1\end{pmatrix}\e^{\chi(g_+(\frac\nu{|X|})-g_0)\s_3}\biggr]\biggr\rbrace_{2,1}
\\
\nonumber
&=\frac 1{2\pi\i}\biggl\lbrace\begin{pmatrix}\e^{-\chi(g_+(\frac\nu{|X|})-g_0)} & 0 \\ -\e^{\chi(g_-(\frac\nu{|X|})-g_0)}&\e^{\chi(g_+(\frac\nu{|X|})-g_0)}\end{pmatrix}G^{-1}\bigl(-\tfrac\nu{|X|}-a\bigr)^{-\frac 14\s_3}R^{-1}(\tfrac\nu{|X|})
\\
&\qquad\qquad\qquad
\label{eq:lastL}
\pa_\nu\biggl[R(\tfrac\nu{|X|})\bigl(-\tfrac\nu{|X|}-a\bigr)^{\frac 14\s_3}G
\begin{pmatrix}\e^{\chi(g_+(\frac\nu{|X|})-g_0)} & 0 \\ \e^{\chi(g_-(\frac\nu{|X|})-g_0)}&\e^{-\chi(g_+(\frac\nu{|X|})-g_0)}\end{pmatrix}\biggr]\biggr\rbrace_{2,1}
\end{align}
where we denote $\chi:=|X|^{\frac 32}T^{-\frac 12}$ and in the last step we use the relation $\phi_+=g_+-g_-$, cf.~\cite[equation below~(4.17)]{CharlierClaeysRuzza}.
As we proved in~\eqref{eq:Tbounded}, $\e^{\chi(g_\pm(\frac\nu{|X|})-g_0)}$ is bounded away from $0,\infty$ and so the triangular matrices appearing in~\eqref{eq:lastL} are $O(1)$.
Therefore, when the derivative in $\nu$ acts in~\eqref{eq:lastL} it produces terms of order $O(|X|^{-1})$ when it acts on the first two factors, see also~\eqref{as:Rleft}, and another term when it acts on the triangular matrix, which provides the leading asymptotic contribution, yielding
\be
\wh L_{X,T}^\s(\nu,\nu)=\frac{\chi}{2\pi\i\,|X|}\e^{\chi\bigl(g_+(\frac\nu{|X|})+g_-(\frac\nu{|X|})-2g_0\bigr)}\bigl(g_-'(\tfrac\nu{|X|})-g_+'(\tfrac\nu{|X|})\bigr)+O(|X|^{-1}).
\ee 
By the construction of $g$, cf.~\cite[Section~4.2]{CharlierClaeysRuzza}, we have (recall that $F=1/(1-\sigma)$)
\be
\label{eqforlaterref}
\chi\bigl(g_+(w)+g_-(w)-2g_0\bigr)=\bigl(\log F\bigr)(|X|w),
\ee
hence we can rewrite the last expression as 
\be
\wh L_{X,T}^\s(\nu,\nu)=\frac{|X|^{\frac 12}}{2\pi\i\,T^{\frac 12}}F(\nu)\bigl(g_-'(\tfrac\nu{|X|})-g_+'(\tfrac\nu{|X|})\bigr)+O(|X|^{-1}).
\ee
Next, we use~\eqref{eq:gprimeSP}, a change of integration variable, and an integration by parts in order to get
\begin{align}
\nonumber
g_-'(w)-g_+'(w)&=2\i\sqrt{a-w}\biggl(1+\frac{T^{\frac 12}}{2\pi|X|^{\frac 12}}\mathrm{p.v.}\!\int_{-\infty}^{a|X|}(\log F)'(s)\frac 1{\sqrt{a-s|X|^{-1}}}\frac{\d s}{s-|X|w}\biggr)
\\
&=2\i\biggl(\sqrt{a-w}+\frac{T^{\frac 12}}{2\pi|X|^{\frac 12}}\int_{-\infty}^{a|X|}(\log F)''(s)\log\biggl|\frac{\sqrt{a-w}+\sqrt{a-s|X|^{-1}}}{\sqrt{a-w}-\sqrt{a-s|X|^{-1}}}\biggr|\d s\biggr).\label{eq:gpm}
\end{align}
Now, it is useful to recall the following asymptotic properties for $a=a(X,T)$ from~\cite[Proposition~4.1]{CharlierClaeysRuzza}:
\be
a=\frac{\bigl(\sqrt{1+y}-1\bigr)^2}y\biggr|_{y=\frac{\pi^2}{c_+^2}|X|/T}+O(|X|^{-\frac 32}T^{\frac 12})
\ee
uniformly in $-X/T\leq K$, $T\geq T_0$ (for any $K,T_0>0$). Hence,
\be 
a=1-\frac{2c_+T^{\frac 12}}{\pi|X|^{\frac 12}}+O(|X|^{-1}T),
\qquad\sqrt{a-\frac\nu{|X|}}
=1-\frac{c_+T^{\frac 12}}{\pi|X|^{\frac 12}}+O(|X|^{-1}T),
\ee
as $X/T\to-\infty$, $T\geq T_0$.
{Let us now show that the second term in \eqref{eq:gpm} is sub-dominant. To start with, notice that, in the same limit, and uniformly for $s\in(-\infty, a|X|)$,
\be
\label{eq:useasymplog}
\log\left|\frac{\sqrt{a-\nu|X|^{-1}}+\sqrt{a-s|X|^{-1}}}{\sqrt{a-\nu|X|^{-1}}-\sqrt{a-s|X|^{-1}}}\right|
=\log\frac{4a|X|}{|s-\nu|}+O(1)+O\left(\log(s|X|^{-1})\right)
.
\ee

It follows that 
\be
\wh L_{X,T}^\s(\nu,\nu)=\frac{|X|^{\frac 12}}{\pi\,T^{\frac 12}}F(\nu)\biggl[1-\frac{c_+T^{\frac 12}}{\pi|X|^{\frac 12}}
+O\left(|X|^{-\frac 12}T^{\frac 12}\log|X|\right)\biggr]\quad \mbox{as $XT^{-1}\log^{-2}|X|\to -\infty$}.
\ee

To achieve this bound, we used
\be\label{eq:rem1}
\int_{-\infty}^{a|X|}(\log F)''(s)\d s=\int_{-\infty}^{+\infty}(\log F)''(s)\d s+o(1)=c_++o(1)=O(1)
\ee
(note that, by Assumption~\ref{assumption:strong}, $\int_{-\infty}^{+\infty}(\log F)''(s)\d s=\displaystyle\lim_{s\to+\infty}(\log F)'(s)-\lim_{s\to-\infty}(\log F)'(s)=c_+$) and
\be\label{eq:rem2}
\int_{-\infty}^{a|X|}(\log F)''(s)\log\frac 4{|s-\nu|}\d s=\int_{-\infty}^{+\infty}(\log F)''(s)\log\frac 4{|s-\nu|}\d s+o(1)=O(1),
\ee
which, in turn, follow by the Lebesgue dominated convergence theorem, and a similar bound for $\int_{-\infty}^{a|X|}(\log F)''(s)\log |s|\d s$.}
Thus we obtain
\begin{align}
\nonumber
\wh L_{X,T}^\s(\nu,\nu)&=\frac{|X|^{\frac 12}}{\pi\,T^{\frac 12}}F(\nu)\biggl[1
+O\left(T^{\frac 12}|X|^{-\frac 12}\log|X|\right)
\biggr],
\end{align}
and the thesis follows.
\end{proof}

\begin{remark}\label{remark:extendas}
It is straightforward to adapt the above proof in order to obtain asymptotics in the full region $X/T\to -\infty$, slightly larger than the region $\frac{X}{T\log^2|X|}\to -\infty$. Note however that the error term will then no longer be small, and the asymptotic expression contains several terms, see~\eqref{eq:rem1} and~\eqref{eq:rem2}. For the sake of simplicity, we present the results only as $\frac{X}{T\log^2|X|}\to -\infty$.
\end{remark}

Using the two previous results, we can prove an important decorrelation property: since the matrix $\wh L_{X,T}^\s(\underline\nu,\underline\nu)$ is dominated by its diagonal, the $m$-point correlation function $\det\wh L_{X,T}^\s(\underline\nu,\underline\nu)$ decomposes at leading order into a product of one-point correlation functions.
Similarly, its second logarithmic derivative decomposes at leading order into a sum of rapidly oscillating terms.
\begin{proposition}\label{prop:decor}
Let $T_0>0$, $\underline\nu\in\mathbb R^m$. We have
\begin{align}
\label{eq:decor1}
\det\wh L_{X,T}^\s(\underline\nu,\underline\nu)&=\bigl(1+O(|X|^{-\frac 12}T^{\frac 12})\bigr)\prod_{i=1}^m\wh L_{X,T}^\s(\nu_i,\nu_i)
\\
\partial_X^2\log\det\wh L_{X,T}^\s(\underline\nu,\underline\nu)
\label{eq:decor2}
&=
\frac 1{\sqrt{|X|T}}\sum_{i=1}^m\cos\biggl(\frac{4|X|^{\frac 32}}{3T^{\frac 12}}(1+A_{X,T})-\frac{2|X|^{\frac 12}}{T^{\frac 12}}\nu_i(1+B_{X,T}(\nu_i))\biggr)+O(|X|^{-1})
\end{align}
uniformly for $T\geq T_0$ as $\frac{X}{T\log^2 |X|}\to -\infty$, where $A_{X,T}, B_{X,T}(\nu)$ converge to $0$ as $\frac{X}{T\log^2 |X|}\to -\infty$.
\end{proposition}
\begin{proof}
For the ease of notation, let us denote $\wh L=(\wh L_{ij})$ for the $m\times m$~matrix with entries $\wh L_{ij}:=\wh L_{X,T}^\s(\nu_i,\nu_j)$.
By Lemma~\ref{lemmadec1} and Lemma~\ref{lemmadec2} we have
\be
\label{eq:asdecorL}
\wh L_{ij}=\wh L_{ii}\bigl(\delta_{ij}+O(|X|^{-\frac 12}T^{\frac 12})\bigr),\qquad 1\leq i,j\leq m.
\ee
Taking determinants we get~\eqref{eq:decor1}.
Moreover, \eqref{eq:asdecorL} also implies
\be
\label{eq:asdecorLinv}
(\wh L^{-1})_{ij}=\frac 1{\wh L_{ii}}\bigl(\delta_{ij}+O(|X|^{-\frac 12}T^{\frac 12})\bigr)=
\frac{\delta_{ij}}{\wh L_{ii}}+O(|X|^{-1}T),
\qquad 1\leq i,j\leq m.
\ee
(In the second equality we use again Lemma~\ref{lemmadec2}.)
By a direct computation, we have
\be\label{eq:D2log}
\pa_X^2\log\det\wh L
=\sum_{i,j=1}^m(\pa_X^2\wh L_{ij})(\wh L^{-1})_{ji}\,-\sum_{i,j,k,l=1}^m(\pa_X\wh L_{ij})(\wh L^{-1})_{jk}(\pa_X\wh L_{kl})(\wh L^{-1})_{li}.
\ee
We have the relation $T^{\frac 12}\pa_X\wh L_{ij}={-}\wh\varphi_i\wh\varphi_j$, where we denote $\wh\varphi_i:=\wh\varphi_\s(\nu_i;X,T|\emptyset)$.
This is the analogue, for the full set of cKdV variables $X,T$, of the relation~\eqref{eq:dLs}.
As a consequence, $T^{\frac 12}\pa_X^2\wh L_{ij}=-(\pa_X\wh\varphi_i)\wh\varphi_j-\wh\varphi_i(\pa_X\wh\varphi_j)$.
Hence, by~\eqref{eq:finalphiphiprimesec4}, we have
\be
\label{estimatesLXT}
\pa_X\wh L_{ij}=O(|X|^{-\frac 12}T^{-\frac 12}),\qquad
\pa_X^2\wh L_{ij}=O(T^{-1}).
\ee
Therefore, by Lemma~\ref{lemmadec2},~\eqref{eq:asdecorLinv}, and~\eqref{estimatesLXT}, we have the estimates
\be
\sum_{i,j=1}^m(\pa_X^2\wh L_{ij})(\wh L^{-1})_{ji}=\sum_{i=1}^m\frac{\pa_X^2\wh L_{ii}}{\wh L_{ii}}+O(|X|^{-1}),
\quad
\sum_{i,j,k,l=1}^m(\pa_X\wh L_{ij})(\wh L^{-1})_{jk}(\pa_X\wh L_{kl})(\wh L^{-1})_{li}=O(|X|^{-2}).
\ee
In the last one we combined~\eqref{eq:asdecorLinv} and Lemma~\ref{lemmadec2} to get $(\wh L^{-1})_{ij}=O(|X|^{-\frac 12}T^{\frac 12})$.
Substituting these estimates into \eqref{eq:D2log}, we obtain
\be
\pa_X^2\log\det\wh L=-\frac{2}{\sqrt T}\sum_{i=1}^m
\frac{1}{\widehat L_{ii}}\wh\varphi_i\pa_X\wh\varphi_i+O(|X|^{-1})=
-\frac{1}{\pi T}\sum_{i=1}^m\frac{1}{\widehat L_{ii}}
\left(\wh\Theta(\nu_i){\rm E}_{12}\wh\Theta(\nu_i)^{-1}\right)_{2,2}+O(|X|^{-1}),
\ee
where the elementary unit matrix ${\rm E}_{12}$ is defined in~\eqref{eq:notationE12}, and where we used \eqref{eq:whTheta}.

Using \eqref{eq:ccr4}, we can write after straightforward computations
\be
\left(\wh\Theta(\nu){\rm E_{12}}\wh\Theta(\nu)^{-1}\right)_{2,2}=\mathbf v B \mathbf w, 
\ee
where, writing $\chi:=|X|^{\frac 32}T^{-\frac 12}$ as before,
\be
\mathbf v=\begin{pmatrix}0,1\end{pmatrix}R(\tfrac \nu{|X|}),\qquad \mathbf w=R(\tfrac \nu{|X|})^{-1}\begin{pmatrix}\i |X|^{-\frac 12}\widehat p+\i\chi g_1\\1
\end{pmatrix},
\ee
and,  
\be
B=\frac{\e^{2\chi(g_+(\frac{\nu}{|X|})-g_0)}}{\i}(\tfrac{\nu}{|X|}-a)^{\frac{\sigma_3}4}
G\begin{pmatrix}1 & 0\\
\e^{-\chi\phi_+(\frac{\nu}{|X|})}  & 1
\end{pmatrix}
E_{12}
\begin{pmatrix}1 & 0\\
-\e^{-\chi\phi_+(\frac{\nu}{|X|})}  & 1
\end{pmatrix}G^{-1}
(\tfrac{\nu}{|X|}-a)^{-\frac{\sigma_3}4}.
\ee
Using \eqref{as:Rleft} and the asymptotic estimate \eqref{estimateforlaterref}, we obtain
\be 
\mathbf v=\begin{pmatrix}O(\chi^{-1}),1+O(\chi^{-1})\end{pmatrix},\qquad \mathbf w=\begin{pmatrix}O(\chi^{-1})\\1+O(\chi^{-1})
\end{pmatrix},\qquad \mbox{as $X/T\to -\infty$,}
\ee
and in particular as $\frac{X}{T\log^2|X|}\to -\infty$.
By \eqref{alotofstufffromccr4}, we can simplify the expression for $B$ and obtain
\begin{align}
\nonumber
B&=-\i F(\nu)(\tfrac{\nu}{|X|}-a)^{\frac {\sigma_3}4}
G\begin{pmatrix}-1 & \e^{\chi\phi_+(\tfrac \nu{|X|})}\\
-\e^{-\chi\phi_+(\nu/|X|)}  & 1
\end{pmatrix}
G^{-1}
(\tfrac{\nu}{|X|}-a)^{-\frac {\sigma_3}4}\\
&=\frac{F(\nu)}{2}(\tfrac{\nu}{|X|}-a)^{\frac {\sigma_3}4}
\begin{pmatrix}\e^{\chi\phi_+(\tfrac \nu{|X|})}+\e^{-\chi\phi_+(\tfrac \nu{|X|})} & -2-\i\e^{\chi\phi_+(\tfrac \nu{|X|})}+\i \e^{-\chi\phi_+(\tfrac \nu{|X|})}\\
2-\i \e^{\chi\phi_+(\tfrac \nu{|X|})}+\i\e^{-\chi\phi_+(\tfrac \nu{|X|})}  & -\e^{\chi\phi_+(\tfrac \nu{|X|})}-\e^{-\chi\phi_+(\tfrac \nu{|X|})}
\end{pmatrix}
(\tfrac{\nu}{|X|}-a)^{-\frac {\sigma_3}4}.
\end{align}
Since $\phi_+(\tfrac {\nu}{|X|})$ is purely imaginary and $B$ is bounded and bounded away from $0$, we have
\be
B=\frac{1}{2}F(\nu)
\begin{pmatrix}O(1)&O(1)\\
O(1)&-\e^{\chi\phi_+(\tfrac{\nu}{|X|})}-\e^{-\chi\phi_+(\tfrac{\nu}{|X|})}
\end{pmatrix}.
\ee
Hence
\be
\left(\wh\Theta(\nu){\rm E_{12}}\wh\Theta(\nu)^{-1}\right)_{2,2}=\mathbf v B \mathbf w=-F(\nu)\cos\left({\chi|\phi_+(\tfrac \nu{|X|})|}\right)+O(\chi^{-1}). 
\ee
We finally obtain
\begin{align}
\nonumber
\partial_X^2\log\det\widehat L&=\frac{1}{\pi T}\sum_{i=1}^m\frac{F(\nu_i)}{\hat L_{ii}}\cos\left({\chi|\phi_+(\tfrac{\nu_i}{|X|})|}\right)+O(|X|^{-1})\\
&=
\frac{1}{\sqrt{|X|T}}\sum_{i=1}^m\cos\left({|X|^{\frac 32}T^{-\frac 12}|\phi_+(\tfrac{\nu_i}{|X|})|}\right)
+O(|X|^{-1}),
\end{align}
where we used Lemma \ref{lemmadec2}.
It remains to compute
\be
|\phi_+(\tfrac{\nu_i}{|X|})|=\int_{\nu_i/|X|}^a|(g_+-g_-)'(s)|\d s.
\ee
For this, we recall~\eqref{eq:gpm} and the estimates below that equation; in the same way as in the proof of Lemma \ref{lemmadec2}, we then obtain 
\be
|\phi_+(0)|\to \frac{4}{3},\qquad 
|\phi_+(\tfrac{\nu_i}{|X|})|-|\phi_+(0)|\sim -2\frac{\nu_i}{|X|},
\ee
as $\frac{X}{T\log^2 |X|}\to -\infty$.
The argument of the cosine is thus equal to
\be 
\frac{4|X|^{\frac 32}}{3\sqrt{T}}(1+A_{X,T})-\frac{2\sqrt{|X|}}{\sqrt T}\nu_i(1+B_{X,T}(\nu_i)),
\ee
with $A_{X,T}\to 0$, $B_{X,T}(\nu_i)\to 0$ as $\frac{X}{T\log^2 |X|}\to -\infty$, and the result follows.
\end{proof}

Combining the above result with \eqref{eq:factorizationasy}, and then substituting the asymptotics from Lemma~\ref{lemmadec2}, we complete the proof of part (ii) of Theorem \ref{thm:asy}.

With some more effort, we could obtain asymptotics in the slightly bigger asymptotic region where $X/T\to -\infty$, as already mentioned in Remark \ref{remark:extendas}.

\subsection{Intermediate regimes\texorpdfstring{: $-KT\leq X\leq MT^{\frac 13}$}{}}\label{section:PII}

We will now discuss the asymptotic behavior as $T\to\infty$ of various relevant quantities in the intermediate regimes where $-KT\leq X\leq MT^{\frac 13}$ for sufficiently large constants $K,M>0$. The asymptotics for the J\'anossy density $J_\sigma(X,T|\underline\nu)$ and the cKdV solution $V_\sigma(X,T|\underline\nu)$ become, unfortunately, rather involved and implicit. In order to understand the mechanisms behind these asymptotics, an interesting and relevant object  to consider, is the kernel $\widehat L_{X,T}^{\s}(\nu_1,\nu_2)$.
Indeed, in view of the factorization \eqref{eq:factorizationasy}, determinants of this kernel describe the effect of the points $\nu_1,\ldots, \nu_m$ on the J\'anossy densities $J_\s(X,T|\underline\nu)$, and the second logarithmic $X$-derivative of such determinants describe the effect of the points $\nu_1,\ldots, \nu_m$ on the cKdV solutions $V_\s(X,T|\underline\nu)$.
Recall that the kernel $\widehat L_{X,T}^{\s}(\nu_1,\nu_2)$ is expressed in terms of the RH solution $\wh\Theta$ through \eqref{eq:kernelXTTheta}.
We distinguish three further asymptotic regimes.

\paragraph{Left-intermediate regime: $-KT\leq X\leq -K'T^{\frac 12}$ for any $K,K'>0$.}
The asymptotic analysis of the RH problem for $\wh\Theta$ has been carried through in \cite{CharlierClaeysRuzza} and is very similar to the one utilized for the left tail. However, there is an important difference in that the decorrelation property from Proposition \ref{prop:decor} no longer holds. For that reason, even if we could obtain asymptotics for $\wh L_{X,T}^\s(\nu_1,\nu_2)$, the explicit asymptotic behavior of the determinants $\det\wh L_{X,T}^\s(\underline\nu,\underline\nu)$ and their logarithmic derivatives becomes cumbersome for $m>1$.

\paragraph{Right-intermediate regime: $-MT^{\frac 13}\leq X\leq MT^{\frac 13}$ for any $M>0$.}
In this case, it was proved in \cite[Theorem 1.15]{CafassoClaeysRuzza} that
there exists a (sufficiently large) $T_0>0$ such that for all $M>0$ we have
\begin{align}
\label{eq:asyempty2LogJ}
\log J_\s(X,T|\emptyset)&=\log F_{\mathrm{TW}}(XT^{-\frac 13})+O(T^{-\frac 16}),\\
\label{eq:asyempty2V}
V_\s(X,T|\emptyset)&=-T^{-\frac 23}y_{\mathrm{HM}}\bigl(XT^{-\frac 13}\bigr)^2+O(T^{-1}),
\end{align}
uniformly for $|X|\leq MT^{\frac 13}$ and $T\geq T_0$,
{where $y_{\mathrm{HM}}$ is the Hastings--McLeod solution of the Painlev\'e II equation,
and $F_{\mathrm{TW}}$ is the Tracy--Widom distribution (see also Example \ref{examples:zeroHeaviside}).}

The asymptotic analysis of $\wh\Theta$ has also been obtained in~\cite{CafassoClaeysRuzza}, and it implies that the leading order asymptotics of $\wh L_{X,T}^\s(\lambda,\mu)$ are determined by the soft-to-hard edge transition kernel $L_s^{1_{(0,\infty)}}$ from Example~\ref{example:Heaviside2}, as we prove next.

\begin{proposition}\label{prop:intermediate1}
Let $M>0$.
As $T\to\infty$, we have uniformly for $-MT^{\frac 13}\leq X\leq MT^{\frac 13}$, and uniformly for $\nu_1,\nu_2$ in compact subsets of the real line that
\be 
\wh L_{X,T}^\s(\nu_1,\nu_2)=T^{-\frac13} L_{s=XT^{-\frac13}}^{1_{(0,+\infty)}}(T^{-\frac13}\nu_1,T^{-\frac13}\nu_2)+O(T^{-\frac 23}),
\ee
and this error estimate continues to hold upon differentiating an arbitrary number of times with respect to $\nu_1$ and $\nu_2$.
\end{proposition}
\begin{proof}
The proof relies on the RH analysis performed in~\cite[Section 6]{CafassoClaeysRuzza}: the result is that, for every fixed $\nu\in\C$, we have the factorization
\be
\label{eq:factpII}
\wh \Theta_\s(\nu;X,T|\emptyset)=\begin{pmatrix}
	1 & \wh p_\s \\ 0 & 1
\end{pmatrix}
\e^{\frac{\i\pi}4\s_3}T^{\frac 1{12}\s_3} S(\nu T^{-\frac 13})\Psi_{1_{(0,+\infty)}}(\nu T^{-\frac 13};XT^{-\frac 13}|\emptyset)\begin{pmatrix}
			1&a(\nu T^{-\frac 13})\\
			0&1
\end{pmatrix}\e^{-\frac{\i\pi}4\s_3},
\ee
where $\wh p_\s=\wh p_\s(X,T|\emptyset)$, $a$ has an explicit expression which is not needed for our purposes (cf.~\cite[equation~(6.10)]{CafassoClaeysRuzza}).
The matrix $S$ should be interpreted as an error term: it satisfies a {\it small-norm} RH problem, which means that, provided $|w|<1$, $S(w)=I+O(T^{-\frac 13})$ and $\pa_wS(w)=O(T^{-\frac 13})$, uniformly for $-MT^{\frac 13}\leq X\leq MT^{\frac 13}$ (for any $M>0$). In particular, for every fixed $\nu_1,\nu_2 \in \R $ and $T$ sufficiently large we also have
\be
S(\nu_2 T^{-\frac 13})^{-1}S(\nu_1 T^{-\frac 13})=I+O(T^{-\frac 23}(\nu_1-\nu_2)).
\ee
Combining this with~\eqref{eq:kernelXTTheta} and~\eqref{eq:factpII}, we get
\begin{align}
\nonumber
\wh L_{X,T}^\s(\nu_1,\nu_2)&= \frac{\bigl(\Psi_{1_{(0,+\infty)}}(\nu_2T^{-\frac 13};XT^{-\frac 13}|\emptyset)^{-1} S(\nu_2T^{-\frac 13})^{-1}S(\nu_1 T^{-\frac 13}) \Psi_{1_{(0,+\infty)}} (\nu_1T^{-\frac 13};XT^{-\frac 13}|\emptyset)\bigr)_{2,1}}{2\pi\i(\nu_1-\nu_2)}
\\
\nonumber
&= T^{-\frac 13}L_{s=XT^{-\frac 13}}^{1_{(0,+\infty)}}(T^{-\frac 13}\nu_1,T^{-\frac 13}\nu_2)
\\
&\quad\qquad+\frac{\bigl(\Psi_{1_{(0,+\infty)}} (\nu_2T^{-\frac13};XT^{-\frac13}|\emptyset)^{-1}O(T^{-\frac23})\Psi_{1_{(0,+\infty)}} (\nu_1T^{-\frac13};XT^{-\frac13}|\emptyset)\bigr)_{2,1}}{2\pi\i},
\end{align}
where we also use the identity 
\be
L_{s}^{1_{(0,+\infty)}}\left(T^{-\frac13}\nu_1,T^{-\frac13}\nu_2\right)= \frac{\left(\Psi_{1_{(0,+\infty)}} (\nu_2T^{-\frac13};s|\emptyset)^{-1}\Psi_{1_{(0,+\infty)}} (\nu_1T^{-\frac13};s|\emptyset)\right)_{2,1}}{2\pi\i T^{-\frac13} (\nu_1-\nu_2)},
\ee
which is a special case of~\eqref{eq:Lkernel}.
The last term is $O(T^{-\frac 23})$ since it is a combination of the entries of the first column of $\Psi_{1_{(0,+\infty)}}(\nu_i T^{-\frac 13})$, $i=1,2$, which are both entire.
The above identities extend to $\nu_1,\nu_2$ in compact subsets of the complex plane, hence we can apply Cauchy's formula to differentiate, without affecting the error term.
\end{proof}

A first, crucial, obstruction for obtaining explicit asymptotics for the J\'anossy densities $J_\s(X,T|\underline\nu)$ lies in the fact that the kernel $ L_{s=XT^{-\frac 13}}^{1_{(0,+\infty)}}$ is itself a transcendental object, which we cannot evaluate explicitly.
However, we can proceed in the hope of describing $J_\s(X,T|\underline\nu)$ asymptotically in terms of the $\s$-independent quantity $ L_{s=XT^{-\frac 13}}^{1_{(0,+\infty)}}$.
For $m=1$, we immediately find by \eqref{eq:factorizationasy} that
\be
J_\s(X,T|\nu)\sim T^{-\frac 13} L_{s=XT^{-\frac 13}}^{1_{(0,+\infty)}}(0,0)J_\s(X,T|\emptyset),
\ee
where the asymptotics for $J_\s(X,T|\emptyset)$ are given by \eqref{eq:asyempty2LogJ}.
For $m>1$, we can estimate $J_\s(X,T|\underline\nu)$ as follows.
\begin{proposition}\label{prop:PII2}
Let $M>0$, $\underline\nu=(\nu_1,\ldots, \nu_m)\in\mathbb R^m$, and $\nu_i\neq \nu_j$ for $i\neq j$. As $T\to\infty$, we have uniformly for $-MT^{\frac 13}\leq X\leq MT^{\frac 13}$ that
\be
J_\s(X,T|\underline\nu)= C_m \,T^{-\frac{m^2}{3}}\,F_{\rm TW}(XT^{-\frac 13})\prod_{1\leq j<k\leq m}(\nu_j-\nu_k)^2+o\left(T^{-\frac{m^2}{3}}\right),
\ee
for some constant $C_m\geq 0$ possibly depending on $m$ but not on $\sigma, \underline\nu, X,T$.
\end{proposition}
\begin{proof}
Let us abbreviate
$L=\wh L_{X,T}^\s$ and $\widetilde L=L_{s=XT^{-\frac 13}}^{1_{(0,+\infty)}}$. By \eqref{eq:factorizationasy}, we have
\be
J_\s(X,T|\underline\nu)=\det\left(L(\nu_i,\nu_j)\right)_{i,j=1}^m
J_{\sigma}(X,T|\emptyset),
\ee
hence by \eqref{eq:asyempty2LogJ}, it remains to prove that
\be
\label{eq:VdMtoprove}
\det\left(L(\nu_i,\nu_j)\right)_{i,j=1}^m= C_m
T^{-\frac{m^2}{3}}\prod_{1\leq j<k\leq m}(\nu_j-\nu_k)^2+o\left(T^{-\frac{m^2}{3}}\right),
\ee
in the relevant limit.
Since $L(\cdot , \cdot)$ is entire in its variables, we have
\be
L(\nu_i,\nu_j)=\sum_{a,b\geq 0}L^{(a,b)}(0,0)\frac{\nu_i^a\nu_j^b}{a!b!}
\ee
hence
\begin{align}
\nonumber
&\begin{pmatrix}
L(\nu_1,\nu_1) & L(\nu_1,\nu_2) & \cdots & L(\nu_1,\nu_m)\\
L(\nu_2,\nu_1) & L(\nu_2,\nu_2) & \cdots & L(\nu_2,\nu_m)\\
\vdots & \vdots & \ddots & \vdots\\
L(\nu_m,\nu_1) & L(\nu_m,\nu_2) & \cdots & L(\nu_m,\nu_m)
\end{pmatrix}\\
&\label{eq:CB}\quad=
\begin{pmatrix}
1& \nu_1 & \frac{\nu_1^2}{2!} &\cdots\\
1& \nu_2 & \frac{\nu_2^2}{2!} &\cdots\\
\vdots & \vdots & \vdots &\cdots\\
1& \nu_m & \frac{\nu_m^2}{2!} &\cdots\\
\end{pmatrix}
\begin{pmatrix}
L^{(0,0)}(0,0) & L^{(1,0)}(0,0) & L^{(2,0)}(0,0) &\cdots\\
L^{(0,1)}(0,0) & L^{(1,1)}(0,0) & L^{(2,1)}(0,0) &\cdots\\
L^{(0,2)}(0,0) & L^{(1,2)}(0,0) & L^{(2,2)}(0,0) &\cdots\\ 
\vdots & \vdots & \vdots & \ddots
\end{pmatrix}
\begin{pmatrix}
1& 1 & \cdots &1\\
\nu_1& \nu_2 &\cdots &\nu_m\\
\frac{\nu_1^2}{2!}& \frac{\nu_2^2}{2!} &\cdots & \frac{\nu_m^2}{2!}\\
\vdots & \vdots & \cdots &\vdots
\end{pmatrix}.
\end{align}
Then we use Proposition \ref{prop:intermediate1} to obtain
\be
L^{(i-1,j-1)}(0,0)\sim T^{-\frac{i+j-1}{3}}\widetilde L^{(i-1,j-1)}(0,0),\ \ \mbox{as $T\to\infty$, $|XT^{-\frac 13}|\leq M$.}
\ee
Expanding \eqref{eq:CB} by the Binet--Cauchy identity, we immediately see that the leading order as $T\to\infty$ is given by
\begin{multline}
\det_{i,j=1}^m\left(\frac{\nu_i^{j-1}}{(j-1)!}\right)\det_{i,j=1}^m\left(T^{-\frac{i+j-1}{3}}L^{(i-1,j-1)}(0,0)\right)\det_{i,j=1}^m\left(\frac{\nu_i^{j-1}}{(j-1)!}\right)\\=T^{-\frac{m^2}{3}}\prod_{k=1}^{m-1}\frac{1}{k!^2}\det_{i,j=1}^m\left(\widetilde L^{(i-1,j-1)}(0,0)\right)\det_{i,j=1}^m\left(\nu_i^{j-1}\right)^2.
\end{multline}
In the latter, we recognize the Vandermonde determinant, and the result follows.
\end{proof}
\begin{remark}
The appearance of the squared Vandermonde determinant $\prod_{1\leq j<k\leq m}(\nu_j-\nu_k)^2$ in the above result is remarkable, and it expresses a repulsion as two points $\nu_j, \nu_k$ are close to each other.
As can be seen from the proof, this repulsion, with quadratic vanishing, is a consequence of the determinantal structure of the correlation functions.
\end{remark}
\paragraph{Middle-intermediate regime: $-K'T^{\frac 12}\leq X\leq -MT^{\frac 13}$ for some $K',M>0$.}

Here, the asymptotic analysis of the RH problem for $\wh\Theta$ has also been completed in \cite{CafassoClaeysRuzza}, but the asymptotics are implicit and described in terms of the solution of an integro-differential generalization of the fifth Painlev\'e equation.

\section*{Acknowledgments}
The work of the authors is supported by the Fonds de la Recherche Scientifique-FNRS under EOS project O013018F.
TC was also supported by FNRS Research Project T002823F. 
GR was also supported by the FCT grant 2022.07810.CEECIND.
The authors are grateful to Percy Deift for an insightful correspondence about the KdV Trace Formula.

\section*{Data availability statement.}
Data sharing not applicable to this article as no datasets were generated or analysed during the current study.

\section*{Declarations.}
The authors have no competing interests to declare that are relevant to the content of this article.

\end{document}